\PassOptionsToPackage{sort&compress}{natbib}  

\documentclass[preprint]{elsarticle}

\makeatletter
\def\ps@pprintTitle{%
 \let\@oddhead\@empty
 \let\@evenhead\@empty
 \def\@oddfoot{\footnotesize
   This paper is published in:
   \textit{Computer Physics Communications},
   Vol.~305 (2024), Article~109351.\hfill}%
 \let\@evenfoot\@oddfoot}
\makeatother

\usepackage[bookmarksopen=false]{hyperref}

\usepackage{graphicx}
\usepackage{amsmath}
\usepackage{amsfonts}
\usepackage{amssymb}
\usepackage{amsthm} 
\usepackage{ifthen} 
\usepackage{comment} 
\usepackage{algorithm} 
\usepackage{algpseudocode} 
\usepackage{mathtools} 
\usepackage{xcolor} 
\usepackage[normalem]{ulem} 

\algrenewcommand\algorithmicrequire{\textbf{Input:}}
\algrenewcommand\algorithmicensure{\textbf{Output:}}

\journal{Computer Physics Communications}

\newlength{\mm}
\newlength{\cm}

\newcommand{\etal}{\textit{et al.}}
\newcommand{\ie}{\textit{i.e.}}
\newcommand{\eg}{\textit{e.g.}}

\newcommand{\mbigskip}{\vspace{0.69mm}}

\providecommand\bcdot{\boldsymbol{\cdot}}

\newcommand{\ud}{\mathrm{d}}
\newcommand\mnewcommand[1]{%
\let#1\relax \newcommand#1 }

\newcommand{\timeintparam}{\kappa}
\newcommand{\SpMM}{\texttt{SpMM}}
\newcommand{\SpMV}{\texttt{SpMV}}

\newtheorem{theorem}{Theorem}
\newtheorem{remark}{Remark}

\mnewcommand{\Nx}{N_{\x}}
\mnewcommand{\Ny}{N_{\y}}
\mnewcommand{\Nz}{N_{\z}}
\mnewcommand{\Nm}{N_m}
\mnewcommand{\Ns}{N_s}
\mnewcommand{\complconj}[1]{#1^{*}}
\mnewcommand{\mvbrack}[1]{\left[ #1 \right]}
\mnewcommand{\step}{\Delta}
\mnewcommand{\dt}{\step t}
\mnewcommand{\traspose}{^{T}}
\mnewcommand{\avgtime}[1]{\left< #1 \right>}
\mnewcommand{\avg}[1]{\overline{#1}}
\mnewcommand{\mcdot}{\bcdot}
\mnewcommand{\mnabla}{\nabla}

\mnewcommand{\real}{\mathbb{R}}
\mnewcommand{\complex}{\mathbb{C}}
\mnewcommand{\imag}{i\real}

\mnewcommand{\realpart}{\mathsf{Re}}
\mnewcommand{\imagpart}{\mathsf{Im}}

\mnewcommand{\foutrans}[1]{\hat{#1}}
\mnewcommand{\modetrans}[1]{\tilde{#1}}

\mnewcommand{\timeintparam}{\kappa}

\mnewcommand{\x}{x}
\mnewcommand{\y}{y}
\mnewcommand{\z}{z}

\newcommand{\xvec}{\isvector{x}}

\newcommand{\gammax}{\gamma_{\x}}

\newcommand{\Ra}{Ra}
\newcommand{\PR}{Pr}

\mnewcommand{\vels}{u_s}
\mnewcommand{\uvel}{u_1}
\mnewcommand{\vvel}{u_2}
\mnewcommand{\wvel}{u_3}

\mnewcommand{\facevel}{[\velh]_\nface}

\mnewcommand{\flux}{f}
\mnewcommand{\Dx}{\Delta x}
\mnewcommand{\Dy}{\Delta \y}

\mnewcommand{\isvector}[1]{\boldsymbol{#1}}
\mnewcommand{\istensor}[1]{\mathsf{#1}}

\mnewcommand{\va}{\isvector{a}}
\mnewcommand{\vb}{\isvector{b}}
\mnewcommand{\vc}{\isvector{c}}
\mnewcommand{\vd}{\isvector{d}}
\mnewcommand{\sca}{\phi}
\mnewcommand{\scafield}{\isvector{\sca}}
\mnewcommand{\scafieldc}{\scafield_{c}}

\mnewcommand{\vel}{\isvector{u}}
\mnewcommand{\velv}{\isvector{v}}
\mnewcommand{\velw}{\isvector{w}}
\mnewcommand{\velz}{\isvector{z}}
\mnewcommand{\normal}{\isvector{n}}

\mnewcommand{\veluc}{\isvector{u}_1}
\mnewcommand{\velvc}{\isvector{u}_2}
\mnewcommand{\velwc}{\isvector{u}_3}

\mnewcommand{\basis}{\isvector{w}}

\mnewcommand{\tensor}{\istensor{T}}
\mnewcommand{\Identity}{\istensor{I}}

\mnewcommand{\bodyforce}{\isvector{f}}
\mnewcommand{\velh}{\vel_s}
\mnewcommand{\velhc}{\vel_c}
\mnewcommand{\velhcCB}[1]{\ifthenelse{\equal{#1}{}}{\velhc^{\ominus}}{\velhc^{#1,\ominus}}}
\mnewcommand{\velhcCBFree}[1]{\ifthenelse{\equal{#1}{}}{\velhc^{\oplus}}{\velhc^{#1,\oplus}}}
\mnewcommand{\velhcHF}[1]{\ifthenelse{\equal{#1}{}}{\velhc^{>}}{\velhc^{#1,>}}}
\mnewcommand{\velhcLF}[1]{\ifthenelse{\equal{#1}{}}{\velhc^{<}}{\velhc^{#1,<}}}
\mnewcommand{\velhsHF}[1]{\ifthenelse{\equal{#1}{}}{\velh^{>}}{\velh^{#1,>}}}
\mnewcommand{\velhsLF}[1]{\ifthenelse{\equal{#1}{}}{\velh^{<}}{\velh^{#1,<}}}
\mnewcommand{\vortd}{\vort_v}
\mnewcommand{\presh}{\isvector{p}_c}
\mnewcommand{\pseudopresh}{{\isvector{\tilde{p}}}_c}
\mnewcommand{\bodyforceh}{\bodyforce_c}
\mnewcommand{\diver}{\mnabla \cdot}
\mnewcommand{\lapl}{\mnabla^2}
\mnewcommand{\grad}{\mnabla}
\mnewcommand{\Ru}[1]{\isvector{R} \left( #1 \right)}
\mnewcommand{\Rud}[1]{\ifthenelse{\equal{#1}{}}{\mathsfbi{R}}{\istensor{R} \left( #1 \right)}}
\mnewcommand{\vecnull}{\isvector{0}}
\mnewcommand{\vecnulls}{\vecnull_{s}}
\mnewcommand{\vecnullg}{\vecnull_{h}}
\mnewcommand{\vecnullc}{\vecnull_{c}}
\mnewcommand{\vecone}{\isvector{1}}
\mnewcommand{\veconec}{\vecone_{c}}
\mnewcommand{\vecones}{\vecone_{s}}
\mnewcommand{\veconeg}{\vecone_{h}}
\mnewcommand{\veconetresc}{\vecone_{3c}}

\mnewcommand{\velhg}{\vel_h}
\mnewcommand{\preshg}{\isvector{p}_h}

\mnewcommand{\dim}{3}
\mnewcommand{\Ndim}{d}

\mnewcommand{\ud}{d}
\mnewcommand{\vort}{\mathbf{w}}
\mnewcommand{\rot}[1]{\mnabla \times #1}
\mnewcommand{\selfinnerprod}[1]{\innerprod{#1}{#1}}
\mnewcommand{\innerprod}[2]{< #1 , #2 >}
\mnewcommand{\convective}[2]{C \left( #1 , #2 \right)}
\mnewcommand{\intvol}[1]{\int_{\Omega} #1 \ud \Omega}
\mnewcommand{\intsurf}[1]{\int_{\partial \Omega} #1 \ud S}

\mnewcommand{\nvc}{k}
\mnewcommand{\nedge}{v}
\mnewcommand{\axis}{i}
\mnewcommand{\nface}{f}
\mnewcommand{\cp}{c}
\mnewcommand{\cpA}{{c1}}
\mnewcommand{\cpB}{{c2}}
\mnewcommand{\Fedge}[1]{F_e ( #1 )}
\mnewcommand{\Fcell}[1]{F_f ( #1 )}
\mnewcommand{\Fvolume}[1]{F_c ( #1 )}

\mnewcommand{\mathsfbi}[1]{\mathsf{#1}}

\mnewcommand{\conv}{\mathsfbi{C}\left( \velh \right)}
\mnewcommand{\convg}{\mathsfbi{C}\left( \velhg \right)}
\mnewcommand{\convc}{\mathsfbi{C}_{c} \left( \velh \right)}
\mnewcommand{\convu}{\mathsfbi{C}_{u} \left( \velh \right)}
\mnewcommand{\convtraspose}{\mathsfbi{C}\traspose\left( \velh \right)}
\mnewcommand{\convctraspose}{\mathsfbi{C}_{c}\traspose\left( \velh \right)}
\mnewcommand{\convgtraspose}{\mathsfbi{C}\traspose\left( \velhg \right)}
\mnewcommand{\convarg}[1]{\mathsfbi{C}\left( #1 \right)}
\mnewcommand{\convcarg}[1]{\mathsfbi{C}_{c}\left( #1 \right)}
\mnewcommand{\convargtraspose}[1]{\mathsfbi{C}\traspose\left( #1 \right)}
\mnewcommand{\convutraspose}{\mathsfbi{C}_{u} \traspose\left( \velh \right)}
\mnewcommand{\convmat}{\mathsfbi{C}}
\mnewcommand{\convUP}{\convmat^{\UP}_{c} \left( \velh \right)}

\DeclareRobustCommand{\spreg}[2]{\ifthenelse{\equal{#2}{}}{\convmat_{#1}}{\convmat_{#1}\left( #2 \right)}}
\mnewcommand{\velauxc}{\velv_c}
\mnewcommand{\velauxWc}{\velw_c}
\mnewcommand{\filter}{\mathsfbi{G}_\epsilon}

\mnewcommand{\velhcmode}[1]{\ifthenelse{\equal{#1}{+}}{\velhcHF{}}{\ifthenelse{\equal{#1}{-}}{\velhcLF{}}{\ifthenelse{\equal{#1}{=}}{\velhc}{}}}}
\mnewcommand{\velhsmode}[1]{\ifthenelse{\equal{#1}{+}}{\velhsHF{}}{\ifthenelse{\equal{#1}{-}}{\velhsLF{}}{\ifthenelse{\equal{#1}{=}}{\velh}{}}}}

\mnewcommand{\velhctype}[2]{\ifthenelse{\equal{#2}{f}}{\overline{\velhcmode{#1}}}{\ifthenelse{\equal{#2}{r}}{\left(\velhcmode{#1}\right)^\prime}{\ifthenelse{\equal{#2}{v}}{\velhcmode{#1}}{}}}}
\mnewcommand{\velhstype}[2]{\ifthenelse{\equal{#2}{f}}{\overline{\velhsmode{#1}}}{\ifthenelse{\equal{#2}{r}}{\left(\velhsmode{#1}\right)^\prime}{\ifthenelse{\equal{#2}{v}}{\velhsmode{#1}}{}}}}

\mnewcommand{\triadic}[6]{\left(\velhctype{#1}{#2}\right)\traspose \convarg{\velhstype{#3}{#4}} \velhctype{#5}{#6}}

\mnewcommand{\diff}{\mathsfbi{D}}
\mnewcommand{\diffg}{\mathsfbi{D}}
\mnewcommand{\diffc}{\diff_{\cp}}
\mnewcommand{\diffu}{\diff_{u}}
\mnewcommand{\dive}{\mathsfbi{M}}
\mnewcommand{\divec}{\mathsfbi{M}_c}
\mnewcommand{\diveg}{\mathsfbi{M}_h}
\mnewcommand{\divescaf}{\dive_{sf}}
\mnewcommand{\divevecf}{\dive_{vf}}
\mnewcommand{\graddc}{\mathsfbi{G}_c}
\mnewcommand{\gradd}{\mathsfbi{G}}
\mnewcommand{\graddg}{\mathsfbi{G}_h}
\mnewcommand{\lapld}{\mathsfbi{L}}
\mnewcommand{\lapldc}{\lapld_{c}}
\mnewcommand{\vcvects}{\mathsfbi{\Omega}_s}
\mnewcommand{\pseudovcvects}{\tilde{\mathsfbi{\Omega}}_s}
\mnewcommand{\vcvectg}{\mathsfbi{\Omega}_h}
\mnewcommand{\vcvectc}{\mathsfbi{\Omega}_c}
\mnewcommand{\vcvectv}{\mathsfbi{\Omega}_v}
\mnewcommand{\vcvectvc}{\mathsfbi{\Omega}}
\mnewcommand{\vcvect}{\mathsfbi{\Omega}}
\mnewcommand{\nullmat}{\mathsfbi{0}}
\mnewcommand{\normd}[1]{|| #1 ||}
\mnewcommand{\rotd}{\mathsfbi{R}}
\mnewcommand{\graddscaf}{\gradd_{sf}}
\mnewcommand{\graddvecf}{\gradd_{vf}}
\mnewcommand{\graddx}{\gradd_{\x}}
\mnewcommand{\graddy}{\gradd_{\y}}
\mnewcommand{\graddz}{\gradd_{\z}}
\mnewcommand{\graddd}[1]{\ifthenelse{\equal{#1}{1}}{\graddx}{\ifthenelse{\equal{#1}{2}}{\graddy}{\ifthenelse{\equal{#1}{3}}{\graddz}{\gradd_{x_i}}}}}
\mnewcommand{\graddprod}[2]{\graddd{#1}\traspose \vcvect \graddd{#2}}

\mnewcommand{\fluxh}[2]{T_{#1}\left({#2}\right)}

\mnewcommand{\twoD}{two-dimensional}
\mnewcommand{\threeD}{three-dimensional}
\mnewcommand{\TwoD}{Two-dimensional}
\mnewcommand{\ThreeD}{Three-dimensional}
\mnewcommand{\biD}{\twoD~}
\mnewcommand{\triD}{\threeD~}
\mnewcommand{\BiD}{\TwoD~}
\mnewcommand{\TriD}{\ThreeD~}
\mnewcommand{\biandtriD}{two- and \threeD~}
\mnewcommand{\BiandtriD}{Two- and \threeD~}
\mnewcommand{\Dim}[1]{\ifthenelse{\equal{#1}{2}}{\twoD}{\ifthenelse{\equal{#1}{3}}{\threeD}{KK}}}

\mnewcommand{\inttypeoftext}{\mathsf}
\mnewcommand{\kinener}{\inttypeoftext{E}}
\mnewcommand{\enstrophy}{\inttypeoftext{\Omega}}
\mnewcommand{\helicity}{\inttypeoftext{H}}
\mnewcommand{\vorthelicity}{\helicity_{\vort}}
\mnewcommand{\palinstrophy}{\inttypeoftext{P}}

\mnewcommand{\helicityd}{\helicity_c}
\mnewcommand{\enstrophyd}{\enstrophy_c}

\mnewcommand{\vvlength}{m}
\mnewcommand{\pvlength}{n}
\mnewcommand{\evlength}{e}

\mnewcommand{\Isc}{\Gamma}
\mnewcommand{\Ics}{\Isc\traspose}
 
\mnewcommand{\Scs}{\Gamma_{c \rightarrow s}}
\mnewcommand{\Ssc}{\Gamma_{s \rightarrow c}}

\mnewcommand{\Sscal}{\Pi_{c \rightarrow s}}
\mnewcommand{\Sthreescal}{\Pi}
\mnewcommand{\NormalVect}[1]{\ifthenelse{\equal{#1}{}}{\istensor{N}_{s}}{\istensor{N}_{s,#1}}}

\mnewcommand{\Correction}{\istensor{P}}
\mnewcommand{\PseudoCorrection}{\tilde{\Correction}}
\mnewcommand{\kernel}[1]{Ker \left( #1 \right)}
\mnewcommand{\sCorrection}{\Correction_s}
\mnewcommand{\cCorrection}{\PseudoCorrection_c}

\mnewcommand{\Ivs}{\Psi}
\mnewcommand{\Isv}{\Psi\traspose}

\mnewcommand{\OIsc}{\Upsilon}
\mnewcommand{\OIcs}{\Pi}

\mnewcommand{\order}{o}
\mnewcommand{\coarsemesh}{i}
\mnewcommand{\sizevc}{{V}}
\mnewcommand{\error}{\epsilon}

\mnewcommand{\convH}{\text{(Conv)}_{\helicityd}}
\mnewcommand{\diffH}{\text{(Diff)}_{\helicityd}}
\mnewcommand{\presH}{\text{(Pres)}_{\helicityd}}

\newcommand{\Dt}{\dt}
\newcommand{\Ampl}{\istensor{A}}
\newcommand{\Us}{\istensor{U}_{s}}
\newcommand{\As}{\istensor{A}_{s}}
\newcommand{\Fs}{\istensor{F}_{s}}
\newcommand{\diffv}{\isvector{\alpha}_s}
\newcommand{\diffm}{\istensor{\Lambda}_s}
\newcommand{\pdiffm}{\tilde{\istensor{\Lambda}}_s}
\newcommand{\Tcs}{\istensor{T}_{cs}}
\newcommand{\Tsc}{\istensor{T}_{sc}}
\newcommand{\DX}{\istensor{\Delta}_s}

\newcommand{\diag}{\mathop{\mathrm{diag}}}
\newcommand{\sign}{\mathop{\mathrm{sign}}}
\newcommand{\trace}[1]{\mathop{\mathrm{tr}}(#1)}
\newcommand{\Hadprod}{\circ}
\newcommand{\offdiag}{\mathrm{off}}

\newcommand{\SP}{\mathrm{SP}}
\newcommand{\UP}{\mathrm{UP}}

\newcommand{\canobas}{\isvector{e}}

\newcommand{\genmat}{\istensor{A}}
\newcommand{\genmatcoeff}{a}
\newcommand{\genmatB}{\istensor{B}}
\newcommand{\genmatBcoeff}{b}

\newcommand{\scaUPblend}{\Psi}
\newcommand{\UPblend}{\isvector{\scaUPblend}_s}

\newcommand{\Stensor}{\istensor{S}}
\newcommand{\press}{p}

\newcommand{\sepcomma}{\hspace{1mm},\hspace{1mm}}

\newcommand{\diffFlux}{\istensor{A}^\convmat}
\newcommand{\diffFluxoff}{\istensor{A}^{\convmat,\offdiag}}

\mnewcommand{\vapf}{\lambda}
\mnewcommand{\angle}{\varphi}
\mnewcommand{\vapfn}{-e^{-i\angle}}
\mnewcommand{\dtn}{\widetilde{\dt}}
\mnewcommand{\vapg}{\lambda^{\pm}}
\mnewcommand{\Kopt}{K_{opt}}
\mnewcommand{\Topt}{T_{opt}}
\mnewcommand{\CFLAB}{CFL+AB2}
\mnewcommand{\NewMeth}{AlgEigCD+\timeintparam 1L2}
\mnewcommand{\AlgEigCD}{{\it AlgEigCD}}
\mnewcommand{\EigenCD}{{\it EigenCD}}
\mnewcommand{\parNiki}{{q}}

\begin{document}

\begin{frontmatter}

\title{An efficient eigenvalue bounding method: CFL condition revisited}

\tnotetext[mytitlenote]{An efficient eigenvalue bounding method: CFL condition revisited}

\author[cttc]{F.X.~Trias}
\ead{francesc.xavier.trias@upc.edu}

\author[cttc]{X.~\'{A}lvarez-Farr\'{e}}
\ead{xavier.alvarez.farre@upc.edu}

\author[cttc,tf]{A.~Alsalti-Baldellou}
\ead{adel.alsalti@upc.edu}

\author[kiam]{A.~Gorobets}
\ead{andrey.gorobets@gmail.com}

\author[cttc]{A.~Oliva}
\ead{asensio.oliva@upc.edu}

\address[cttc]{Heat and Mass Transfer Technological Center, Technical University of Catalonia \\ 
ESEIAAT, c/ Colom 11, 08222 Terrassa (Barcelona), Spain}
\address[tf]{Termo Fluids S.L. Carrer de Mag\'{i} Colet 8, 08204 Sabadell (Barcelona), Spain}
\address[kiam]{Keldysh Institute of Applied Mathematics of Russian Academy of Sciences\\
4A, Miusskaya Sq., Moscow, Russia 125047}

\begin{abstract}
  {Direct and large-eddy simulations of turbulence are often
    solved using explicit temporal schemes. However, this imposes very
    small time-steps because the eigenvalues of the (linearized)
    dynamical system, re-scaled by the time-step, must lie inside the
    stability region. In practice, fast and accurate estimations of
    the spectral radii of both the discrete convective and diffusive
    terms are therefore needed. This is virtually always done using
    the so-called CFL condition. On the other hand, the large
    heterogeneity and complexity of modern supercomputing systems are
    nowadays hindering the efficient cross-platform portability of CFD
    codes. In this regard, our {\it leitmotiv} reads: {\it relying on
      a minimal set of (algebraic) kernels is crucial for code
      portability and maintenance!} In this context, this work focuses
    on the computation of eigenbounds for the above-mentioned
    convective and diffusive matrices which are needed to determine
    the time-step {\it \`{a} la} CFL. To do so, a new inexpensive
    method, that does not require to re-construct these time-dependent
    matrices, is proposed and tested. It just relies on a
    sparse-matrix vector product where only vectors change on
    time. Hence, both implementation in existing codes and
    cross-platform portability are straightforward. The effectiveness
    and robustness of the method are demonstrated for different test
    cases on both structured Cartesian and unstructured
    meshes. Finally, the method is combined with a self-adaptive
    temporal scheme, leading to
    {significantly} larger time-steps compared with other more
    conventional CFL-based approaches.}
\end{abstract}


\end{frontmatter}


\section{Introduction}

\label{intro}

We consider the simulation of turbulent incompressible flows of
Newtonian fluids. Under these assumptions, the governing equations
read
\begin{equation}
\label{NS_eqs}
\partial_t \vel + ( \vel \cdot \nabla ) \vel = 2 \rho^{-1} \diver \left( \mu \Stensor ( \vel ) \right) - \grad \press, \hspace{6.69mm} \diver \vel = 0 ,
\end{equation}
\noindent where $\vel ( \xvec , t )$ and $p ( \xvec , t)$ denote the
velocity and {kinematic} pressure fields, and $\Stensor = 1/2
( \grad \vel + \grad \vel\traspose )$ is the rate-of-strain tensor.
The density, $\rho$, is constant whereas the dynamic viscosity, $\mu
( \xvec, t)$, may depend on space and time. Notice that for
(spatially) constant viscosity, and recalling the vector calculus
identity $\diver ( \grad \vel )\traspose = \grad ( \diver \vel )$, the
diffusive term simplifies to $2 \nu \diver \Stensor ( \vel )
= \nu \diver \grad \vel + \nu \diver ( \grad \vel )\traspose
= \nu \lapl \vel$ where $\nu = \mu / \rho$ is the kinematic viscosity.

\mbigskip

{Then, these equations have to be discretized both in space and
time. These two problems are usually addressed separately despite many
numerical effects result from their entanglement: numerical
dispersion~\cite{HU96,RUATRI19-DIS}, artificial
dissipation~\cite{SAN13,CAP17} and stability analysis~\cite{SUN21} are
examples thereof. {In this regard,}
this work is mainly focused on finding cheap and accurate eigenbounds
for the linear stability analysis of explicit time-integration
schemes. {Nevertheless,} the spatial discretization must also be
considered {since it eventually determines the coefficients of the
matrices and, therefore, their spectral properties (location of
eigenvalues in the complex plane) which can have a significant impact
on the overall method.} Therefore, {for the sake of completeness},
both spatial and time-integration methods are reviewed in the next
paragraphs.}


\subsection{Spatial discretization of Navier--Stokes equations}

\label{intro_spacediscr}

The basic physical properties of the Navier--Stokes (NS)
equations~(\ref{NS_eqs}) are deduced from the symmetries of the
differential operators (see Ref.\cite{FRI95}, for example). In a
discrete setting, such operator symmetries must be retained to
preserve the analogous (invariant) properties of the continuous
equations. This idea goes back to the mid-'60s with the pioneering
works by Lilly~\cite{LIL65}, Arakawa~\cite{ARA66}, and
Bryan~\cite{BRY66}. They basically showed that numerical schemes that
preserve certain integral properties of the continuous equations can
also eliminate non-linear instabilities. It is remarkable the
derivation {\it \`{a} la} finite-volume method (FVM) by
Bryan~\cite{BRY66} (at that time the concept of FVM did not exist
yet!)  who showed that an unweighted cell-to-face interpolation for
the advected variable is necessary to preserve total kinetic
energy. This eliminated the non-linear instability problems described
by Phillips~\cite{PHI59} a few years before.

\mbigskip

Around $20$~years later, the increasing capacity of high-performance
computing (HPC) systems enabled to carry out the first scale-resolving
simulations of turbulent channel flows either (wall-resolved)
large-eddy simulation (LES)~\cite{MOI82} or direct numerical
simulation (DNS)~\cite{KIM87}. These simulations made use of a
Fourier--Chebyshev pseudospectral method using the $3/2$~dealiasing
rule for the non-linear terms. However, the applicability of this type
of methods is restricted to simple canonical flows, \eg~homogeneous
isotropic turbulence, channel flow, boundary layers, etc. Therefore,
mesh-based methods such as FVM (also finite difference and finite
element) are necessary to tackle more complex configurations. On the
other hand, turbulence phenomenon results from an intricate dynamical
process: the non-linear convective term generates a continuous
transfer of kinetic energy from large to small scales, up to the point
where viscous dissipation becomes strong enough to counterbalance the
nonlinear production. Numerically, schemes that produce artificial
dissipation may dramatically affect this subtle balance of forces
inter-played at the smallest scales. In this context,
Morinishi~\etal~\cite{MOR98} reviewed the existing conservative
second-order finite-difference schemes for structured meshes, and
presented a fourth-order one which is fully-conservative only on
uniform meshes, whereas it is ``nearly conservative'' on non-uniform
ones. Later, Vasilyev~\cite{VAS00} generalized these schemes for
non-uniform meshes using a mapping technique. However, they do not
simultaneously preserve momentum and kinetic energy: it depends on the
form chosen for the convective term. On the other hand, Verstappen and
Veldman~\cite{VER98,VER03} proposed to preserve the symmetries of the
underlying operators at discrete level: the convective operator is
represented by a skew-symmetric matrix, whereas the diffusive operator
is a symmetric positive semi-definite matrix. In this way,
the {semi-discrete} system is unconditionally
stable. Regarding the accuracy, they recalled the work by Manteuffel
and White~\cite{MAN86}: given stability, a second-order local
truncation error is a sufficient but not a necessary condition for a
second-order global truncation error, \ie~the actual error in the
numerical solution. In this regard, they showed that both second- and
fourth-order {\it symmetry-preserving discretizations} yields second-
and fourth-order accurate solution although the local truncation error
is indeed first-order on non-uniform meshes~\cite{VER03}.

\mbigskip

All the above-explained conservative discretizations are restricted to
staggered Cartesian grids. The way for reliable DNS and LES
simulations on unstructured grids goes back to the works by
Perot~\cite{PER00} and Zhang~\etal~\cite{ZHA02}, where both collocated
and staggered formulations were proposed. The staggered one
discretizes the NS equations in rotational form, which implies to
compute the vorticity at the edges of the mesh. An easier alternative
consisting in the combination of collocated discrete operators was
also proposed in the same work~\cite{PER00} and subsequently explored
by other researchers~\cite{HIC05}. Mahesh~\etal~\cite{MAH04} also
developed both staggered and collocated conservative schemes for LES
in complex domains. A paper reviewing the existing conservative
methods on unstructured grids was published by
Perot~\cite{PER11}. Nevertheless, due to its simplicity to discretize
momentum equations on unstructured grids nowadays collocated
discretizations are the solution adopted by most of the
general-purpose CFD codes such as
ANSYS-FLUENT\textsuperscript{\textregistered} or
OpenFOAM\textsuperscript{\textregistered}. However, there exist
intrinsic errors in the conservation of mass and kinetic energy due to
the improper pressure-velocity
coupling~\cite{MOR98,HAM04,TRI08-JCP}. These errors can eventually
have severe implications for DNS and LES simulations of turbulent
flows: they can introduce far too much artificial dissipation,
significantly affecting the dynamics of the small scales and even
overwhelming the dissipation introduced by subgrid-scale LES
models~\cite{KOM17,KOMTRI20-JCP}. In this context, a {\it
symmetry-preserving discretization} method for collocated unstructured
grids was proposed in Ref.~\cite{TRI08-JCP}: it exactly preserves the
symmetries of the underlying differential operators while introducing
a minimal amount of artificial dissipation due to the
pressure-velocity coupling. This was clearly shown in
Ref.~\cite{KOMTRI20-JCP}, where this {\it symmetry-preserving
discretization} was implemented in
OpenFOAM\textsuperscript{\textregistered} and compared with the
standard version. To complete this subsection, it is worth mentioning
other related methods or concepts, such as the Keller box
schemes~\cite{PER07}, the mimetic methods~\cite{LIP14}, or the
discrete calculus methods~\cite{ROB11,TON14}. All these approaches
share the idea of preserving the mathematical structure of the space,
naturally producing physics-compatible numerical
methods~\cite{KOR14}. In conclusion, the accuracy of DNS and LES
simulations of turbulence is not automatically improved by simply
increasing the order of accuracy of the numerical schemes but also
retaining the symmetries of the continuous equations. Recent works in
this vein can be found, for instance, in
Refs.~\cite{COP19a,COP19b,VEL19,VALTRI19-EPLS,VEL21,ZHA22}.


\subsection{Time-integration methods}

\label{intro_timeint}

Starting from the above-mentioned channel flow simulations by
Kim~\etal~\cite{MOI82,KIM87} in the mid-'80s, DNS and LES simulations
of incompressible flows have usually been carried out by means of a
fractional step method together with an explicit or semi-implicit
time-integration method for momentum. In these initial works, they
used a second-order explicit Adams--Bashforth (AB2) scheme for the
non-linear convective terms, whereas the viscous terms were advanced
in time by a second-order implicit Crank--Nicolson scheme. Later, a
three-step third-order semi-implicit Runge--Kutta (RK3) scheme was
proposed by Le and Moin~\cite{MOI91}. The non-linearity was again
treated explicitly, and the Poisson equation was solved only at the
final step to project the velocity vector onto a divergence-free
space. Slightly different variants of the method can be found
in~\cite{VER96,SPA91,RAI91}, for instance. Shortly, the RK3 algorithm
has three steps and, therefore, it requires three times more
operations (except for the Poisson equation that is solved only in the
last step) to advance to a new time level. To compensate this,
stability analysis leads to significantly larger CFL
numbers~\cite{CFL28}, that is, larger time-steps compared with the AB2
method; therefore, RK3 method is often the favorite
option. However, Verstappen and
Veldman~\cite{VER97,VER03} showed that a minor modification with
respect the original AB2 method may lead to similar computational cost
as the RK3 method proposed in~\cite{MOI91} without affecting the
accuracy. {Later, this idea was extended in~\cite{TRI08-JCP2}
where a self-adaptive second-order scheme was proposed.} Despite the
most widespread schemes are the above-mentioned AB2
and RK3
methods, other schemes (or variants) have also been used in the
context of the numerical simulation of the unsteady Navier--Stokes
equations. In~\cite{JOT03}, a six-step fourth-order implicit
Runge--Kutta method in conjunction with several non-linear solvers was
presented. A semi-implicit third-order accurate in time Runge--Kutta
scheme was proposed in~\cite{NIK06}; it is based on the original
three-step RK3 scheme with one additional sub-step to achieve a
higher-order of accuracy. In the context of collocated spatial
formulations, a third-order-explicit Gear-based scheme was proposed to
mitigate the unwanted spatial
oscillations~\cite{FIS09}. {Finally, symplectic~\cite{SAN13} and
pseudo-symplectic~\cite{CAP17} RK schemes have been proposed to get
rid of the artificial dissipation introduced by the time-integration
of momentum equation.} {More recent works in a similar vein can
be found, for instance, in Refs.~\cite{PAG23,KIN23,RIC23}.}
Nevertheless, all explicit time-integration schemes require an
estimation of the (maximum) eigenvalues of the dynamical system to
compute an upper bound for the time-step. In the CFD community, this
is virtually always done via the so-called {\it CFL} condition
originally proposed in the seminal paper by R.~{\it C}ourant, K.~{\it
F}riedrichs, and H.~{\it L}ewy~\cite{CFL28}. This is revised in detail
later in this paper.


\subsection{Motivation and scope of the present work} 

\label{intro_scope}

{In the last decades, CFD has become a standard design tool in many
fields, such as the automotive, aeronautical, and wind power
industries. The driving force behind is the above-explained
development of numerical techniques in conjunction with the progress
of HPC systems. However, we can say that progress is nowadays hindered
by its legacy from the {90-2000s}. The reasons are
two-fold. Firstly, the design of digital processors constantly evolves
to overcome limitations and bottlenecks. The formerly compute-bound
nature of processors led to compute-centric programming languages and
simulation codes. However, raw computing power grows at a (much)
faster pace than the speed of memory access, turning around the
problem. Increasingly complex memory hierarchies are found nowadays in
computing systems, and optimizing traditional applications for these
systems is cumbersome. Moreover, new parallel programming languages
emerged to target modern hardware (\eg~OpenMP, CUDA,
OpenCL, {HIP}), and porting algorithms and applications has become
restrictive. Secondly, legacy numerical methods chosen to solve
(quasi)steady problems using RANS models are inappropriate for more
accurate (and expensive) techniques such as LES or DNS. We aim to
interlace these two pillars with the final goal of enabling LES and
DNS of industrial applications to be efficiently carried out on modern
HPC systems while keeping codes easy to port, optimize, and
maintain. In this regard, the fully-conservative discretization for
collocated unstructured grids proposed in Ref.~\cite{TRI08-JCP} is
adopted: it constitutes a very robust approach that can be easily
implemented in existing codes such as
OpenFOAM\textsuperscript{\textregistered}~\cite{KOMTRI20-JCP}.}

\mbigskip

{On the other hand, breaking the interdependency between algorithms
and their computational implementation allows casting calculations
into a minimalist set of universal kernels. There is an increasing
interest towards the development of more abstract implementations. For
instance, the PyFR framework~\cite{WIT14} is mostly based on matrix
multiplications and point-wise operations. Another example is the
Kokkos programming model~\cite{EDW14}, which includes computation
abstractions for frequently used parallel computing patterns and data
structures. Namely, implementing an algorithm in terms of Kokkos
entities allows mapping the algorithm onto multiple architectures. In
this regard, in previous works~\cite{ALVTRI17HPC2,ALVTRI20-HPC2} we
showed that virtually all calculations in a typical CFD algorithm for
LES or DNS of incompressible turbulent flows can be boiled down to
three basic linear algebra subroutines: sparse matrix-vector product
(\SpMV), linear combination of vectors and dot product. From now on,
we refer to implementation models based on algebraic subroutines as
algebraic or algebra-based. In this implementation approach, the
kernel code shrinks to dozens of lines; the portability becomes
natural, and maintaining multiple implementations takes minor
effort. Besides, standard libraries optimized for particular
architectures (\eg~cuSPARSE~\cite{NVIDIA2013},
clSPARSE~\cite{GRE-IWOCL2016}) can be linked in addition to
specialized in-house implementations. Nevertheless, the algebraic
approach imposes restrictions and challenges that must be addressed,
such as the inherent low arithmetic intensity of the~\SpMV, the
reformulation of flux limiters~\cite{VALALVTRI21-FL}, or the efficient
computation of eigenbounds to determine the time-step. This work
focuses on the latter problem and aims to answer the following
research question:
\noindent {\it Can we avoid to explicitly construct both convective,
  and diffusive matrices while still being able to compute proper
  eigenbounds in an inexpensive manner?}
{Hereafter, the idea of avoiding the (re)construction of matrices
should be understood in a broad sense: \ie~any sort of specific kernel
that needs to (re)compute the coefficients or any other similar in
essence operation must be avoided.} Preliminary stages of this work
were presented in the 8th ECCOMAS
conference~\cite{TRI22ECCOMAS-SymPres}.}

\mbigskip

The rest of the paper is arranged as follows. In the next section, the
fully conservative spatial discretization of the incompressible NS
equations~(\ref{NS_eqs}) for collocated unstructured grids is briefly
described following the same algebraic notation as in the original
paper~\cite{TRI08-JCP}. Then, in Section~\ref{EigenCD_vs_CFL}, the
idea of replacing the CFL condition by more accurate bounds of the
eigenvalues of both convective and diffusive operators is outlined. It
is essentially based on applying the Gershgorin circle theorem to the
corresponding matrices; therefore, it is suitable for any type of mesh
and discretization method. However, for the reasons explained above we
seek a method that is entirely composed of very basic algebraic
kernels. In this regard, a new method named \AlgEigCD~is presented in
Section~\ref{AlgEig_CD_method}. The key points are the fact that no
new matrix has to be re-computed every time-step (lower memory
footprint) and that, in practice, only relies on an~\SpMV. Therefore,
{implementation and cross-platform portability} of the method are
straightforward. Moreover, apart from this computational benefits,
results presented in Section~\ref{results} show that
the \AlgEigCD~method is also able to provide much better eigenbounds
than a classical CFL condition. Benefits became even more evident on
unstructured grids. Finally, relevant results are summarized and
conclusions are given.

\section{Symmetry-preserving spatial discretization of NS equations}

\label{SymPres}

An energy-preserving discretization of the NS equations~(\ref{NS_eqs})
on collocated unstructured grids is briefly described in this
section. Otherwise stated, we follow the same matrix-vector notation
as in the original paper~\cite{TRI08-JCP}. The spatial discretization
exactly preserves the symmetries of the underlying differential
operators: the convective operator is represented by a skew-symmetric
matrix and the diffusive operator by a symmetric negative
semi-definite matrix. Shortly, the temporal evolution of the
collocated velocity vector, $\velhc \in \real^{\dim \pvlength}$, is
governed by the following algebraic system
\begin{align}
\label{discr_mom}
\vcvect \frac{\ud \velhc}{\ud t} + \conv \velhc &= \diff \velhc - \vcvect \graddc \presh , \\
\label{discr_mass}
\dive \velh &= \vecnullc ,
\end{align}
\noindent where $\presh \in \real^{\pvlength}$ is the cell-centered
pressure and $\pvlength$ is the number of control volumes. The
sub-indices $c$ and $s$ are used to refer whether the variables are
cell-centered or staggered at the faces. The collocated velocity,
$\velhc \in \real^{\dim \pvlength}$, is arranged as a column vector
containing the three spatial velocity components as $\velhc = ( \veluc
, \velvc , \velwc )\traspose$ where $\vel_i = ( [\vel_i]_1 ,
[\vel_i]_2 , \dots , [\vel_i]_{\pvlength}) \in \real^{\pvlength}$ are
vectors containing the velocity components corresponding to the
$x_i$-spatial direction. The staggered velocity vector $\velh = (
[ \vels ]_1 , [ \vels ]_2 , \dots , [ \vels ]_{\vvlength}
)\traspose \in
\real^{\vvlength}$, which is needed to computed the convective term,
$\conv$, results from the projection of a staggered predictor
velocity, $\velh^p$. The matrices
$\vcvect \in \real^{\dim \pvlength \times \dim \pvlength}$, $\conv \in
\real^{\dim \pvlength \times \dim \pvlength}$, $\diff \in \real^{\dim
  \pvlength \times \dim \pvlength}$ are square block diagonal matrices
given by
\begin{equation}
\label{block_diag}
\vcvect = \Identity_{\dim} \otimes \vcvectc , \hphantom{kkk}
\conv = \Identity_{\dim} \otimes \convc , \hphantom{kkk}
\diff = \Identity_{\dim} \otimes \diffc ,
\end{equation}
\noindent where $\Identity_{\dim} \in \real^{\dim \times \dim}$ is the
identity matrix, $\vcvectc \in \real^{\pvlength \times \pvlength}$ is
a diagonal matrix containing the sizes of the cell-centered control
volumes and, $\convc \in \real^{\pvlength \times \pvlength}$ and
$\diffc\in \real^{\pvlength \times \pvlength}$ are the collocated
convective and diffusive operators, respectively. Finally, $\graddc
\in \real^{\dim \pvlength \times \pvlength}$ represents the discrete
gradient operator whereas the matrix $\dive \in \real^{\pvlength
  \times \vvlength}$ is the face-to-cell discrete divergence operator.

\mbigskip

The spatially discrete momentum equation~(\ref{discr_mom}) is
discretized in time using the fully-explicit second-order
$\timeintparam$1L2 scheme (see Ref.~\cite{TRI08-JCP2}
and~\ref{SAT_summary}, for details) for both convection and diffusion
whereas the pressure-velocity coupling is solved using a fractional
step method.


\subsection{Constructing the discrete operators}

This subsection briefly revise the construction of all the discrete
operators needed to solve the NS
equations. {The
constraints} imposed by the {operator} (skew-)symmetries strongly
simplifies ``the discretization problem'' to a set of five basic
discrete operators {(see Ref.~\cite{TRI08-JCP}, for
details)}. Namely,
\begin{equation}
\label{basic_set_operators}
\{ \vcvectc, \vcvects, \NormalVect{}, \dive , \Sscal\} .
\end{equation}
\noindent The first three correspond to basic geometrical information
of the mesh: namely, the diagonal matrices containing the
cell-centered and staggered control volumes, $\vcvectc$ and
$\vcvects$, and the matrix containing the face normal vector,
$\NormalVect{} \equiv ( \NormalVect{1}, \NormalVect{2}, \NormalVect{3}
) \in \real^{\vvlength \times \dim \vvlength}$ where
$\NormalVect{\axis} \in \real^{\vvlength \times \vvlength}$ are
diagonal matrices containing the $x_i$-spatial components of the face
normal vectors, $\normal_{\nface}$. The staggered control volumes,
$\vcvects$, are given by
\begin{equation}
\label{stag_mesh}
[ \vcvects ]_{f,f} \equiv A_{\nface} \delta_{\nface} ,
\end{equation}
\noindent where $A_{\nface}$ is the area of the face $\nface$ and
$\delta_{\nface} =
| \normal_{\nface} \cdot \protect \overrightarrow{\cpA \cpB} |$ is the
projected distance between adjacent cell centers (see
Figure~\ref{mesh}). In this way, the sum of volumes is exactly
preserved $\trace{\vcvects}=\trace{\vcvect}= d \trace{\vcvectc}$
($d=2$ for 2D and $d=3$ for 3D) regardless of the mesh quality and the
location of the cell centers.

\mbigskip

Then, the face-to-cell discrete (integrated) divergence operator,
$\dive$, is defined as follows
\begin{equation}
[ \dive \velh ]_{\nvc} = \sum_{\nface \in \Fcell{\nvc}} \facevel A_{\nface} ,
\end{equation}
\noindent where $\Fcell{\nvc}$ is the set of faces bordering the cell
$\nvc$. Finally, $\Sscal \in \real^{\vvlength \times \pvlength}$ is an
unweighted cell-to-face scalar field interpolation,
\begin{equation}
\sca_{\nface} \approx [ \Sscal \scafieldc ]_{\nface} = \frac{\sca_{\cpA} + \sca_{\cpB}}{2} ,
\end{equation}
\noindent where $\cpA$ and $\cpB$ are the cells adjacent to the face
$\nface$ (see Figure~\ref{mesh}, left). This is needed to construct
the skew-symmetric convective operator (see Eq.~\ref{block_diag}) as
follows
\begin{equation}
\label{convc_def}
\convc \equiv \dive \Us \Sscal ,
\end{equation}
\noindent where $\Us \equiv \diag (\velh) \in \real^{\vvlength \times \vvlength}$
is a diagonal matrix that contains the face velocities,
$\velh \in \real^{\vvlength}$.

\mbigskip

The cell-to-face gradient, {$\gradd \in \real^{\vvlength \times \pvlength}$ is related
with the discrete (integrated) divergence operator, $\dive$, via}
\begin{equation}
\label{gradd}
{\gradd \equiv - \vcvects^{-1} \dive\traspose .}
\end{equation}
{\noindent Then, the discrete Laplacian operator, $\lapld \in
\real^{\pvlength \times \pvlength}$ is, by construction, a symmetric
negative semi-definite matrix}
\begin{equation}
\label{lapld}
{\lapld \equiv \dive \gradd = - \dive \vcvects^{-1} \dive\traspose ,}
\end{equation}
{\noindent Together with Eq.(\ref{stag_mesh}), they respectively lead
to the discrete gradient}
\begin{equation}
[ \vcvects \gradd \presh ]_{\nface} = ( \press_{\cpA} - \press_{\cpB} ) A_{\nface} \hphantom{kk} \Longrightarrow \hphantom{kk}
[ \gradd \presh ]_{\nface} = \frac{\press_{\cpA} - \press_{\cpB}}{\delta_{\nface}} ,
\end{equation}
\noindent 
Laplacian and diffusive operators (see Eq.~\ref{block_diag})
\begin{equation}
[ \lapld \scafieldc ]_{\nvc} = \sum_{\nface \in \Fcell{\nvc}} \frac{( \sca_{\cpA} - \sca_{\cpB} ) A_{\nface}}{\delta_{\nface}} \hspace{6.69mm} \text{and} \hspace{6.69mm}
\diffc \equiv \nu \lapld ,
\end{equation}
\noindent where $\nu$ is the kinematic viscosity. Notice that this
discretization of the diffusive operator is valid
for {incompressible} fluids with constant viscosity. For
non-constant viscosity values, the discretization method has to be
modified accordingly~\cite{TRI12-DmuS}. Finally, the cell-to-face
(momentum) interpolation is constructed as follows
\begin{equation}
\label{Scs_def}
\Scs \equiv \NormalVect{} \vcvects^{-1} \Sthreescal \vcvect \hspace{6.69mm} \text{where} \hphantom{k} \Sthreescal = \Identity_{\dim} \otimes \Sscal ,
\end{equation}
\noindent which {is needed to construct the cell-to-cell gradient
operator, $\graddc \equiv \Ssc \gradd$, where
$\Ssc \equiv \vcvect^{-1} \Scs\traspose \vcvects$ is the face-to-cell
interpolation. Notice that $\Scs$} is basically a volume-weighted
interpolation. It must be noted that an unweighted interpolation,
$\Scs = \NormalVect{} \Sthreescal$, was proposed in the original
paper~\cite{TRI08-JCP}. However, as mentioned above, this can lead to
stability issues. This has been recently
addressed~\cite{SANTRI22ECCOMAS-FSM} showing that the volume-weighted
interpolation defined in Eq.(\ref{Scs_def}) is necessary to guarantee
that the method is unconditionally stable regardless of the mesh
quality.

\begin{figure}[!t]
  \centering{ \includegraphics[height=0.46\textwidth,angle=0]{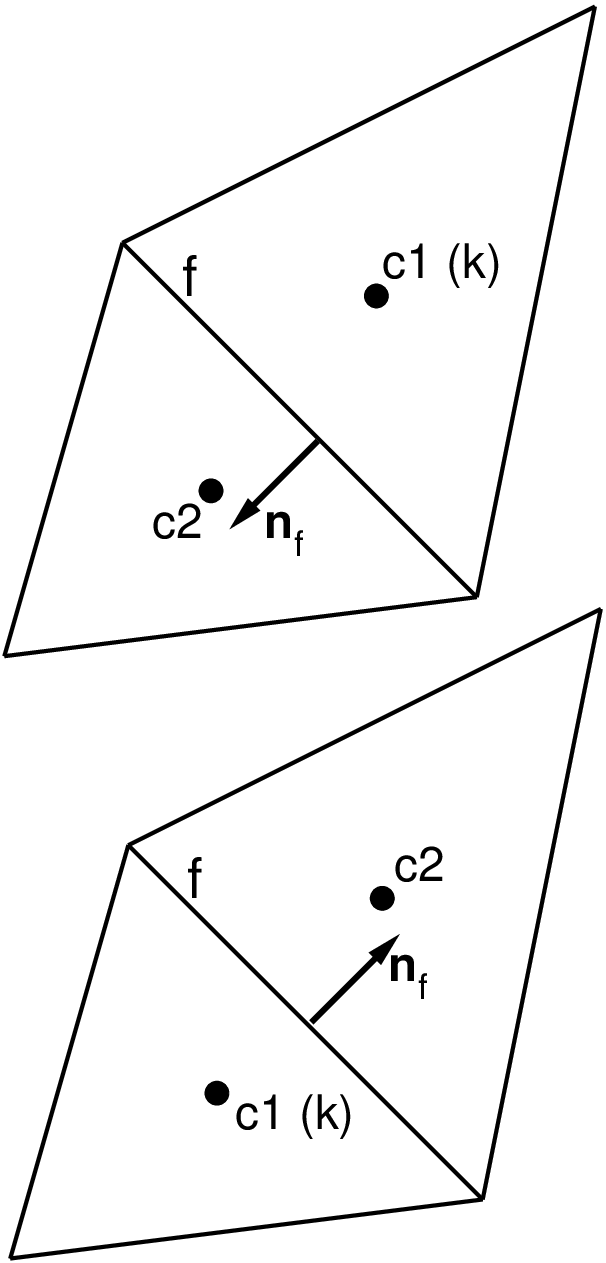} \hspace{9.69mm} \includegraphics[height=0.43\textwidth,angle=0]{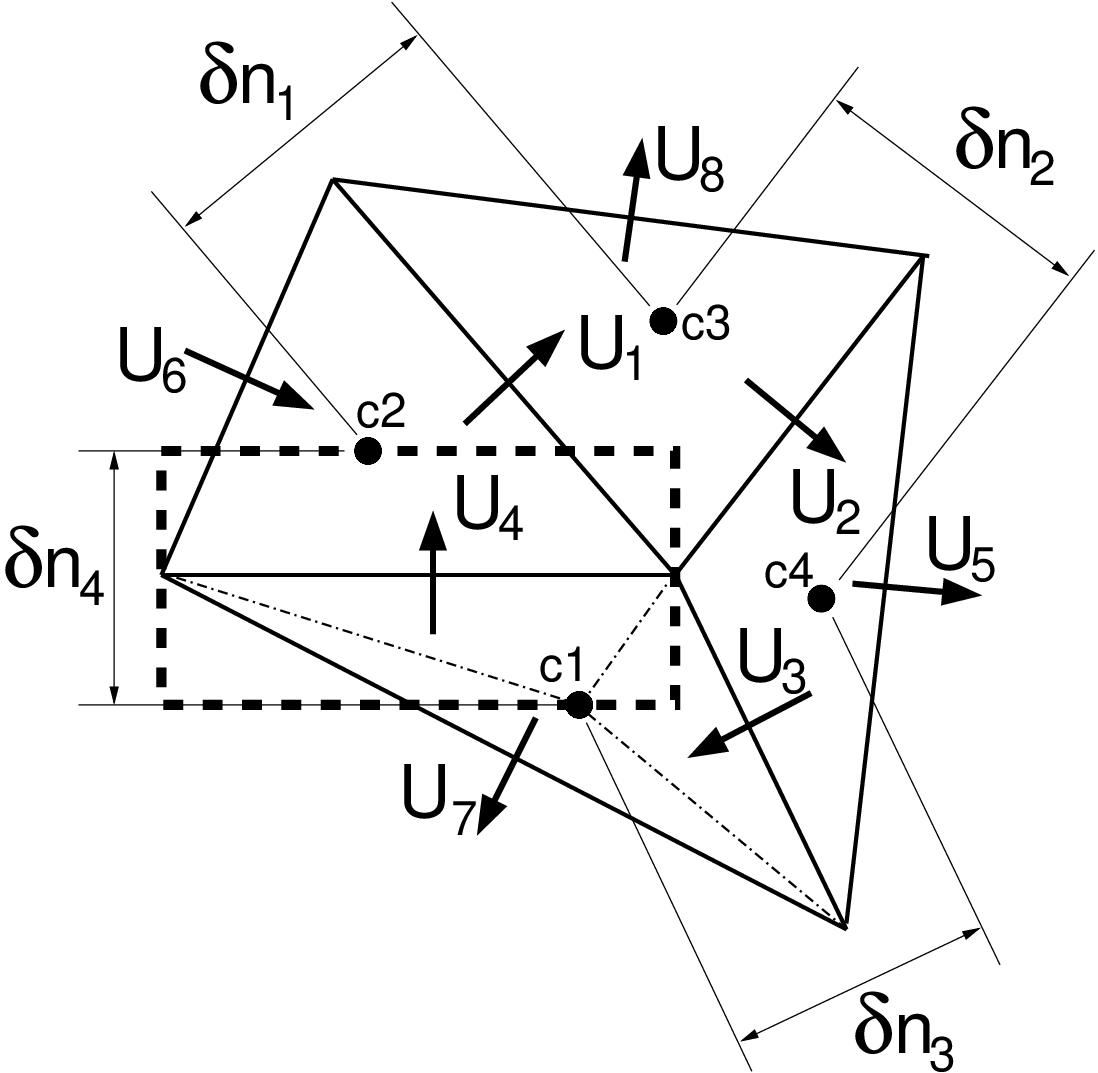}
    }
\caption{Left: face normal and neighbor labeling criterion. Right:
  definition of the volumes, $\vcvects$, associated with the the
  face-normal velocities, $\velh$. Thick dashed rectangle is the
  volume associated with the {staggered} velocity $\mathrm{U}_{4}
  = [ \velh ]_{4}$, \ie~$[ \vcvects ]_{4,4} = A_{4} \delta_{4}$ where
  $A_{4}$ is the face area and $\delta_{4} =
  | \normal_{4} \cdot \protect \overrightarrow{\cpA \cpB} |$ is the
  projected distance between adjacent cell centers. Thin dash-dotted
  lines are placed to illustrate that the sum of volumes is exactly
  preserved $\trace{\vcvects}=\trace{\vcvect}= d \trace{\vcvectc}$
  ($d=2$ for 2D and $d=3$ for 3D) regardless of the location of the
  cell nodes.}
\label{mesh}
\end{figure}


\section{Rethinking CFL condition: eigenbounds of convective and diffusive operators}

\label{EigenCD_vs_CFL}

\subsection{Gershgorin-based linear stability analysis}

Explicit (and semi-explicit) time-integration schemes impose severe
restrictions on the time-step, $\Dt$, due to the fact that the
eigenvalues of the amplification matrix must lie inside the stability
region of the time-integration method. Namely, linearizing (if needed)
the dynamical system (\eg~momentum equation on a 3D collocated mesh
with $\pvlength$ volumes and $\vvlength$ faces) leads to
\begin{equation}
\label{dyn_sys}
\frac{\ud \velhc}{\ud t} = \Rud{}_{\vel} \velhc \hspace{6.69mm} \text{where} \hspace{3.69mm} \Rud{}_{\vel} = \left( \Identity_{\dim} \otimes \Rud{} \right) \in \real^{\dim \pvlength \times \dim \pvlength} ,
\end{equation}
\noindent where the matrix $\Rud{} \equiv \vcvectc^{-1} ( - \convc +
\diffc ) \in \real^{\pvlength \times \pvlength}$ accounts for
the effects of convection and diffusion, and $\velhc \in (\veluc ,
\velvc , \velwc )\traspose \in \real^{\dim \pvlength}$. Then,
different time-integration schemes lead to different stability
regions~\cite{WES01}. The simplest example thereof in the first-order
Euler explicit scheme:
\begin{equation}
\label{forward_Euler}
\frac{\velhc^{n+1}-\velhc^{n}}{\Dt} = \Rud{}_{\vel} \velh^{n} \hspace{3.69mm} \Longrightarrow \hspace{3.69mm} \velhc^{n+1} = \left( \Identity_{\dim} \otimes \Ampl \right) \velhc^{n} \hspace{2.69mm} \text{where} \hspace{2.69mm} \Ampl \equiv ( \Identity + \Dt \Rud{} ) .
\end{equation}
\noindent The A-stability is guaranteed if the spectral radius of the
amplification matrix, $\Ampl$, is smaller than one, \ie~$\rho( \Ampl )
< 1$. This leads to the stability region in terms of the eigenvalues
of $\tilde{\Rud{}} \equiv \Dt \Rud{}$ shown in
Figure~\ref{stab_region} (top). Similar analysis can be done for
other temporal schemes~\cite{WES01}. An example thereof is shown in
the same figure for the one-parameter second-order explicit method
\begin{equation}
\label{k1L2}  
\frac{\velhc^{n+\timeintparam+1/2}-\velhc^{n+\timeintparam-1/2}}{\Dt} = \Rud{}_{\vel} \velhc^{n+\timeintparam} ,
\end{equation}
\noindent where the off-step velocities are given by
\begin{equation}
\label{off-step}  
\velhc^{n+\timeintparam+1/2} = ( \timeintparam + 1/2 ) \velhc^{n+1} - ( \timeintparam - 1/2 ) \velhc^{n} \hspace{3.69mm} \text{and} \hspace{3.69mm} \velhc^{n+\timeintparam} = ( 1 + \timeintparam ) \velhc^{n} - \timeintparam \velhc^{n-1} .
\end{equation}
\noindent This time-integration scheme named $\timeintparam$1L2
(see~\ref{SAT_summary}, for details) can be viewed as a generalization
of the classical second-order Adams--Bashforth (AB2) scheme
($\timeintparam=1/2$). This was used in Refs.~\cite{VER98,VER03} for
DNS of incompressible flows keeping the parameter $\timeintparam$
constant during the simulation. Then, in Ref.~\cite{TRI08-JCP2}, a
self-adaptive strategy was proposed: the parameter $\timeintparam$ is
being re-computed to adapt the linear stability domain to the
instantaneous flow conditions in order to maximize $\Dt$. The idea of
the method is depicted in Figure~\ref{stab_region} (bottom). Hence, at
the end, this or any other method necessarily relies on bounding the
eigenvalues of the dynamical system, \ie~in our case finding
eigenbounds of the matrix $\Rud{}$ given in Eq.(\ref{dyn_sys}). In the
original work~\cite{TRI08-JCP2}, this was done by applying the
Gershgorin circle theorem to $\vcvectc^{-1} \convc$ and
$\vcvectc^{-1} \diffc$ together with the Bendixson theorem {(for
a graphical representation see Figure~\ref{stab_region}, bottom).}
\begin{theorem}[Bendixson~\cite{BEN1902}]
\label{Bendixson_theorem}
Given two square matrices of equal size, $\istensor{X}$ and
$\istensor{Y}$, one with real-valued eigenvalues,
$\lambda^{\istensor{X}} \in \real$, and the other with imaginary ones,
$\lambda^{\istensor{Y}} \in \imag$, then every eigenvalue of the sum,
$\istensor{X}+\istensor{Y}$, is contained in the rectangle
\begin{equation}
\label{Bendixson}
\lambda^{\istensor{X}}_{\min} \le \realpart(\lambda^{\istensor{X}+\istensor{Y}}) \le \lambda^{\istensor{X}}_{\max} \hspace{6.69mm}
\imagpart(\lambda^{\istensor{Y}}_{\min}) \le \imagpart(\lambda^{\istensor{X}+\istensor{Y}}) \le \imagpart(\lambda^{\istensor{Y}}_{\max}) .
\end{equation}
\end{theorem}
\noindent This can be easily applied to matrix $\Rud{} = -
\vcvectc^{-1} \convc + \vcvectc^{-1} \diffc$ recalling that $\convc = -
\convctraspose$, \ie~$\lambda^{\convmat} \in \imag$, and $\diffc =
\diffc\traspose$ negative semi-definite,~\ie~$\lambda^{\diff} \in
\real_{\le 0}$. At this point, there are a couple of technical issues
that worth mentioning. Although the (skew-)symmetry is lost when
matrices $\convc$ and $\diffc$ are left-multiplied by $\vcvectc^{-1}$,
their eigenvalues are still imaginary and real-valued,
respectively. They actually have the same spectrum as the
(skew-)symmetric matrices $\vcvectc^{-1/2} \convc \vcvectc^{-1/2}$ and
$\vcvectc^{-1/2} \diffc \vcvectc^{-1/2}$,
\begin{align}
\nonumber \vcvectc^{-1} \diffc \velv = \lambda^{\diff} \velv \hspace{1.69mm} &\Longrightarrow \vcvectc^{1/2} ( \vcvectc^{-1} \diffc ) \vcvectc^{-1/2} \vcvectc^{1/2} \velv = \lambda^{\diff} \vcvectc^{1/2} \velv \\
  &\Longrightarrow \hspace{1.69mm} (\vcvectc^{-1/2} \diffc \vcvectc^{-1/2}) \velw = \lambda^{\diff} \velw ,
\end{align}
\noindent where $\velw = \vcvectc^{1/2} \velv$. Notice that the matrix
$\vcvectc$ has strictly positive diagonal elements. This method to
bound the eigenvalues of $\Rud{}$ was originally proposed and referred
as \EigenCD~in Ref.~\cite{TRI08-JCP2}. Later, it was successfully used
for a large variety of DNS and LES simulations on both structured and
unstructured meshes (see
Refs.~\cite{TRI07-IJHMT-I,TRI07-IJHMT-II,TRI08-JCP2,TRI14-CYL,PONTRI18-BFS,DABTRI19-3DTOPO-RB,CAL21},
among others).

\begin{figure}[!t]
  \centering{
    \includegraphics[height=0.69\textwidth,angle=-90]{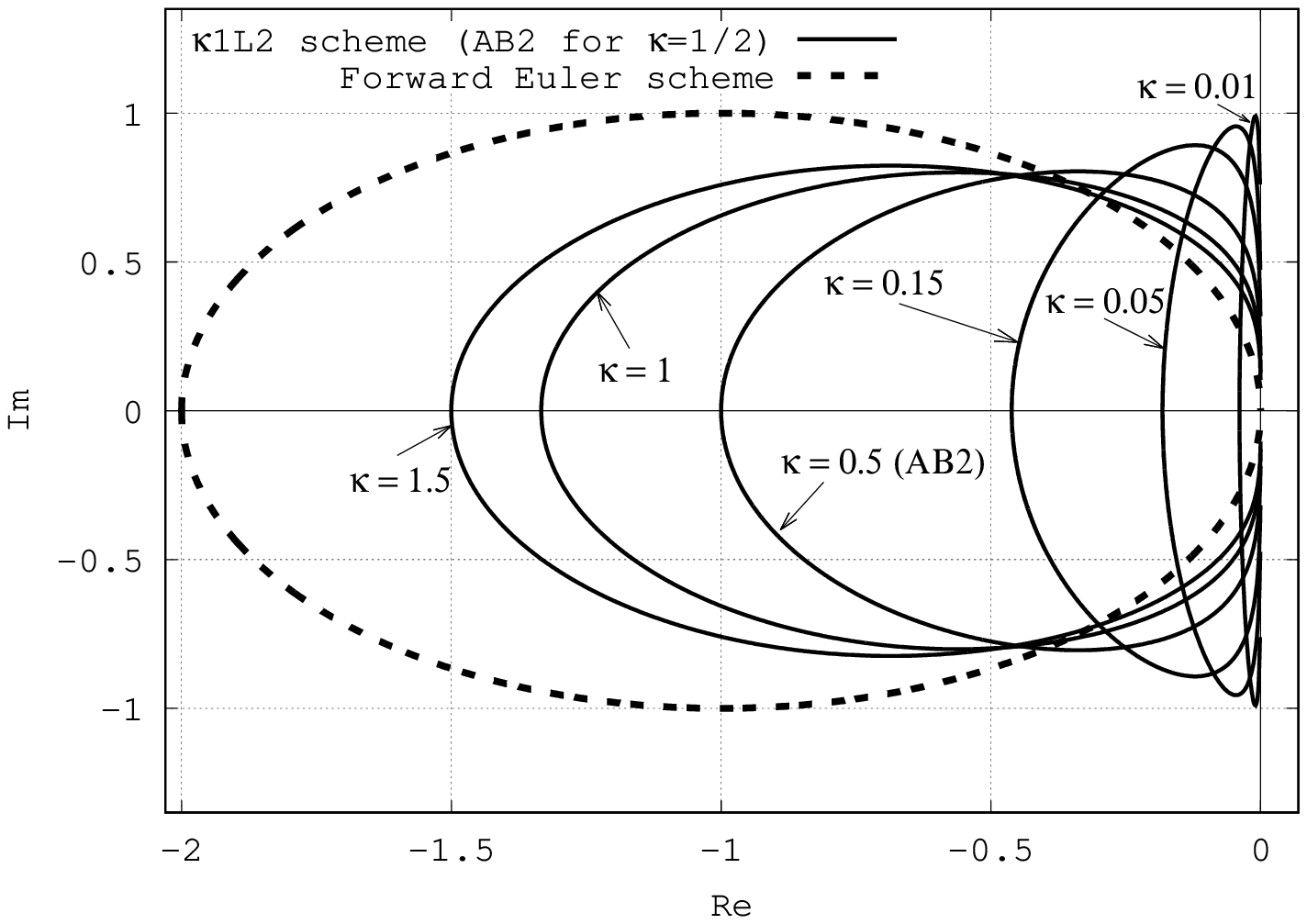}
    \includegraphics[height=0.69\textwidth,angle=-90]{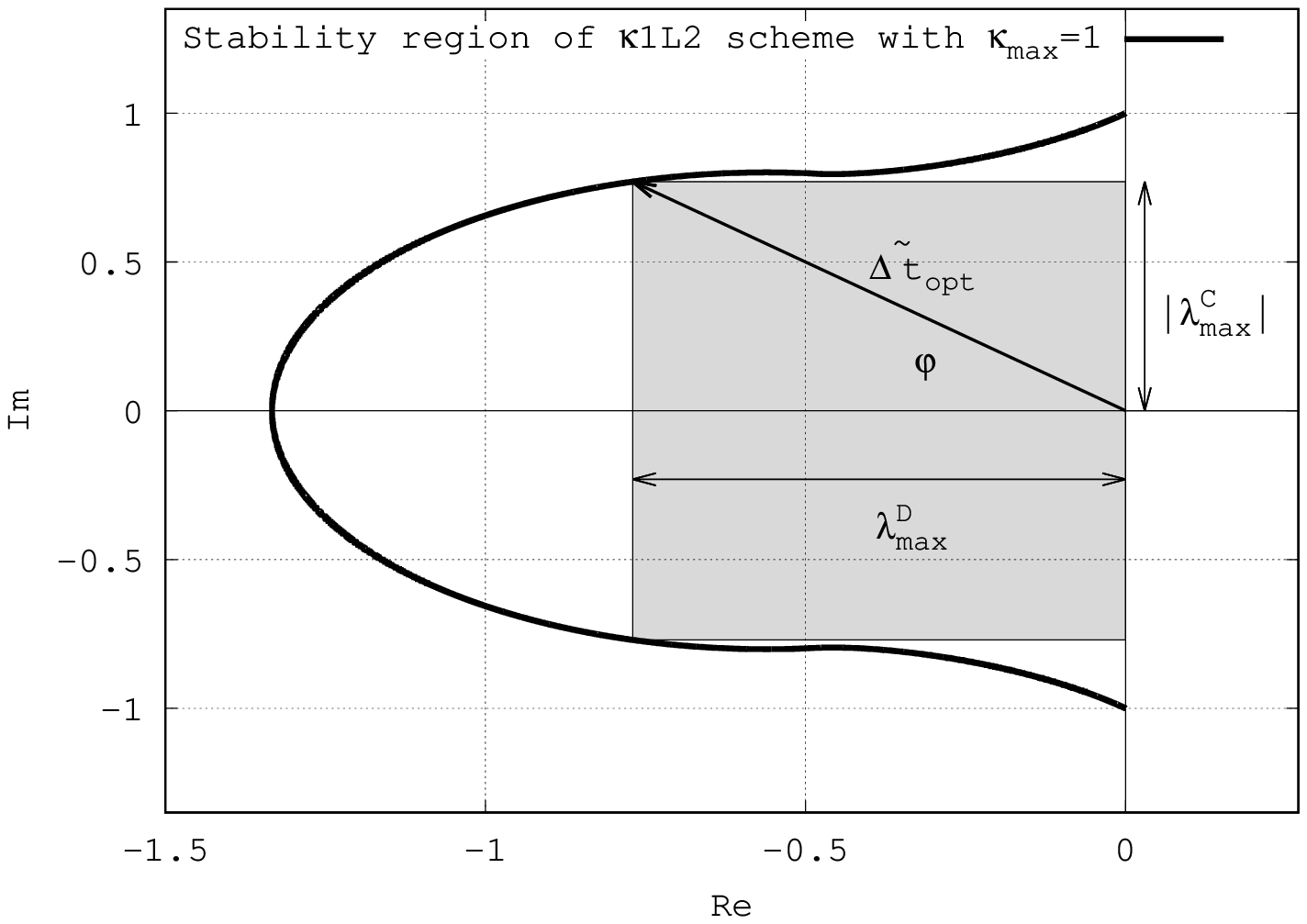}
  }
\caption{Stability region of the first-order forward Euler
  scheme (Eq.~\ref{forward_Euler}) together with the family of
  $\timeintparam$-dependent second-order $\timeintparam$1L2
  time-integration scheme (see Eqs.~\ref{k1L2} and~\ref{off-step}
  and~\ref{SAT_summary}) (top) and their envelope (bottom). {The
  shaded region is a graphical representation of the Bendixson
  theorem (see Theorem~\ref{Bendixson_theorem}).}}
\label{stab_region}
\end{figure}


\subsection{CFL condition: brief historical review}

\label{CFLhistory}

Nevertheless, this is not the standard way to bound the eigenvalues of
$\Rud{}$. In the CFD literature, and in virtually all CFD packages,
stability constraints for $\Dt$ are usually expressed in terms of the
so-called {\it CFL} condition originally proposed in the seminal
paper~\cite{CFL28} by R.~{\it C}ourant, K.~{\it F}riedrichs, and
H.~{\it L}ewy in 1928! They derived the following stability condition
\begin{equation}
\label{CFL}
C = \frac{u \Dt}{\Dx} < C_{\max} ,
\end{equation}
\noindent for a 1D transport equation
\begin{equation}
\frac{\partial \sca}{\partial t} + u \frac{\partial \sca}{\partial x} = 0 ,
\end{equation}
\noindent discretized in a uniform mesh with spacing equal to $\Dx$
where $u$ is the advection velocity. The intuitive idea or ``physical
interpretation" of this formulae can be found, for instance, in the
OpenFOAM\textsuperscript{\textregistered} documentation as {\it``a
measure of the rate at which information is transported under the
influence of a flux field"}~\cite{OpenFOAM}. This, or very similar
formulae can be found in NEK5000~\cite{NEK5000},
COMSOL\textsuperscript{\textregistered}~\cite{COMSOL} or
Basilisk~\cite{BASILISK} codes, among many others. An alternative
definition is used in ANSYS-Fluent~\cite{ANSYS-Fluent}
\begin{equation}
\label{CFL_ANSYS-Fluent}
CFL = \frac{\Dt \sum_{faces} \lambda_f^{\max} A_f}{2V} ,
\end{equation}
\noindent where $A_f$ are the face areas, $V$ is the cell volume and
$\lambda_f^{\max}$ is the maximum of the local eigenvalues. For
incompressible (also compressible at low speed) flows, $\lambda_f =
U_f$ (here, $U_f$ is the face velocity); therefore, this CFL condition
becomes identical to the definition used in
OpenFOAM\textsuperscript{\textregistered}~\cite{OpenFOAM}, SU2
code~\cite{SU2} or Code Saturne~\cite{Saturne} and slightly different
than the definition used in the DLR-TAU code (see Eq.~18
in~\cite{LAN14}). {Interestingly enough, this formula can be
obtained applying the Gershgorin circle theorem to
$\vcvectc^{-1} \conv$ in the particular case where a second-order
symmetry-preserving scheme is used, as we shown in~\cite{TRI08-JCP2}.}
Nevertheless, {it is not
clear when this formula was originally proposed} (at least, not for
the authors) and according to~\cite{privMENTER}, it goes back to
Eq.(22) in Ref.~\cite{WEI99} where the following definition of the CFL
condition is given
\begin{equation}
\label{CFL_Weiss_AIAA99}
CFL = \frac{\Dt \lambda_{max}}{V} ,
\end{equation}
\noindent where $\lambda_{max}$ is the maximum eigenvalue of the
system given by $| u |$ for incompressible (also compressible at low
speed) flows. It must be noted that a multiplication by the face area,
$A_f$, is missing in Eq.(\ref{CFL_Weiss_AIAA99}). Moreover, no
summation by faces is specified here. Going back to previous works by
the same authors, we find the same definition in Ref.~\cite{WEI97}
without specifying how the eigenvalues are being computed. Moreover,
in Ref.~\cite{WEI95} (see Eq.~16) they used the following formula for
bounding the $\Dt$,
\begin{equation}
\label{CFL_Weiss_AIAA95}
\Dt = \min \left( \frac{CFL \Dx}{u'+c'}, \frac{\sigma \Dx^2}{\nu} \right) ,
\end{equation}
\noindent where $\sigma$ is referred as von Neumann number, $u'$ and
$c'$ are respectively the velocity and the speed of sound for the
non-preconditioned system and $\Dx$ is defined as the inter-cell
length scale over which diffusion occurs. Furthermore, in
Ref.~\cite{WEI99-AIAAconf} (Eq.~4) we find the following formula
\begin{equation}
\label{CFL_Weiss_AIAAconf99}
CFL = \Dt \lambda_{max} ( \istensor{D} ) ,
\end{equation}
\noindent where $\lambda_{max} ( \istensor{D} )$ is the maximum eigenvalue
of the chemical Jacobian. The time integration method is a first-order
implicit Euler scheme and the condition~(\ref{CFL_Weiss_AIAAconf99})
is used to keep the system positive-definite,~\ie~$( \Identity
- \Dt \istensor{D} )$ is a positive-definite matrix. The eigenvalues
of $\istensor{D}$ were determined numerically using the LAPACK
library~\cite{LAPACK}.

\mbigskip

{Similar expressions can also be derived through a von Neumann
stability analysis, \ie~non-growth of Fourier modes in the frozen
coefficient case on an unbounded domain (see Ref.~\cite{WES96}, for
instance). Therefore, the analysis is restricted to uniformly spaced
Cartesian grids with periodic boundary conditions. This inherent
restrictions of the Fourier analysis can be by-passed using a local
von Neumann analysis with the local values of the coefficients (see
Chapter~5 in Ref.~\cite{WES01}). This is applicable for Cartesian
meshes with a smooth stretching. However, the analysis cannot be
easily extended to unstructured grids.} In any case, CFL condition
became soon very popular among all the CFD community. To celebrate the
article's 40th anniversary, in 1967 the IBM Journal published a
special issue, that included the English translation of the original
paper~\cite{CFL67}. In 2010, the meeting ``CFL-condition, 80 years
gone" was held in Rio de Janeiro~\cite{CFL13}.


\subsection{Two sides of the same coin}

The CFL condition given in Eq.(\ref{CFL}) can be easily related to the
above explained stability constraints imposed by the eigenvalues of
the matrix $\Rud{}$ given in Eq.(\ref{dyn_sys}). Let us consider a 1D
uniformly spaced mesh with constant advective velocity, $u$. In this
case, the convective and diffusive terms in the NS equations simplify
to
\begin{equation}
\label{Conv-Diff_1D}
\frac{\partial \sca}{\partial t} = - u \frac{\partial \sca}{\partial x} + \nu \frac{\partial^2 \sca}{\partial x^2} .
\end{equation}
\noindent Then, a second-order semi-discrete finite-difference (also
finite-volume) discretization of Eq.(\ref{Conv-Diff_1D}) leads to
\begin{equation}
  \frac{\partial \sca_i}{\partial t} = - u \frac{\sca_{i+1} - \sca_{i-1}}{2 \Dx} + \nu \frac{\sca_{i+1} - 2 \sca_i + \sca_{i-1}}{\Dx^2} .
\end{equation}
\noindent This can be re-arranged in a matrix-vector form as follows
\begin{equation}
\resizebox{0.96\hsize}{!}{$\dfrac{\partial \scafield_h}{\partial t} = 
  \begin{bmatrix}
    0 & \ddots & & \\
    \ddots & 0 & \ddots & \\
    & \dfrac{u}{2\Dx} & 0 & -\dfrac{u}{2\Dx} & \\
    & & \ddots & 0 & \ddots \\ 
    & & & \ddots & 0 & \\ 
  \end{bmatrix} \scafield_h
  +
  \begin{bmatrix}
    0 & \ddots & & \\
    \ddots & 0 & \ddots & \\
    & \dfrac{\nu}{\Dx^2} & -\dfrac{2\nu}{\Dx^2} & \dfrac{\nu}{\Dx^2} & \\
    & & \ddots & 0 & \ddots \\ 
    & & & \ddots & 0 & \\ 
  \end{bmatrix} \scafield_h ,$}
\end{equation}
\noindent where $\scafield_h = ( \sca_1, \cdots, \sca_\pvlength
)\traspose \in \real^{\pvlength}$ is a column vector containing all
the components of the scalar field $\scafield$. Hence, eigenvalues of
the convective and diffusive part can be bounded using the Gershgorin
circle theorem as follows
\begin{equation}
|\lambda^{\convmat}| \le \frac{|u|}{\Dx} \hspace{10.69mm} |\lambda^{\diff}| \le \frac{4 \nu}{\Dx^2} ,
\end{equation}
\noindent which leads to the classical CFL definition proposed almost
a century ago~\cite{CFL28}. {Notice that we are assuming that
the eigenvalues, and in particular its spectral radius, provide an
upper bound for the growth of the power of a matrix. This is true for
non-defective matrices since, as shown in Refs.~\cite{TRE93,TRE97},
the growth of the power can be much larger for defective matrices with
Jordan blocks. In our case, all the matrices are either symmetric or
skew-symmetric (or a combination of both), therefore, they are normal
matrices,~\ie~non-defective matrices.}

\mbigskip

At this point, we expect that it becomes clear that it is probably
more appropriate (and more accurate) to get rid of generalizations of
the classical CFL definition given in Eq.(\ref{CFL}) for general cases
(\ie~multi-dimensional, non-uniform, non-constant velocity,
unstructured meshes...). {In general, for the sake of robustness,
these approaches tend to underestimate $\Dt$ leading to an increase in
the overall computational cost of the simulations.} Instead, the
Gershgorin circle theorem can be applied assuming that the
coefficients of the discrete convective, $\convc$, and diffusive,
$\diffc$, operators are available. This was the main idea of the paper
published one decade ago~\cite{TRI08-JCP}: {to use strict
eigenbounds resulting from the spatial discretization and not inexact
approximations combined with heuristic, sometimes even
trial-and-error, approaches.} In practice, CPU cost
reductions up to more than $4$ {times} were measured for
unstructured grids~\cite{TRI08-JCP2} compared with
a {more classical} CFL condition.


\section{A new efficient approach to compute eigenbounds of convection and diffusion matrices avoiding their construction}

\label{AlgEig_CD_method}

\subsection{Deconstructing convection and diffusion matrices}

Let us consider again the convective operator defined in
Eq.~(\ref{convc_def})
\begin{equation}
\label{conv_oper}
\convc \equiv \dive \Us \Sscal \in \real^{\pvlength \times \pvlength}, 
\end{equation}
\noindent where $\dive \in \real^{\pvlength \times \vvlength}$ is the
face-to-cell divergence operator,
$\Sscal \in \real^{\vvlength \times \pvlength}$ is cell-to-face
interpolation and $\Us = \diag
(\velh) \in \real^{\vvlength \times \vvlength}$ is a diagonal matrix
that contains the face velocities, $\velh \in \real^{\vvlength}$, that
change every time-step. {The direct application of the Gershgorin
circle theorem would require evaluating explicitly the coefficients of
$\convc$ at every time-step and then calling some specific function to
compute the corresponding eigenbounds. As explained before, this type
of approach would increase the code complexity hindering its efficient
cross-platform portability.}

\mbigskip

A similar problem exists for the diffusive term with non-constant (in
time) diffusivity
\begin{equation}
\label{diff_oper}
\diffc(\diffv) \equiv \dive \diffm \gradd \in \real^{\pvlength \times \pvlength},
\end{equation}
\noindent where $\diffm = \diag ( \diffv ) \in \real^{\vvlength
  \times \vvlength}$ is a diagonal matrix containing the diffusivity
  values at the faces, $\diffv \in
\real^{\vvlength}$. Notice that this is also relevant for
eddy-viscosity turbulence models. For details about the
discretization, the reader is referred to Section~\ref{SymPres} or to
the original paper~\cite{TRI08-JCP}.

\mbigskip

At this point, we aim to answer the following research question: {\it
  can we avoid to explicitly reconstruct at each time-step both
  convective, $\convc$, and diffusive, $\diffc(\diffv)$, matrices
  while still being able to compute proper eigenbounds in an
  inexpensive manner?}  To do so, let us firstly write the divergence
  operator, $\dive$, in terms of the cell-to-face,
  $\Tcs \in \real^{\vvlength \times \pvlength}$ and face-to-cell,
  $\Tsc \in \real^{\pvlength \times \vvlength}$, incidence matrices
\begin{equation}
\label{M_deconstuct}
\dive \equiv \Tsc \As \in \real^{\pvlength \times \vvlength},
\end{equation}
\noindent where $\As \in \real^{\vvlength \times \vvlength}$ is a
diagonal matrix containing the face surfaces. Moreover, recalling the
duality between the divergence and the gradient operators (see
Eq.~\ref{gradd})
\begin{equation}
\dive = - ( \vcvects \gradd )\traspose \hspace{3.69mm} \Longrightarrow \hspace{3.69mm} \gradd = - \vcvects^{-1} \dive\traspose , 
\end{equation}
\noindent together with the relation $\Tsc = \Tcs\traspose$ leads to
\begin{equation}
\gradd \equiv - \vcvects^{-1} \As \Tsc\traspose = - \DX^{-1} \Tcs ,
\end{equation}
\noindent where $\DX \equiv \vcvects \As^{-1} \in \real^{\vvlength \times \vvlength}$
is a diagonal matrix containing the projected distances, $\delta
n_{\nface} = | \normal_{\nface} \cdot \overrightarrow{\cpA \cpB} |$,
between the cell centers, $\cpA$ and $\cpB$, of the two cells adjacent
to a face, $f$ (see Figure~\ref{mesh}). Plugging all this into the
definition of the diffusive operator~(\ref{diff_oper}) leads to
\begin{equation}
\label{diff_oper2}
\diffc(\diffv) = - \Tsc \As \diffm \DX^{-1} \Tcs = - \Tsc \pdiffm \Tcs = - \Tcs\traspose \pdiffm \Tcs ,
\end{equation}
\noindent where the diagonal matrix $\pdiffm = \As \diffm \DX^{-1} \in
\real^{\vvlength \times \vvlength}$ has strictly positive diagonal
coefficients. Hence, the diffusive operator is symmetric and negative
semi-definite (see Theorem~\ref{symmetry} in~\ref{appendix_play})
likewise the continuous Laplacian, $\lapl$.

\mbigskip

Similarly, the convective term given in Eq.(\ref{conv_oper}) can be
written as follows
\begin{equation}
\label{conv_oper2}
\convc = \Tsc \Us \As \Sscal ,
\end{equation}
\noindent where the cell-to-face interpolation, $\Sscal$, defines the
numerical scheme we are using. For instance, taking
\begin{equation}
\Sscal^{\SP} = \frac{1}{2} | \Tcs | ,
\end{equation}
\noindent leads to a skew-symmetric matrix, \ie~$\convc = -\convctraspose$
that corresponds to the second-order symmetry-preserving
discretization~\cite{TRI08-JCP,VER03} (see Theorem~\ref{skew-symmetry}
in~\ref{appendix_play}). {Here $|\genmat|$ denotes the entry-wise
absolute value of a real-valued matrix, \ie~$[|\genmat|]_{ij} =
|[\genmat]_{ij}|$.} In summary, convective and diffusive operators
read
\begin{align}
\label{diff_oper3}
\diffc ( \diffv ) &= -\Tcs\traspose \pdiffm \hphantom{|}\Tcs\hphantom{|} \hspace{3.69mm} \text{where } \pdiffm \text{ is a diagonal matrix with } {[\diag(\pdiffm)]_{i} > 0} \hphantom{k} \forall i, \\
\label{conv_oper3}
2 \convc          &= \hphantom{-}\Tcs\traspose \Fs |\Tcs| \hspace{3.69mm} \text{where } \Fs \equiv \As \Us \text{ and } \diag(\Fs) \in ker(\Tcs\traspose) , 
\end{align}
\noindent where, in general, both diagonal matrices $\pdiffm$
(diffusive fluxes) and $\Fs$ (mass fluxes) change on time. Notice that
$\diag(\Fs) \in ker(\Tcs\traspose)$ follows from the incompressibility
constraint given in Eq.(\ref{discr_mass}) and the definition of the
divergence operator given in Eq.(\ref{M_deconstuct}).


\subsection{Eigenbounds for the diffusion matrix}

The idea at this point is to construct other matrices {with} the
same spectrum (except for the zero-valued eigenvalues). To do so, we
will use the following property:
\begin{theorem}
\label{same_eigenvalues}
Let $\istensor{A} \in \real^{\pvlength \times \vvlength}$ and
$\istensor{B} \in \real^{\vvlength \times \pvlength}$ be two
rectangular matrices and $\vvlength
\ge \pvlength$. Then, the square matrices $\istensor{A} \istensor{B}
\in \real^{\pvlength \times \pvlength}$ and $\istensor{A}\traspose
\istensor{B}\traspose \in \real^{\vvlength \times \vvlength}$ have the
same eigenvalues except for the zero-valued ones.
\end{theorem}
\begin{proof}
A square matrix $\istensor{Q}$ and its transpose,
$\istensor{Q}\traspose$, have the same characteristic polynomial,
\ie~$\det(\lambda \Identity - \istensor{Q})=\det(\lambda \Identity -
\istensor{Q}\traspose)=0$; therefore, they also have the same
spectrum. Then, both $\istensor{A}\traspose \istensor{B}\traspose$ and
$\istensor{B} \istensor{A}$ have the same spectrum
\newcommand{\myrightarrow}{\hspace{2.69mm} \rightarrow \hspace{2.69mm}}
\newcommand{\myleftrightarrow}{\hspace{2.69mm} {\Leftrightarrow} \hspace{2.69mm}}
\begin{equation}
\istensor{A}\traspose \istensor{B}\traspose \velw_i = \lambda_i \velw_i \myleftrightarrow
\istensor{B} \istensor{A} \velz_i = \lambda_i \velz_i \hspace{6.69mm} \forall i \in \{1,\dots,\vvlength\} .
\end{equation}
\noindent Then, let $\lambda \neq 0$ be an eigenvalue of $\istensor{A}
\istensor{B}$ with an associated eigenvector $\velv$,
\begin{equation}
\istensor{A} \istensor{B} \velv= \lambda \velv \myrightarrow
\istensor{B} \istensor{A} (\istensor{B} \velv) = \lambda (\istensor{B} \velv) \myrightarrow
\istensor{B} \istensor{A} \velz = \lambda \velz .        
\end{equation}
\noindent Notice that $\istensor{B} \velv \neq \vecnull$ since $\lambda \neq
0$. Hence, $\lambda$ is a non-zero eigenvalue of $\istensor{B}
\istensor{A}$ and subsequently also an eigenvalue of
$\istensor{A}\traspose \istensor{B}\traspose$.  \qedhere
\end{proof}

\mbigskip

Therefore, a family of $\alpha$-dependent matrices with the same
spectrum (except for the zero-valued eigenvalues) as those given in
Eqs.(\ref{diff_oper3}) and~(\ref{conv_oper3}) can be constructed using
Theorem~\ref{same_eigenvalues}. Namely, matrix
\begin{align}
\label{alpha_def1}
- {(}\pdiffm^{\alpha} &\Tcs{)}\hphantom{|}{(}\Tcs\traspose\hphantom{|} \pdiffm^{1-\alpha}{)} && \text{(diffusive)}, 
\end{align}
\noindent has the same spectrum as {$-(\pdiffm^{\alpha} \Tcs)\traspose (\Tcs\traspose \pdiffm^{1-\alpha})\traspose = -\Tcs\traspose \pdiffm \Tcs$}.
Consequently,
\begin{equation}
\label{alpha_def2}
\rho ( \diffc ( \diffv ) ) = \rho ( \pdiffm^{\alpha} \Tcs \hphantom{|}\Tcs\traspose\hphantom{|} \pdiffm^{1-\alpha} ) .
\end{equation}
\noindent regardless of the values of $\alpha$.

\mbigskip

For instance, the following four matrices have the same spectrum
(except for the zero-valued eigenvalues)
\begin{equation}
\label{Diffs_same_spectrum}
\left\{
- \Tcs\traspose \pdiffm \Tcs \sepcomma -\Tcs \Tcs\traspose \pdiffm \sepcomma -\pdiffm^{1/2} \Tcs \Tcs\traspose \pdiffm^{1/2} \sepcomma -\pdiffm \Tcs \Tcs\traspose  
\right\} ,
\end{equation}
\noindent {where the last three correspond to values of
$\alpha=0$, $1/2$, and $1$ in Eq.(\ref{alpha_def2}), respectively.}
The advantage of the new forms is that only the matrix
$-\Tcs \Tcs\traspose$ has to be computed (once) and stored. Note that
this face-to-face matrix has {$-2$ in the diagonal and $\pm 1$ in
the non-zero off-diagonal elements, which correspond to the faces of
the two adjacent control volumes (see Eq.~\ref{Tcs_Tsc_def}
in~\ref{matrix_construction})}. Then, {to find an upper bound (in
absolute value) of the eigenvalues, we can apply the Gershgorin circle
theorem as follows}
\begin{subequations}
\begin{eqnarray}
\label{eigenboundD1}
\rho ( \diffc ( \diffv ) ) =& \rho(\Tcs \Tcs\traspose \pdiffm) &\le \max\{ | \Tcs \Tcs\traspose | \diag( \pdiffm ) \} , \\
\label{eigenboundD2}
\rho ( \diffc ( \diffv ) ) =& \rho(\pdiffm^{1/2} \Tcs \Tcs\traspose \pdiffm^{1/2}) &\le \max\{ \diag(\pdiffm^{1/2}) \Hadprod | \Tcs \Tcs\traspose | \diag( \pdiffm^{1/2} ) \} , \\
\label{eigenboundD3}
\rho ( \diffc ( \diffv ) ) =&  \rho(\pdiffm \Tcs \Tcs\traspose) &\le \max\{ \diag(\pdiffm) \Hadprod | \Tcs \Tcs\traspose | \vecones \} ,
\end{eqnarray}
\end{subequations}
\noindent where $\Hadprod$ denotes the Hadamard product (element-wise
product) {and $\vecones \in \real^{\vvlength}$ is a vector of
ones defined at the faces}. As stated above, these three forms
correspond to values of $\alpha=0$, $1/2$ and $1$ in
Eq.(\ref{alpha_def2}), respectively.
\begin{remark}
\label{remark_vcvectc}
In practice, we need estimations of the spectral radius of
$\vcvectc^{-1} \diffc ( \diffv )$ and not $\diffc ( \diffv )$. This can
be easily done by replacing $| \Tcs \Tcs\traspose |$ by $| \Tcs
\vcvectc^{-1} \Tcs\traspose |$ in Eqs.(\ref{eigenboundD1}),
(\ref{eigenboundD2}) and (\ref{eigenboundD3}). An equivalent remark
can be made for the forthcoming discussion about the convective
matrix, $\convc$.
\end{remark}


\subsection{Eigenbounds for the convective matrix}

The convective term given in Eq.(\ref{conv_oper3}) can be treated in a
similar manner. {Notice that} the
diagonal matrix $\Fs$ (mass fluxes across the faces) can take both
positive and negative values depending on the flow
direction. {In this case, matrices}
{
\begin{align}
\label{alpha_conv1}
{(}|\Fs|^{\alpha} &\Tcs{)} {(}| \Tcs\traspose | | \Fs |^{-\alpha} \Fs{)}
 && \text{(convective)}, \\
\label{alpha_conv2}
{(}|\Fs|^{\alpha-1} \Fs &\Tcs{)} {(}| \Tcs\traspose | | \Fs |^{1-\alpha}{)} && \text{(convective)}, 
\end{align}
}
\noindent {have the same spectrum as $\Tcs\traspose \Fs |\Tcs|$.
It must be noted that indeterminate forms $1/0$ may eventually occur
for $\alpha<0$ or $\alpha>1$ in Eqs.(\ref{alpha_conv1})
and~(\ref{alpha_conv2}) if a mass flux (diagonal terms of $\Fs$)
becomes zero.} {{Then, similarly} to
Eq.(\ref{Diffs_same_spectrum}), {the following} five
matrices have the same spectrum (except for the zero-valued
eigenvalues)}
\begin{equation}
\label{Convs_same_spectrum}
\left\{
\Tcs\traspose \Fs |\Tcs| \sepcomma  \Tcs |\Tcs\traspose| \Fs \sepcomma |\Fs|^{1/2} \Tcs |\Tcs\traspose| |\Fs|^{-1/2} \Fs \sepcomma |\Fs|^{-1/2} \Fs \Tcs |\Tcs\traspose| |\Fs| \sepcomma \Fs \Tcs |\Tcs\traspose| \right\} ,
\end{equation}
\noindent {where the last four correspond to values of $\alpha=0$,
$1/2$ in Eq.(\ref{alpha_conv1}), and $\alpha=1/2$ and $1$ in
Eq.(\ref{alpha_conv2}), respectively.} In the last four splittings,
only the matrix $\Tcs | \Tcs\traspose |$ has to be pre-computed and
stored. This matrix is skew-symmetric with $\pm 1$ in the non-zero
off-diagonal elements. Then, the Gershgorin circle theorem can be
applied as follows
\begin{subequations}
\begin{eqnarray}
\label{eigenboundC1}
2 \rho ( \convc) =& \rho(\Tcs |\Tcs\traspose| \Fs ) &\le \max\{ {\left| \Tcs |\Tcs\traspose|\right|} \diag( | \Fs | ) \} , \\
\label{eigenboundC2}
2 \rho ( \convc) =&\hspace{-3mm} \rho(|\Fs|^{\frac{1}{2}} \Tcs |\Tcs\traspose| |\Fs|^{-\frac{1}{2}} \Fs ) \hspace{-3mm} &\le \max\{ \diag(|\Fs|^{\frac{1}{2}}) \Hadprod {\left| \Tcs |\Tcs\traspose|\right|} \diag( |\Fs|^{-\frac{1}{2}} \Fs ) \} , \\
\label{eigenboundC3}
2 \rho ( \convc) =& \rho(\Fs \Tcs |\Tcs\traspose| ) &\le \max\{ \diag( | \Fs | ) \Hadprod {\left| \Tcs |\Tcs\traspose|\right|} \vecones \} , 
\end{eqnarray}
\end{subequations}
\noindent to find an upper bound of their eigenvalues, which, in this case,
lie on the imaginary axis. However, in practical flows, none of these
approaches is able to provide better (or, at least, similar estimates)
as applying the Gershgorin circle theorem directly to the matrix
$\convc$. A simple explanation for this is the following: matrix
${\left| \Tcs |\Tcs\traspose|\right|} \Fs$ has more non-zero
off-diagonal coefficients per row than matrix $\Tcs\traspose \Fs
|\Tcs|$, \eg~for a structured Cartesian mesh in $d$-dimensions, the
former has $2(2d-1)$ whereas the latter has only $2d$
non-zeros. Therefore, more mass fluxes (in absolute value) are
contributing to the calculation of the Gershgorin circle radii.

\mbigskip
\begin{theorem}[Perron--Frobenius theorem~\cite{PER1907,FRO1912}]
\label{Perron-Frobenius_theorem}
Given a real positive square matrix, \ie~$\istensor{A} \in
\real^{\pvlength \times \pvlength}$ and $[ \istensor{A} ]_{ij} > 0
\hphantom{k} \forall i,j$ , it has a unique largest (in magnitude)
real eigenvalue, $r \in \real^{+}$, with a corresponding eigenvector,
$\isvector{v} \in \real^{n}$, with strictly positive components, \ie
\begin{equation}
\label{Perron-Frobenius}
A \isvector{v} = r \isvector{v} \hspace{6.69mm} \Longrightarrow \hspace{3.69mm} |\lambda|<r \text{\hphantom{kk}and\hphantom{kk}} \isvector{v}_i > 0 \hphantom{kk} \forall{i \in \{1,\cdots,\pvlength\}} ,
\end{equation}
\noindent where $\lambda$ denotes any eigenvalue of $\istensor{A}$
except $r$, and $r$ is the so-called Perron--Frobenius eigenvalue.
\end{theorem}

\begin{theorem}[Wielandt's theorem~\cite{GRA07}]
\label{Wielandt_theorem}
Given a matrix $\istensor{A} \in \real^{\pvlength \times \pvlength}$
that satisfies the conditions of the Perron--Frobenius theorem (see
Theorem~\ref{Perron-Frobenius_theorem}) and a matrix $\istensor{B} \in
\real^{\pvlength \times \pvlength}$ such as
\begin{equation}
\label{Wielandt}
| b_{ij} | \le a_{ij} \hspace{6.69mm} \forall i,j ,
\end{equation}
\noindent where $b_{ij} = [ \istensor{B} ]_{ij}$ and $a_{ij} = [
  \istensor{A} ]_{ij}$. Then, any eigenvalue $\lambda^{\istensor{B}}$
of $\istensor{B}$ satisfies the inequality $| \lambda^{\istensor{B}} |
\le r$ where $r$ is the Perron--Frobenius eigenvalue of $\istensor{A}$.
\end{theorem}
\begin{theorem}[Lemma 2 in Nikiforov~\cite{NIK07}]
\label{Nikiforov_theorem}
Let $\istensor{A} \in \real^{\pvlength \times \pvlength}$ be an
irreducible non-negative symmetric matrix and $\istensor{R} \in
\real^{\pvlength \times \pvlength}$ be the diagonal matrix of its
rowsums, $[\istensor{R}]_{ii} = \sum_{j=1}^{\pvlength} [ \istensor{A}
]_{ij}$. Then
\begin{equation}
\label{Nikiforov}
\rho \left( \istensor{R} + \frac{1}{\parNiki-1} \istensor{A} \right) \ge \frac{\parNiki}{\parNiki-1} \rho ( \istensor{A} ) ,  
\end{equation}
\noindent with equality holding if and only if all rowsums of
$\istensor{A}$ are equal.
\end{theorem}

To circumvent this problem with the bounds of the spectral radius of
the convective term, $\convc$, we can use the Wielandt's theorem (see
Theorem~\ref{Wielandt_theorem}) to relate the spectral radius of the
matrices
\begin{equation}
\label{diff_C}
2 \convc \equiv \Tcs\traspose \Fs |\Tcs| \text{\hphantom{kkk}and\hphantom{kkk}} {\diffFlux} \equiv - \Tcs\traspose | \Fs | \Tcs,
\end{equation}
\noindent where $\convc$ is the same convective operator defined in
Eq.(\ref{conv_oper3}) and
  ${\diffFlux} \in \real^{\pvlength \times \pvlength}$
  is a diffusive-like operator where the face diffusivities are
  replaced by the magnitude of the mass fluxes, $|\Fs|$. The matrix
  $\convc$ is zero-diagonal whereas the matrix
  ${\diffFlux}$ has strictly negative
  diagonal coefficients. At this point, it is worth noticing that the
  off-diagonal elements of $2 \convc$ (in absolute value) and
  ${\diffFlux}$ are equal. Hence, the
  zero-diagonal matrix
\begin{equation}
\label{diff_C2}
{\diffFluxoff \equiv \diffFlux - \diag(\diag(\diffFlux)) = 2 | \convc |,}
\end{equation}
\noindent satisfies the conditions of the Perron--Frobenius theorem
(see Theorem~\ref{Perron-Frobenius_theorem}). Then, we can apply
Wielandt's theorem (Theorem~\ref{Wielandt_theorem}) since
\begin{equation}
\label{ineq_C_Dcoff1}
{2 | [ \convc ]_{ij} | \le [ \diffFluxoff ]_{ij} \hspace{3.69mm} \forall i,j \hspace{6.69mm} \Longrightarrow \hspace{6.69mm} 2 | \lambda^{\convmat} | \le \rho ( \diffFluxoff ) .}
\end{equation}
\noindent {where $\lambda^{\convmat}$ represents any eigenvalue of the matrix $\convc$}.
In our case, taking {$\istensor{R} = - \diag ( \diag
( \diffFlux ))$}, {$\istensor{A} = \diffFluxoff$}
and $\parNiki=2$ in Eq.(\ref{Nikiforov}) of
Theorem~\ref{Nikiforov_theorem} together with the
inequality~(\ref{ineq_C_Dcoff1}) leads to
\begin{equation}
\label{ineq_C_Dcoff2}
{\rho ( | \diffFlux | ) \stackrel{\text{Thm}~\ref{Nikiforov_theorem}}{\ge} 2 \rho ( \diffFluxoff ) \stackrel{(\ref{diff_C2})}{=} 4 \rho ( | \convc | ) \stackrel{(\ref{ineq_C_Dcoff1})}\ge 4 \rho ( \convc ) .}
\end{equation}
\noindent Recalling that the {\it leitmotiv} for all this analysis was to
avoid constructing the matrix $\convc$, it is obvious that relying on
the construction of another (similar in structure) matrix such
as {$| \diffFlux |$} would not
make much sense. At this point, we can make use of the following
properties of incidence and adjacency matrices {(see
Theorem~\ref{prop_incidence_matrix} in~\ref{appendix_play})}
\begin{align}
\label{InAd_prop1}
| \Tcs\traspose \Tcs | &= | \Tcs\traspose | | \Tcs | , \\
\label{InAd_prop2}
\left| \Tcs\traspose | \Fs | \Tcs \right| &= | \Tcs\traspose | | \Fs | | \Tcs | ,
\end{align}
\noindent to show that
\begin{equation}
\label{C_Dc_equality}
{\rho ( | \diffFlux | ) = \rho \left( \left| \Tcs\traspose | \Fs | \Tcs \right| \right) \stackrel{(\ref{InAd_prop2})}{=}
\rho ( | \Tcs\traspose | | \Fs | | \Tcs | ) \stackrel{\text{Thm}~\ref{same_eigenvalues}}{=} \rho ( | \Tcs | | \Tcs\traspose | | \Fs | ) \stackrel{(\ref{InAd_prop1})}{=}
\rho ( | \Tcs \Tcs\traspose | | \Fs | ) .}
\end{equation}
{Notice that identity~(\ref{InAd_prop1}) is just a particular case of
identity~(\ref{InAd_prop2}) with $\Fs = \Identity$.} {Then,
recalling the inequality~(\ref{ineq_C_Dcoff2}), we can finally show
that $\rho ( \convc )$ can be bounded with $\rho ( | \Fs |^{\alpha}
| \Tcs \Tcs\traspose | | \Fs |^{1-\alpha} )$,~\ie
\begin{equation}
\label{AlgEigConv1}
{\rho ( | \Fs |^{\alpha}   | \Tcs \Tcs\traspose | | \Fs |^{1-\alpha} ) \stackrel{\text{Thm}~\ref{same_eigenvalues}}{=} \rho ( | \Tcs \Tcs\traspose | | \Fs | ) \stackrel{(\ref{C_Dc_equality})}{=} \rho ( | \diffFlux | ) \stackrel{(\ref{ineq_C_Dcoff2})}{\ge} 4 \rho ( \convc ) ,}
\end{equation}
\noindent  regardless of the value of $\alpha$. Let us remind that
indeterminate $1/0$ forms may eventually occur for $\alpha<0$ or
$\alpha>1$ if a mass flux (diagonal terms in $\Fs$) becomes zero.}
\begin{remark}
\label{remark_upwind}
In case the discrete convective term is not skew-symmetric, the method
can be easily adapted as follows: imaginary contributions still come
from $\convc$ whereas negative real-valued contributions are added to
the diffusive term by replacing
\begin{equation}
\pdiffm \longrightarrow \pdiffm + \frac{1}{2} \diag ( | \Fs | ( \vecones - \UPblend ) )
\end{equation}
\noindent where $\UPblend \in \real^{\vvlength}$ is a vector that defines
the local blending factor between
symmetry-preserving {($\scaUPblend=1$)} and upwind
schemes {($\scaUPblend=0$). For details,
see~\ref{appendix_upwind}.}
\end{remark}

\begin{figure}[!t]
  \centering{
    \raisebox{-0.5\height}{\includegraphics[height=0.47\textwidth,angle=0]{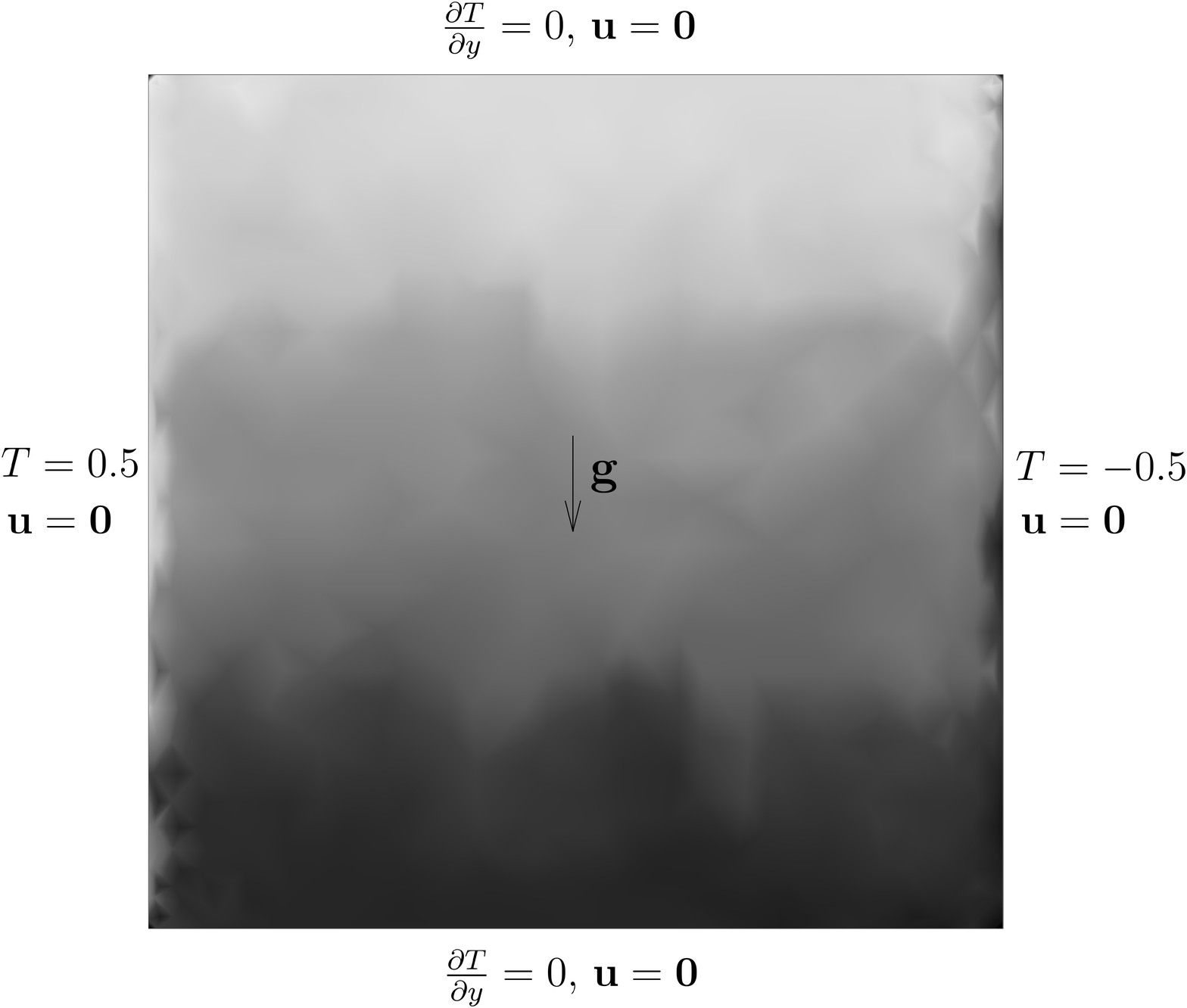}}
    \hspace{3mm}
    \raisebox{-0.5\height}{\includegraphics[height=0.40\textwidth,angle=0]{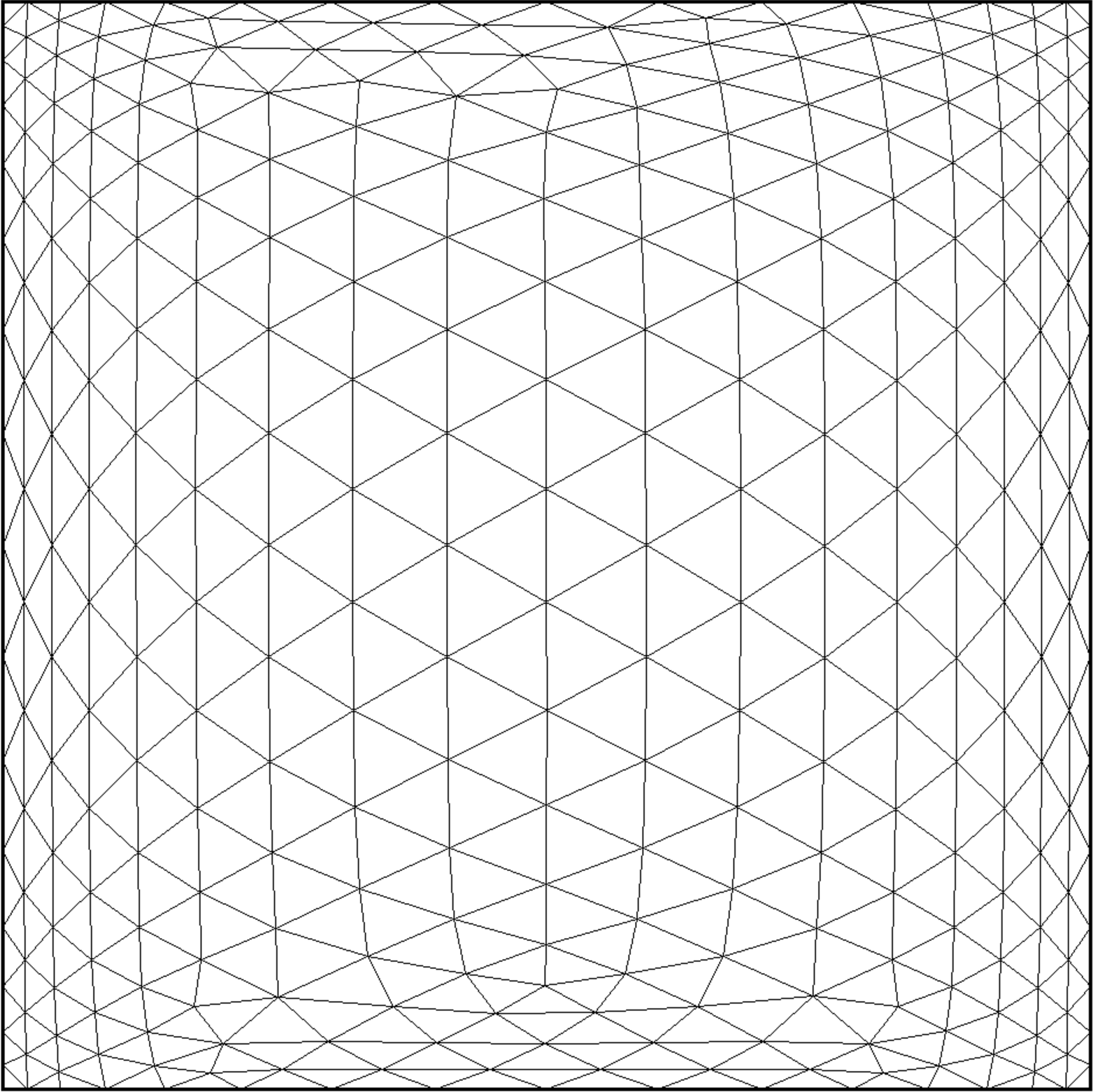}}
  }
  \caption{Two-dimensional air-filled ($\PR=0.71$) differentially
    heated cavity at $\Ra=10^{9}$ in a square domain. Left: schema of
    the flow configuration together with a flow visualization of the
    temperature field corresponding to the statistically steady
    state. Right: unstructured mesh used for the present tests. It is
    composed of $565$ triangular elements stretched to the walls.}
\label{schema_DHC}
\end{figure}

\begin{figure}[!t]
  \centering{
    \includegraphics[height=0.69\textwidth,angle=-90]{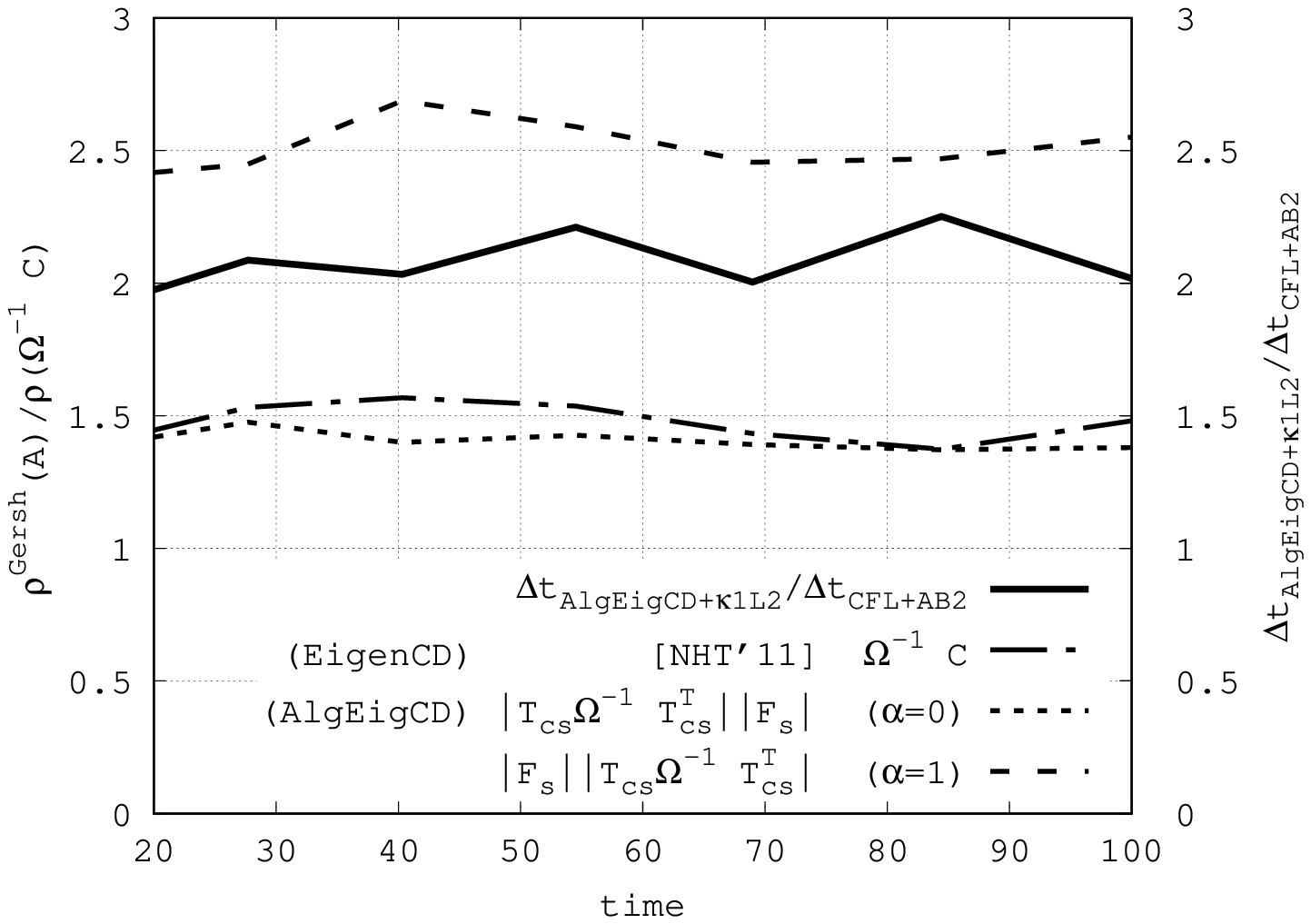}
    \includegraphics[height=0.69\textwidth,angle=-90]{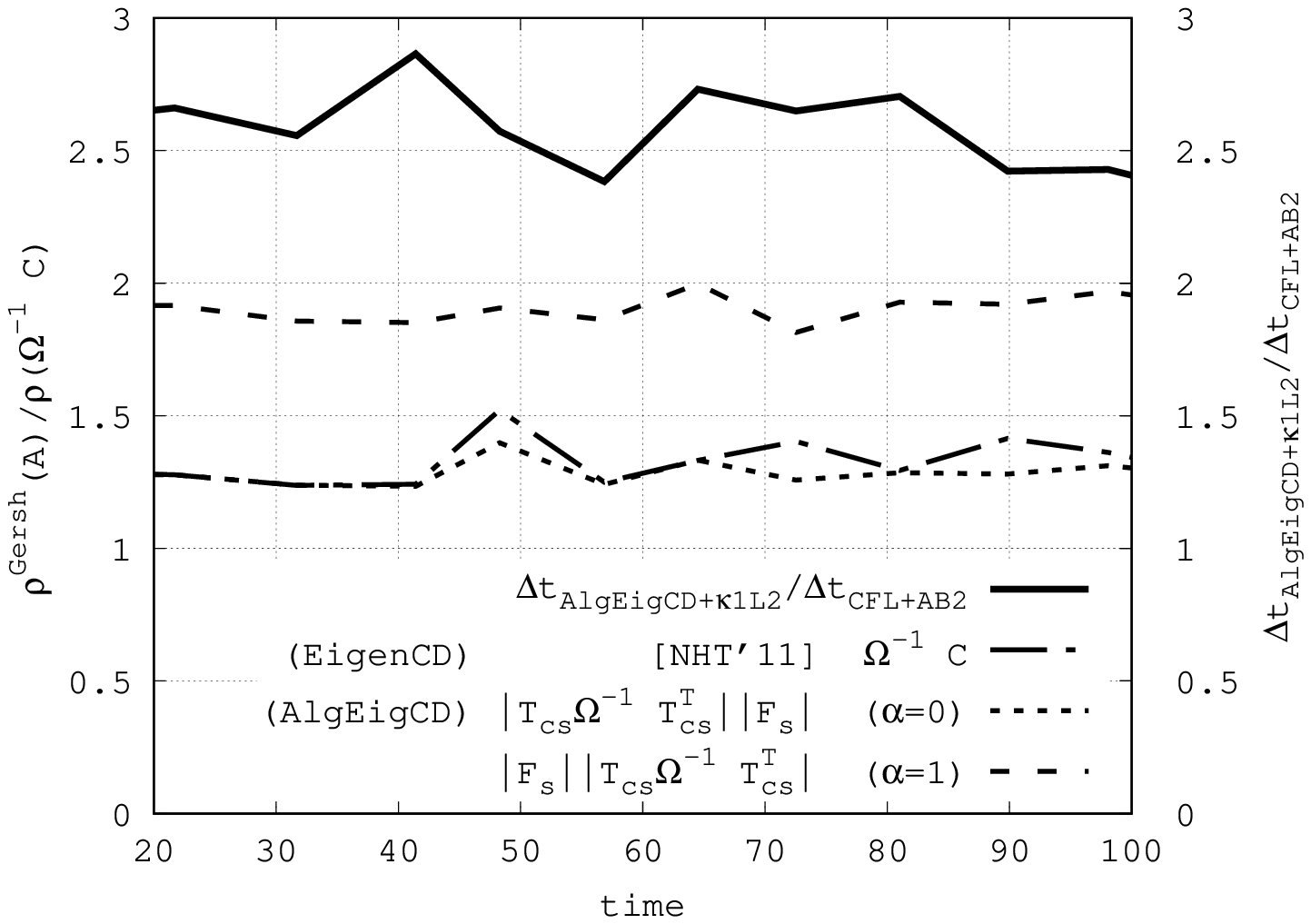}
  }
\caption{Numerical results obtained for the two-dimensional air-filled
  differentially heated cavity displayed in Figure~\ref{schema_DHC}
  using a Cartesian stretched mesh with $23 \times 23 = 529$~control
  volumes (top) and an unstructured mesh composed of $565$~triangular
  elements (bottom). {Notice that for this particular case,
    \ie~estimation of the spectral radius of $\vcvectc^{-1} \conv$
    with a second-order symmetry-preserving discretization, the
    estimations of the \EigenCD~method are exactly the same as those
    given by the discretization-agnostic approach given in
    Eq.(\ref{CFL_ANSYS-Fluent}).}}
\label{results_DHC}
\end{figure}

\begin{table}
\begin{center}
  \resizebox{\columnwidth}{!}{
    \begin{tabular}{l|ccc|ccc|c}
  & $N_x$ & $N_y$ & $N_z$ & $\avgtime{\angle}/{(\pi/2)}$ & $\avgtime{\dt_{\CFLAB}}$ & $\avgtime{\dt_{\NewMeth}}$ & $\frac{\avgtime{{\dt}_{\NewMeth}}}{\avgtime{\dt_{\CFLAB}}}$  \\
\hline
RBC1e8-MeshA++      & $800$  & $416$ & $416$ & $0.251$ & $6.65 \times 10^{-3}$ & $1.29 \times 10^{-2}$  & $1.94$   \\
RBC1e8-MeshA+       & $576$  & $296$ & $296$ & $0.390$ & $1.24 \times 10^{-2}$ & $2.15 \times 10^{-2}$  & $1.73$   \\
RBC1e8-MeshA (DNS)  & $400$  & $208$ & $208$ & $0.499$ & $2.21 \times 10^{-2}$ & $3.59 \times 10^{-2}$  & $1.63$   \\
RBC1e8-MeshB        & $288$  & $144$ & $144$ & $0.606$ & $3.53 \times 10^{-2}$ & $5.84 \times 10^{-2}$  & $1.66$   \\
RBC1e8-MeshC        & $200$  & $104$ & $104$ & $0.696$ & $5.08 \times 10^{-2}$ & $8.49 \times 10^{-2}$  & $1.67$   \\
RBC1e8-MeshD        & $144$  & $76$  & $76$  & $0.777$ & $6.38 \times 10^{-2}$ & $1.20 \times 10^{-1}$  & $1.88$   \\
RBC1e8-MeshE        & $100$  & $52$  & $52$  & $0.852$ & $8.67 \times 10^{-2}$ & $1.88 \times 10^{-1}$  & $2.17$   \\
\hline                                                                                                 
RBC1e10-MeshA (DNS) & $1024$ & $768$ & $768$ & $0.716$ & $4.02 \times 10^{-4}$ & $7.72 \times 10^{-4}$ & $1.92$     \\
RBC1e10-MeshB       & $768$  & $544$ & $544$ & $0.790$ & $5.76 \times 10^{-4}$ & $1.12 \times 10^{-3}$ & $1.94$      \\
RBC1e10-MeshC       & $512$  & $384$ & $384$ & $0.846$ & $8.46 \times 10^{-4}$ & $1.70 \times 10^{-3}$ & $2.01$     \\
RBC1e10-MeshD       & $384$  & $270$ & $270$ & $0.889$ & $1.22 \times 10^{-3}$ & $2.52 \times 10^{-3}$ & $2.06$     \\
RBC1e10-MeshE       & $256$  & $192$ & $192$ & $0.920$ & $1.69 \times 10^{-3}$ & $3.81 \times 10^{-3}$ & $2.25$     
    \end{tabular}
    }
\end{center}
\caption{Tests for the air-filled ($\PR=0.7$) Rayleigh-B\'{e}nard
  convection at Rayleigh numbers $\Ra=10^{8}$ and $10^{10}$ using
  Cartesian meshes stretched towards the walls. Meshes RBC1e8-MeshA
  and RBC1e10-MeshA were respectively used in
  Refs.~\cite{DABTRI15-TOPO-RB,DABTRI19-3DTOPO-RB} to carry out DNS
  simulations.}
\label{meshes_RBC}
\end{table}

\begin{figure}[!t]
  \centering{
    \includegraphics[height=0.47\textwidth,angle=0]{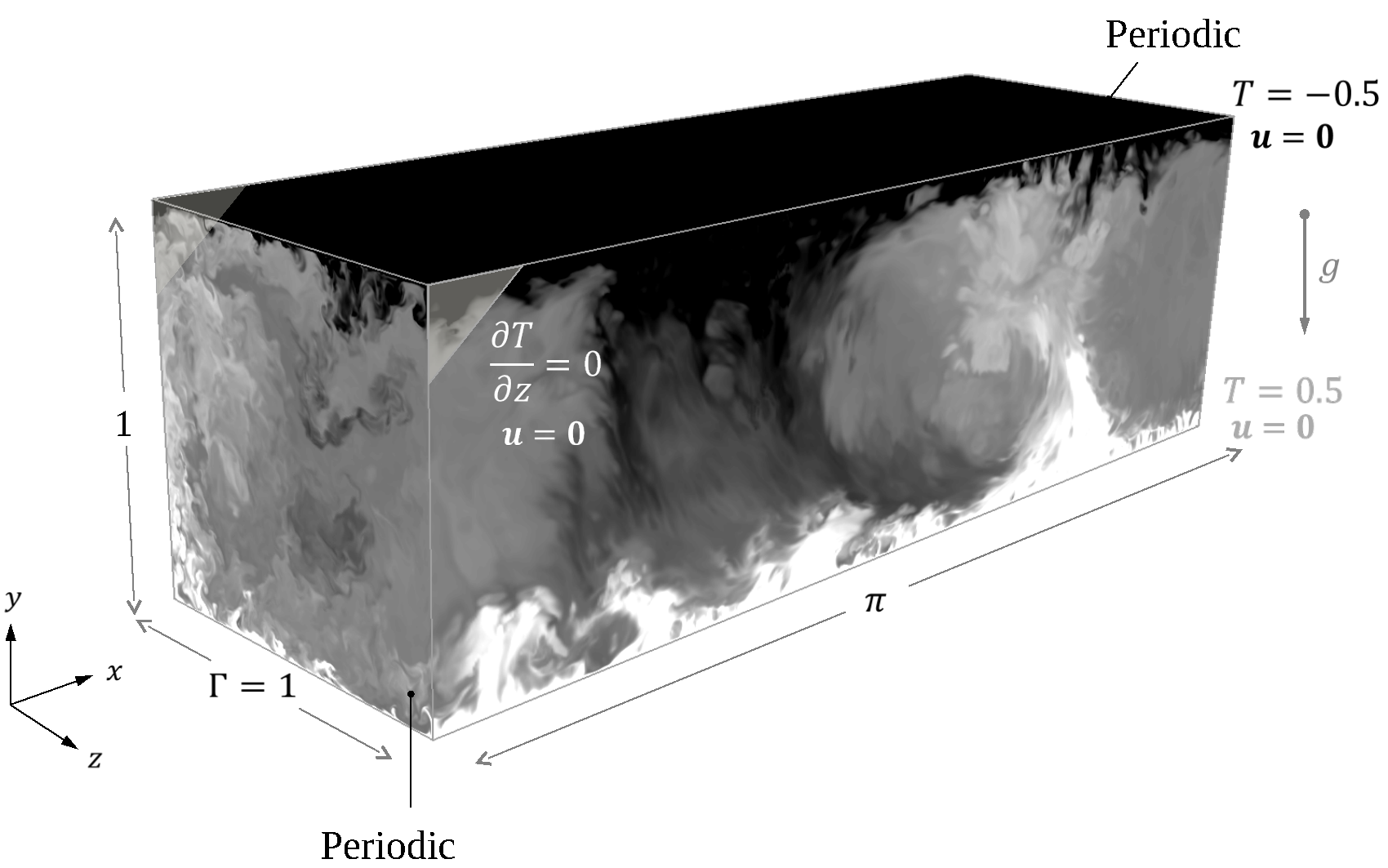}
  }
  \caption{Schema of the Rayleigh-B\'{e}nard configuration studied
    displayed together with an instantaneous temperature field
    corresponding to the air-filled ($\PR=0.7$) DNS (mesh
    RBC1e10-MeshA in Table~\ref{meshes_RBC}) at
    $Ra=10^{10}$~\cite{DABTRI15-TOPO-RB,DABTRI19-3DTOPO-RB}.}
\label{schema_RBC}
\end{figure}

\begin{figure}[!t]
  \centering{
    \includegraphics[height=0.69\textwidth,angle=-90]{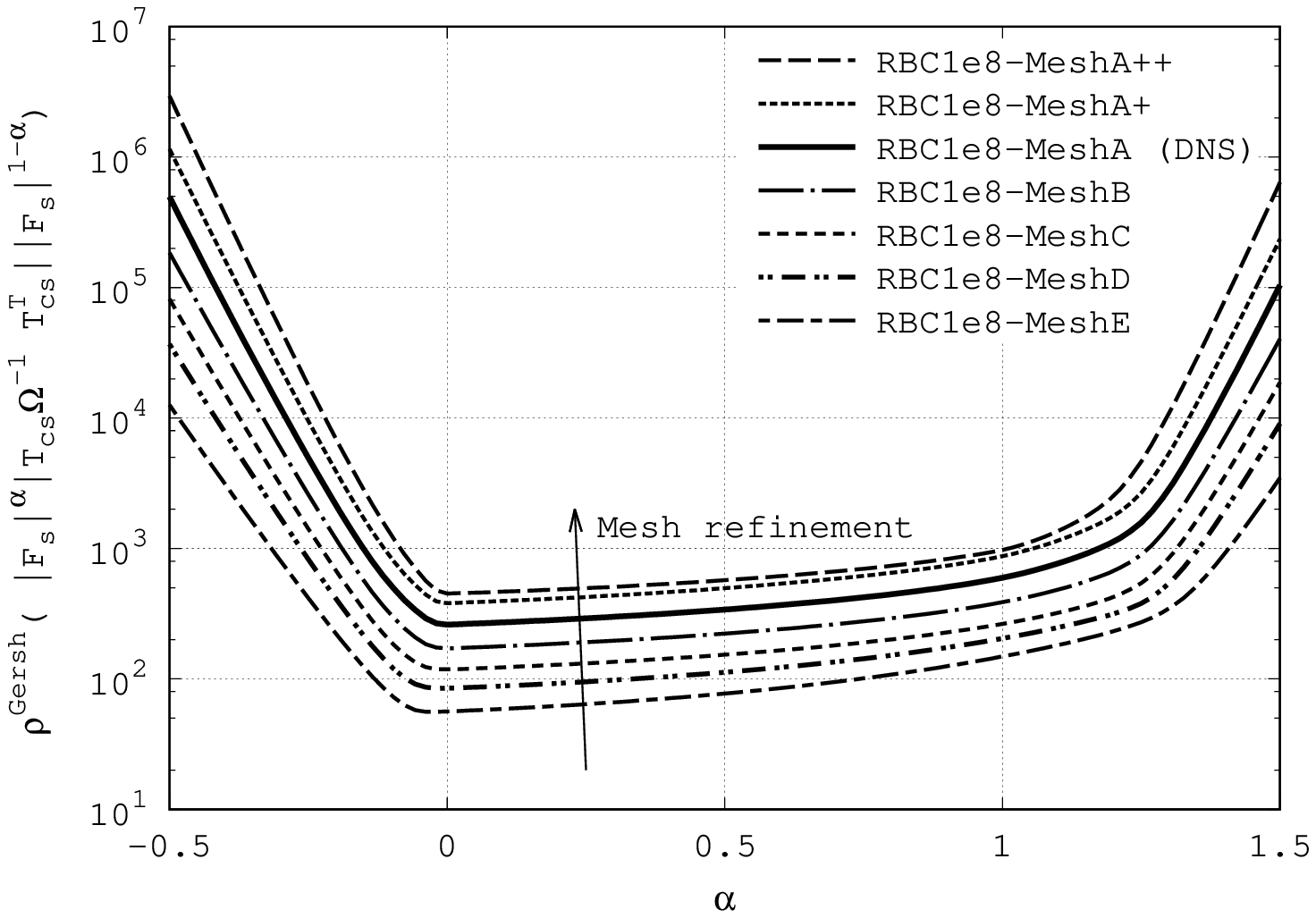}
    \includegraphics[height=0.69\textwidth,angle=-90]{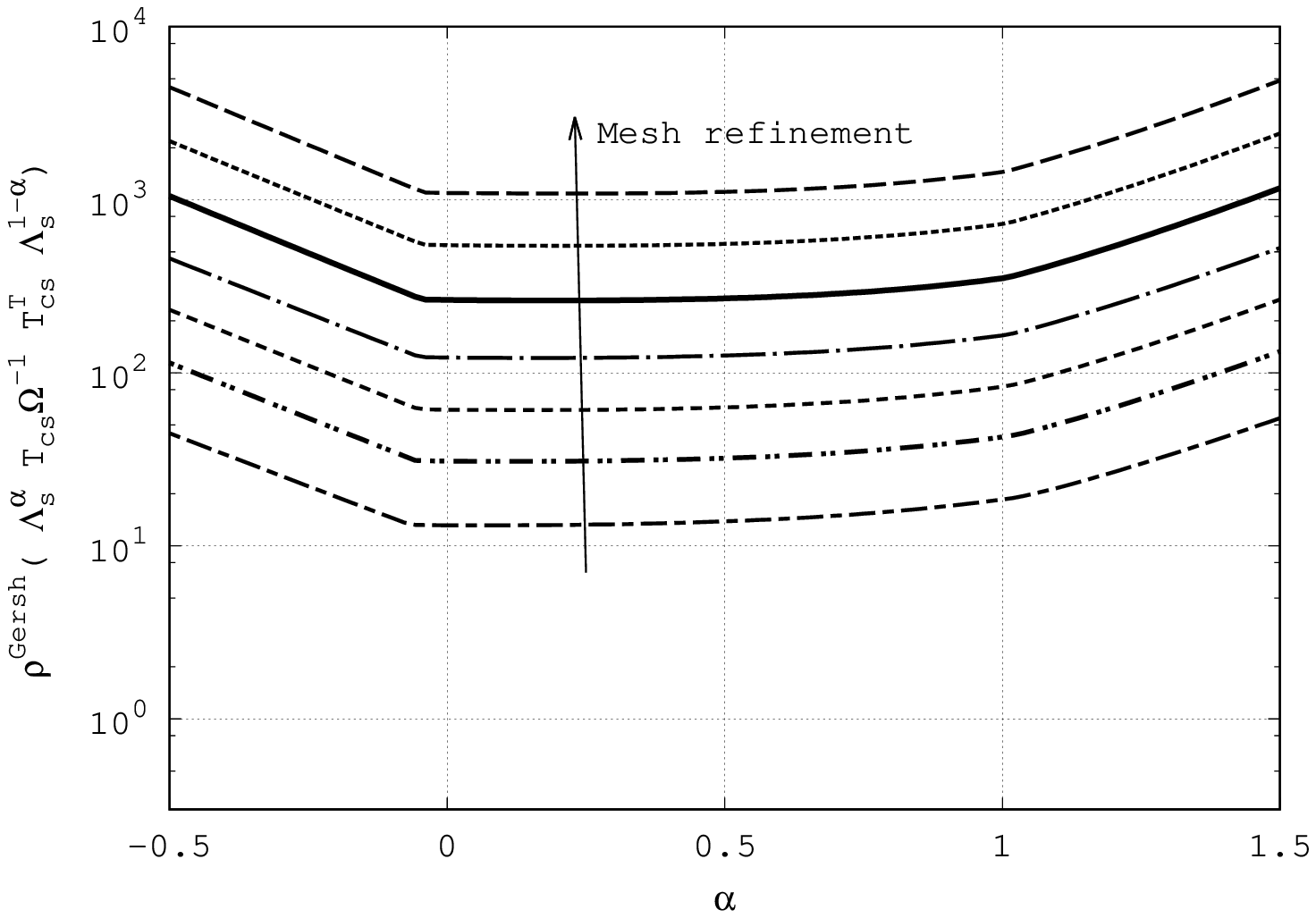}
  }
\caption{Numerical results obtained for the air-filled Rayleigh
  B\'{e}nard configuration displayed in Figure~\ref{schema_RBC} at
  $\Ra=10^{8}$. Eigenbounds for the convective (top) and diffusive
  (bottom) operators using the $\alpha$-dependent expressions given in
  Eqs.(\ref{AlgEigConv2}) and~(\ref{AlgEigDiff2}). Results correspond
  to the statistically steady state and have been averaged over
  time. Details of the meshes are in Table~\ref{meshes_RBC}.}
\label{results_RBC_1e8}
\end{figure}

\begin{figure}[!t]
  \centering{
    \includegraphics[height=0.69\textwidth,angle=-90]{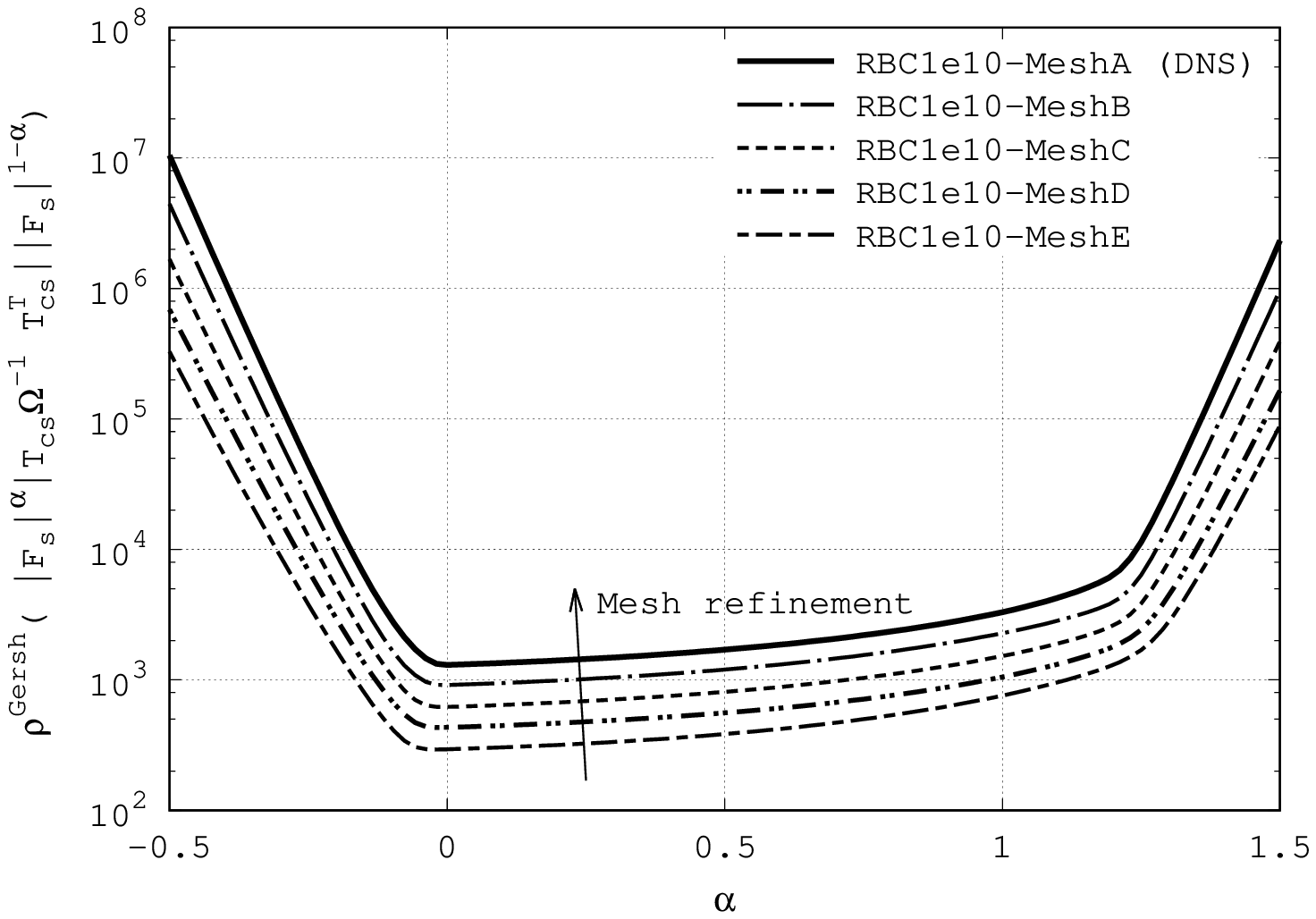}
    \includegraphics[height=0.69\textwidth,angle=-90]{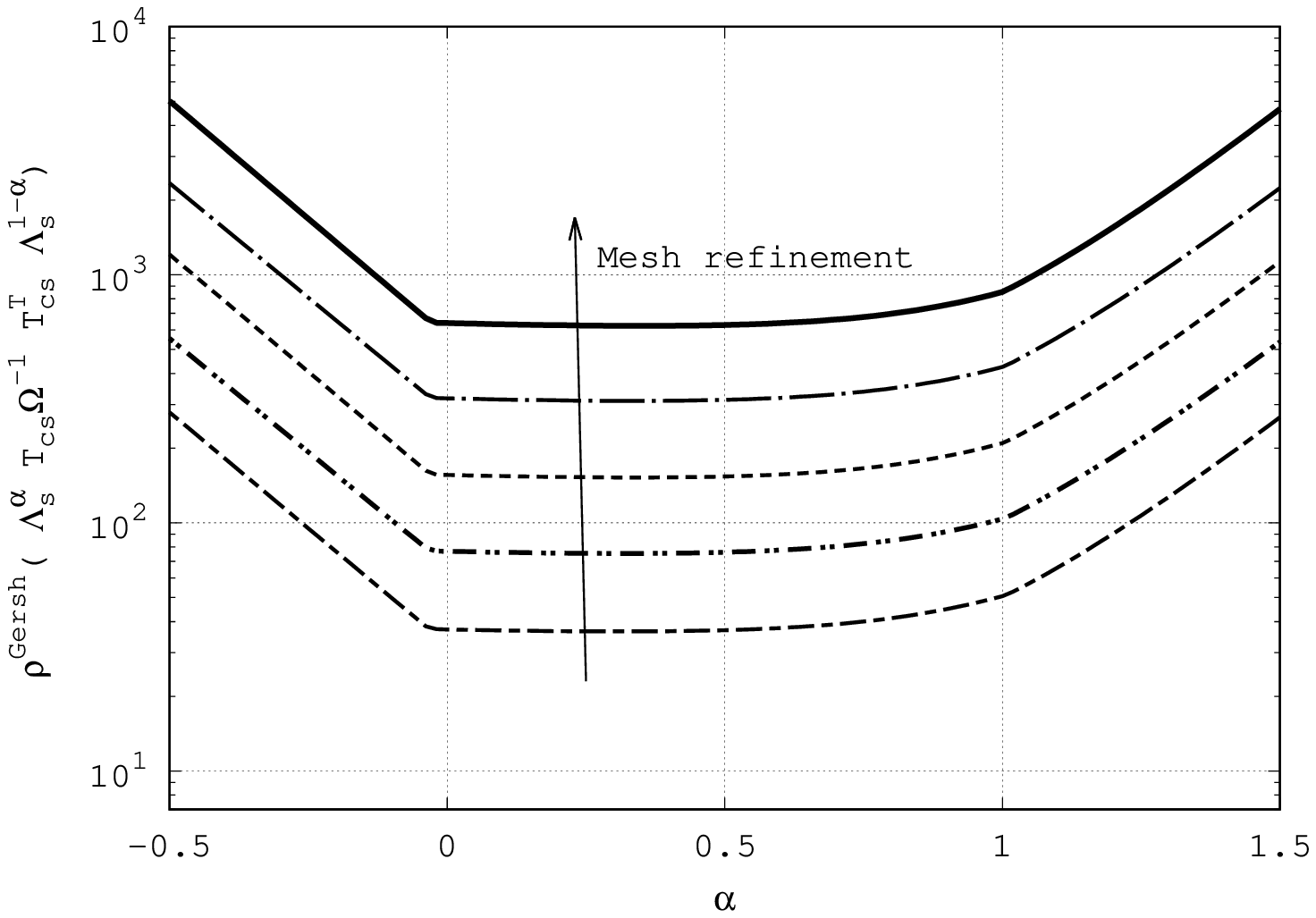}
  }
\caption{Same as in Figure~\ref{results_RBC_1e8} but at $\Ra=10^{10}$.}
\label{results_RBC_1e10}
\end{figure}

\begin{figure}[!t]
  \centering{
    \includegraphics[height=0.69\textwidth,angle=-90]{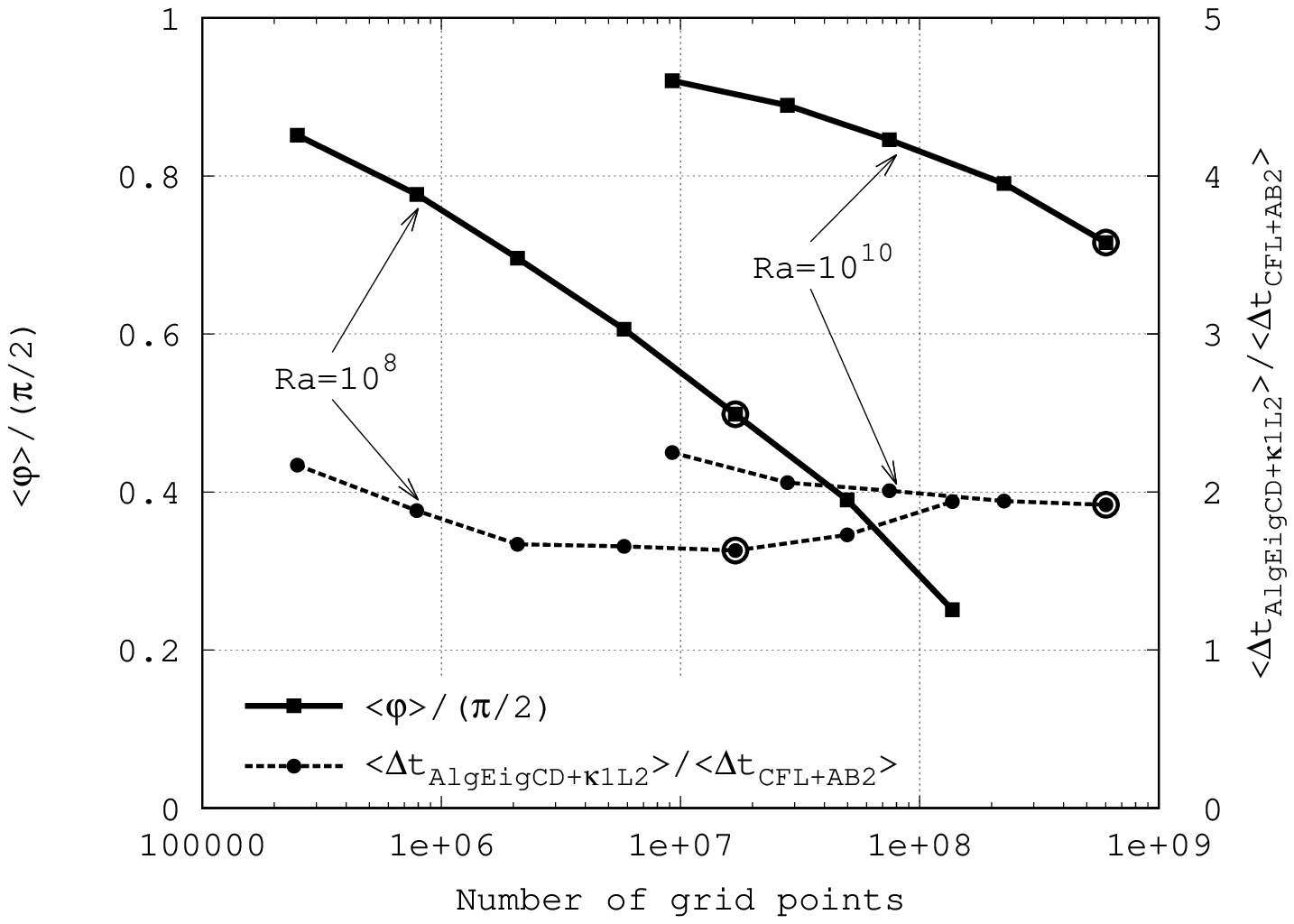}
  }
\caption{Numerical results obtained for the air-filled Rayleigh
  B\'{e}nard configuration displayed in Figure~\ref{schema_RBC} at
  $\Ra=10^{8}$ and $10^{10}$. Average $\angle$ and ratio between
  {the time-step $\Dt^{\NewMeth}$, computed using the
    \AlgEigCD~method in conjunction with the self-adaptive
    $\timeintparam 1L2$ time-integration method
    (see~\ref{SAT_summary}), and $\Dt^{\CFLAB}$ obtained} with the
  more {classical} CFL condition given in
  Eq.(\ref{CFL_code}). Meshes correspond to those shown in
  Table~\ref{meshes_RBC}.}
\label{results_RBC_1e8_1e10}
\end{figure}


\section{Numerical tests}

\label{results}

{Shortly, eigenbounds of convective, $\convc$, and diffusive,
  $\diffc$, matrices can be respectively computed using the
  inequality~(\ref{AlgEigConv1}) and Eq.(\ref{alpha_def2}).  Then, recalling
  Remark~\ref{remark_vcvectc}, eigenbounds of matrices $\vcvectc^{-1}
  \convc$ and $\vcvectc^{-1} \diffc$ can be computed as follows
\begin{align}
\nonumber \rho ( \vcvectc^{-1} \convc ) \stackrel{(\ref{AlgEigConv1})}{\le}& \frac{1}{4} \rho ( | \Fs |^{\alpha} |\Tcs \vcvectc^{-1} \Tcs\traspose | | \Fs |^{1-\alpha} ) \le \frac{1}{4} \rho^{Gersh} ( | \Fs |^{\alpha} |\Tcs \vcvectc^{-1} \Tcs\traspose | | \Fs |^{1-\alpha} ) = \\
\label{AlgEigConv2}                          =& \frac{1}{4} \max \{ \diag ( | \Fs |^{\alpha} ) \Hadprod |\Tcs \vcvectc^{-1} \Tcs\traspose | \diag ( | \Fs |^{1-\alpha} ) \} , \\
\nonumber \rho ( \vcvectc^{-1} \diffc ( \diffv ) ) \stackrel{(\ref{alpha_def2})}{=}& \hphantom{\frac{1}{4}} \rho(\pdiffm^{\alpha} \Tcs \vcvectc^{-1} \Tcs\traspose \pdiffm^{1-\alpha}) \le \rho^{Gersh} (\pdiffm^{\alpha} \Tcs \vcvectc^{-1} \Tcs\traspose \pdiffm^{1-\alpha}) = \\
\label{AlgEigDiff2}                          =& \hphantom{\frac{1}{4}} \max\{ \diag(\hphantom{|}\pdiffm^{\alpha}\hphantom{|}) \Hadprod | \Tcs \vcvectc^{-1} \Tcs\traspose | \diag(\hphantom{|}\pdiffm^{1-\alpha}\hphantom{|}) \} ,
\end{align}
\noindent where $\rho^{Gersh} ( \genmat ) \ge \rho ( \genmat )$ refers
to the eigenbound obtained applying the Gershgorin circle theorem to
matrix $\genmat$. These inequalities rely on the construction of the
matrix $|\Tcs \vcvectc^{-1} \Tcs\traspose |$
(see~\ref{matrix_construction}), which can be done in a pre-processing
stage.}

\mbigskip

In this section, the performance of this methodology is tested and
compared with {our previous \EigenCD~method proposed
  in~\cite{TRI08-JCP2}. Notice that for a second-order
  symmetry-preserving discretization, the estimations of the spectral
  radius of $\vcvectc^{-1} \conv$ given by the \EigenCD~method exactly
  collapses to Eq.~\ref{CFL_ANSYS-Fluent}, which is used in several
  codes (see Section~\ref{CFLhistory}, for details). Apart from this,
  we also compare with} a {classical} CFL criterion
given by
\begin{align}
\label{CFL_code}
\Dt_{CFL} = \min \left\{ \frac{C_{C}}{\lambda_{CFL}^{\convmat}}, \frac{C_{D}}{\lambda_{CFL}^{\diff}} \right\}  \hspace{3mm} \text{where} \hspace{3mm} \lambda_{CFL}^{\convmat} = \max_{\nface} \left\{ \frac{\facevel}{\delta_{\nface}} \right\} \hspace{2mm} \text{and} \hspace{2mm} \lambda_{CFL}^{\diff} = \max_{\nface} \left\{ \frac{{4}\nu_\nface}{\Ndim \delta_{\nface}^2} \right\} ,
\end{align}
\noindent where $\Ndim$ is the number of spatial directions and the
values of $C_{C}$ and $C_{D}$ are set to $0.35$ and 
{0.8}, respectively. These values were used in combination with an
AB2 scheme in the first versions of our {in-house STG
}code~\cite{SORTRI03} {to guarantee that all the eigenvalues lie
  inside the stability region (see Figure~\ref{stab_region}, top)
  regardless of the flow conditions}. {We opt to keep this
  comparison with this less accurate approach since it still remains a
  rather common practice within the CFD community and, as mentioned in
  Section~\ref{CFLhistory}, several popular codes (also many in-house
  codes) are still using it (or very similar approaches).}

\mbigskip

Firstly, we want to measure the actual accuracy of the method to
compute eigenbounds using Eqs.~(\ref{AlgEigConv2})
and~(\ref{AlgEigDiff2}) and compare these results with the exact
spectral radii $\rho(\vcvectc^{-1} \convc)$ and $\rho(\vcvectc^{-1}
\diffc ( \diffv ))$, respectively. The latter are computed using a
singular value decomposition (SVD), which strongly limits the mesh
size. To test this, we have considered a two-dimensional air-filled
(Prandtl number, $\PR=0.71$) differentially heated cavity in a square
domain at Rayleigh number equal to $\Ra=10^{9}$ (for details of this
flow configuration see, for instance,
Refs.~\cite{TRI07-IJHMT-I,TRI07-IJHMT-II}). {Two different meshes
  have been used. Namely, a structured Cartesian mesh with $23 \times
  23 = 529$ control volumes with a stretching towards the walls given
  by the following hyperbolic-tangent function
\begin{equation}
\x_{i} = \frac{L}{2} \left( 1 + \frac{\tanh \left\{ \gammax ( 2 ( i - 1 ) / \Nx - 1) \right\}}{\tanh \gammax}  \right) \hspace{3.69mm} i \in \{ 1, \dots , \Nx + 1 \},
\end{equation}
\noindent where $L$ is the domain size, $\Nx$ is the number of control
volumes in the $\x$-direction and the concentration parameter is set
to $\gammax=1.5$ (also for the $\y$-direction). The second mesh is an
unstructured mesh composed of $565$~triangular elements with an
equivalent stretching (see Figure~\ref{schema_DHC}{, right}).}  It
goes without saying that this mesh resolution is insufficient to
properly resolve this configuration. Nevertheless, it is remarkable
that the symmetry-preserving discretization outlined in
Section~\ref{SymPres} remains stable even with such coarse
meshes. Figure~\ref{results_DHC} displays the results obtained with
these two meshes{ using the in-house
  UMC-code~\cite{TRI08-JCP,TRI12-DmuS}}. Among all the possible values
of $\alpha$ in Eqs.(\ref{AlgEigConv2}) and~(\ref{AlgEigDiff2}) only
results with $\alpha=0$ and $\alpha=1$ are shown together with the
eigenbounds provided by the \EigenCD~method proposed
in~\cite{TRI08-JCP2}. Trends are similar for both cases, showing that
the eigenbounds obtained with $|\Tcs \vcvectc^{-1} \Tcs\traspose | |
\Fs |$ ($\alpha=0$) are significantly better than with $| \Fs | |\Tcs
\vcvectc^{-1} \Tcs\traspose | $ ($\alpha=1$) and slightly better that
those obtained with the \EigenCD~method. Notice that the first method
($\alpha = 0$) is labeled as \AlgEigCD~indicating that this will be
the preferred option at the end. This dependence on $\alpha$ is
analyzed in more detail in the next two test-cases. Furthermore,
Figure~\ref{results_DHC} also shows the ratio between the time-step
$\Dt_{\NewMeth}$ obtained with this \AlgEigCD~method in combination
with the self-adaptive second-order time-integration scheme
$\timeintparam$1L2 (see~\ref{SAT_summary}, for details) and the
time-step $\Dt_{\CFLAB}$ obtained with the CFL condition given in
Eq.(\ref{CFL_code}) in conjunction with the AB2 time-integration
scheme. This is a good measure of the overall benefit of the
methodology. The ratio $\Dt_{\NewMeth}/\Dt_{\CFLAB}$ takes values
around $2$ for the Cartesian stretched mesh (Figure~\ref{results_DHC},
top) whereas slightly higher values are obtained for the unstructured
one (Figure~\ref{results_DHC}, bottom). Similar ratios were already
obtained in Ref.~\cite{TRI08-JCP2} but for the \EigenCD~method.

\mbigskip

The next two tests cases aim to study in more detail the $\alpha$
dependence of Eqs.~(\ref{AlgEigConv2}) and~(\ref{AlgEigDiff2}) and to
test the performance of the method with more realistic
configurations. Firstly, we consider an air-filled ($\PR=0.7$)
Rayleigh-B\'{e}nard configuration at Rayleigh numbers $\Ra=10^{8}$ and
$10^{10}$ using Cartesian meshes stretched towards the walls. A schema
of this configuration together with an instantaneous temperature field
at $\Ra=10^{10}$ is displayed in Figure~\ref{schema_RBC} and the set
of meshes considered here are shown in Table~\ref{meshes_RBC}. Notice
that meshes RBC1e8-MeshA and RBC1e10-MeshA were respectively used in
Refs.~\cite{DABTRI15-TOPO-RB,DABTRI19-3DTOPO-RB} to carry out DNS
simulations at $\Ra=10^{8}$ and $\Ra=10^{10}${ using the
  in-house STG-code~\cite{TRI06-JFM,GORTRI13-PCFD12}}. The reader is
referred to these papers for further details about these
configurations and the criteria used to construct these
meshes. Results analyzing in detail the influence of $\alpha$ in
Eqs.~(\ref{AlgEigConv2}) and~(\ref{AlgEigDiff2}) can be found in
Figures~\ref{results_RBC_1e8} ($\Ra=10^{8}$)
and~\ref{results_RBC_1e10} ($\Ra=10^{10}$) for both the convective
(top) and diffusive (bottom) terms. As expected, for a given
$\Ra$-number, the finer the mesh, the higher the predicted
eigenbounds. However, the most interesting feature of these figures is
the fact that all the plots follow the same trend regardless of the
$\Ra$-number and mesh refinement: namely, there is a clear
over-prediction for $\alpha < 0$ and $\alpha > 1$ whereas within the
range $0 \le \alpha \le 1$ the optimal value of $\alpha$, \ie~the one
that provides the smallest eigenbound, is always located at (or very
close to) $\alpha = 0$. Finally, Figure~\ref{results_RBC_1e8_1e10}
shows how the $\timeintparam$1L2 time-integration scheme dynamically
adapts to flows conditions: the angle $\angle$ (see~\ref{SAT_summary})
is given by
\begin{equation}
\label{angle_def}
\vapfn = \frac{\vapf^{\convmat+\diff}}{\| \vapf^{\convmat+\diff} \|} \hspace{3mm} \Longrightarrow \hspace{3mm} \angle = \tan^{-1} \left( \frac{\imagpart(\vapf^{\convmat+\diff})}{|\realpart(\vapf^{\convmat+\diff})|}  \right), \\
\end{equation}
\noindent where in our case
\begin{align}
{\imagpart(\vapf^{\convmat+\diff})} &= 1/4 \rho^{Gersh} ( |\Tcs \vcvectc^{-1} \Tcs\traspose | | \Fs | ) = \frac{1}{4} \max \{ |\Tcs \vcvectc^{-1} \Tcs\traspose | \diag ( | \Fs | ) \} , \\
{|\realpart(\vapf^{\convmat+\diff})|} &= \hphantom{1/4} \rho^{Gersh} ( \hphantom{|} \Tcs \vcvectc^{-1} \Tcs\traspose \hphantom{|} \hphantom{|} \pdiffm \hphantom{|} ) = \hphantom{\frac{1}{4}} \max\{ | \Tcs \vcvectc^{-1} \Tcs\traspose | \diag(\hphantom{|} \pdiffm \hphantom{|}) .
\end{align}
\noindent Thus, $0 \le \angle/(\pi/2) \le 1$ is a measure of ratio
between the strength of convection respect to diffusion. Therefore,
for a given $\Ra$-number the finer the mesh the stronger the diffusive
term with respect to the convective one. This tendency is clearly
observed in Figure~\ref{results_RBC_1e8_1e10}. Furthermore, we can
also observe that again the ratio
$\avgtime{\Dt_{\NewMeth}}/\avgtime{\Dt_{\CFLAB}}$ takes values close
to~$2$ regardless of the $\Ra$-number and mesh resolution {(exact
  values are shown in the last column of
  Table~\ref{meshes_RBC})}. Here, brackets $\avgtime{\cdot}$ refer to
quantities averaged in time {during the so-called statistically
  steady state.}

\mbigskip

\begin{figure}[!t]
  \centering{
    \includegraphics[height=0.42\textwidth,angle=0]{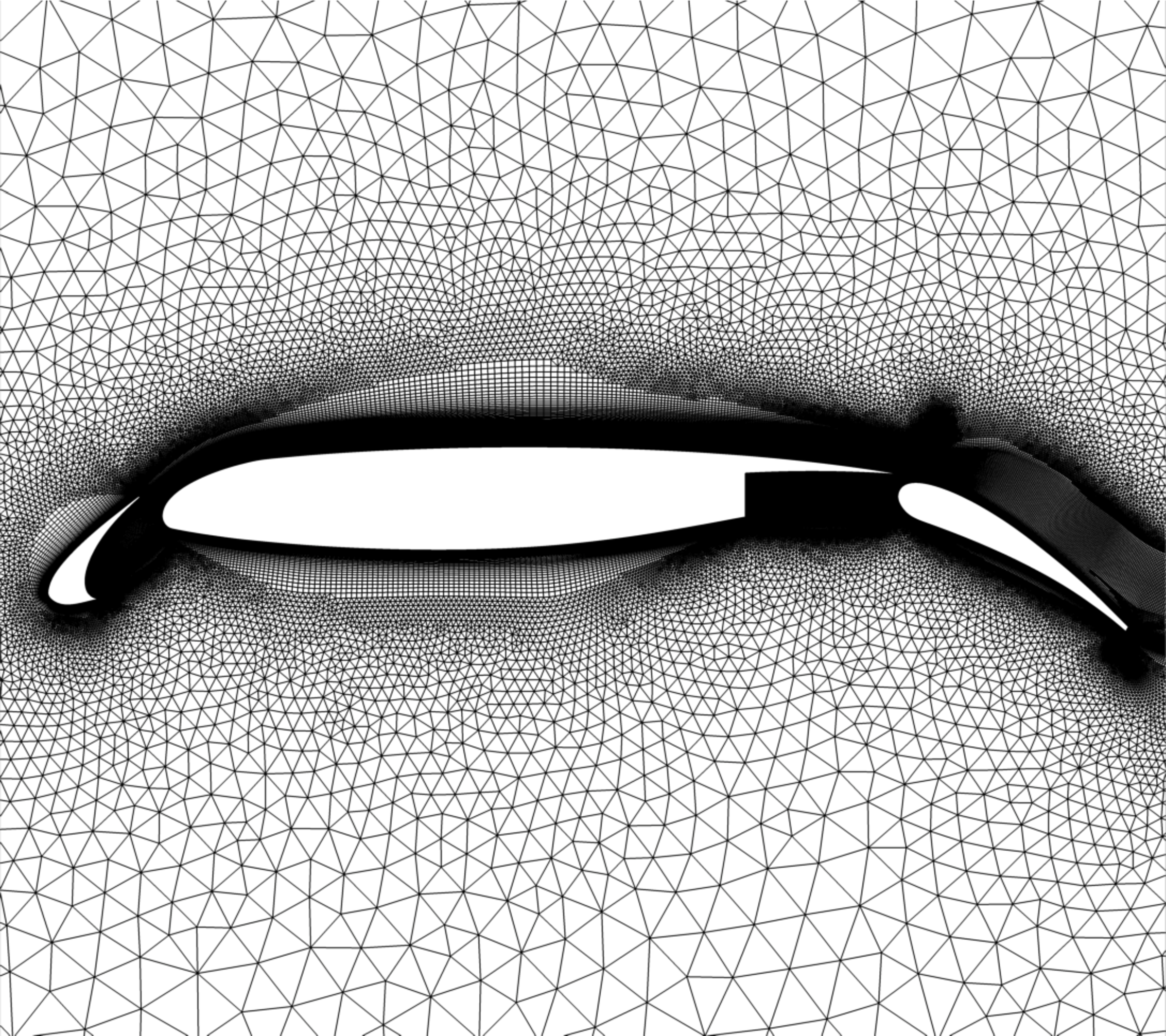}
    \includegraphics[height=0.42\textwidth,angle=0]{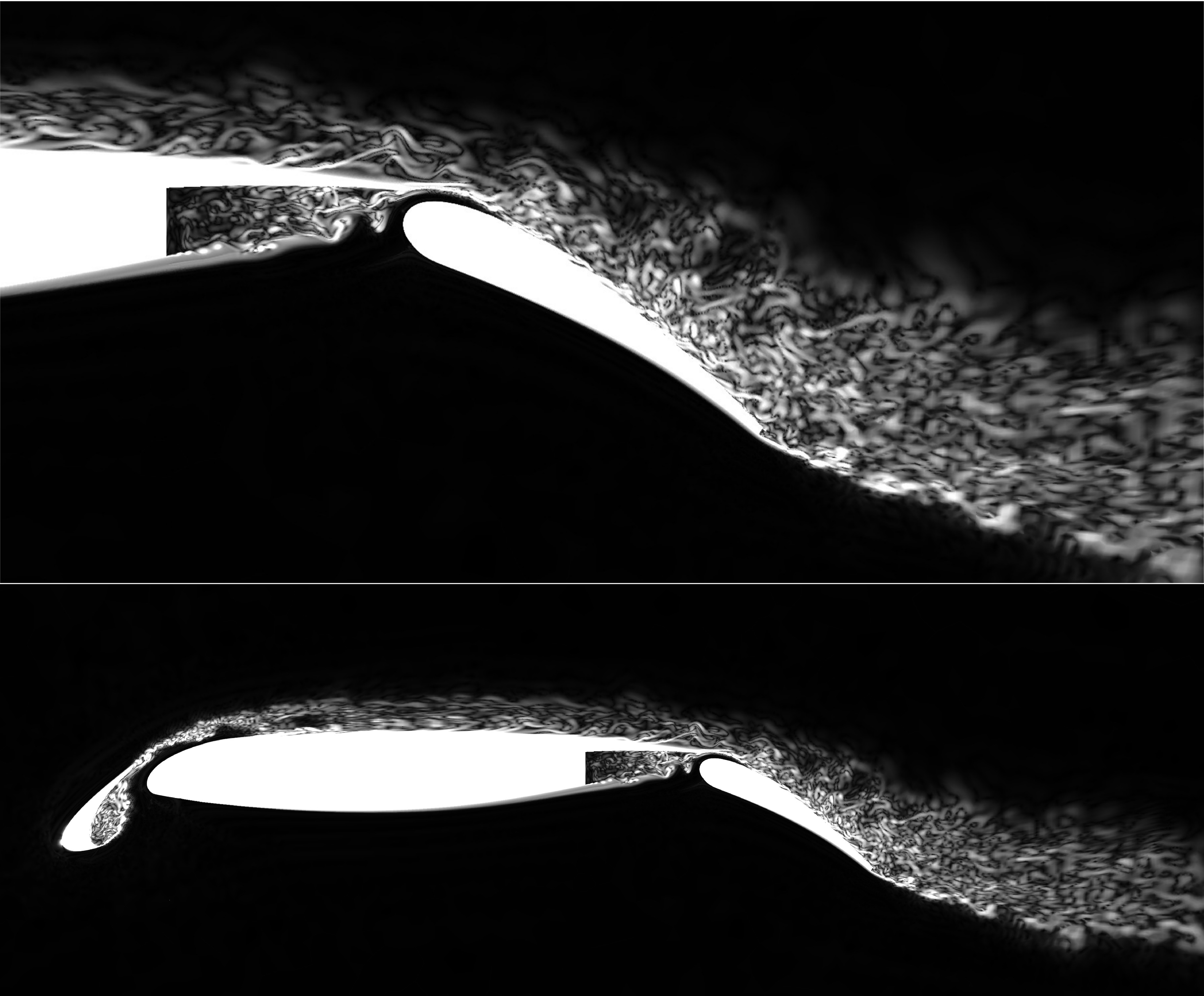}
  }
  \caption{30P30N multi-element high-lift airfoil. Left: zoom around
    the airfoil of the unstructured mesh used for the present
    tests. Right: flow visualization of the vorticity magnitude at
    Reynolds number $10^{6}$ and an angle of attack of $5.5^o$.}
\label{schema_30P30N}
\end{figure}

\begin{figure}[!t]
  \centering{
    \includegraphics[height=0.69\textwidth,angle=-90]{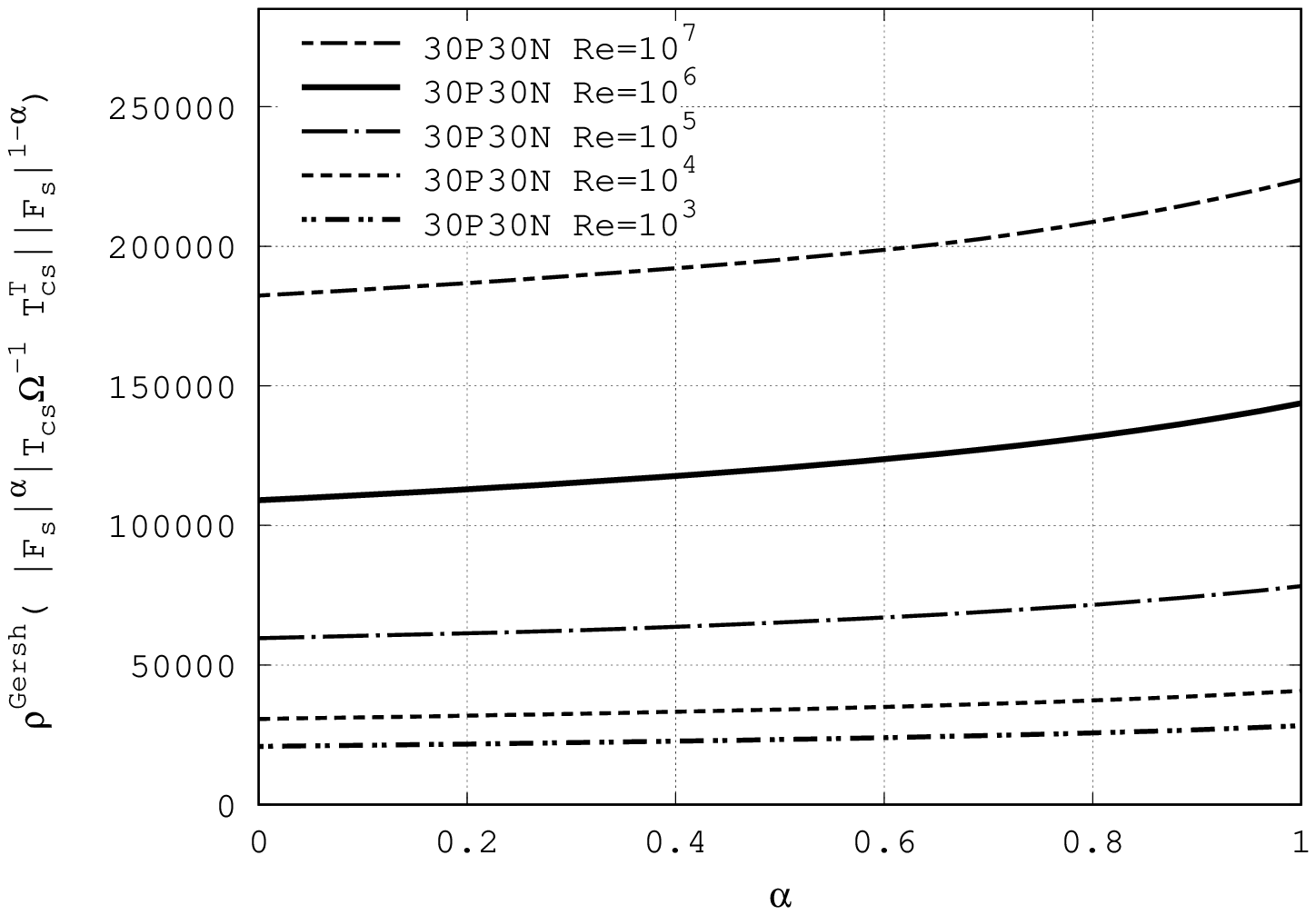}
    \includegraphics[height=0.69\textwidth,angle=-90]{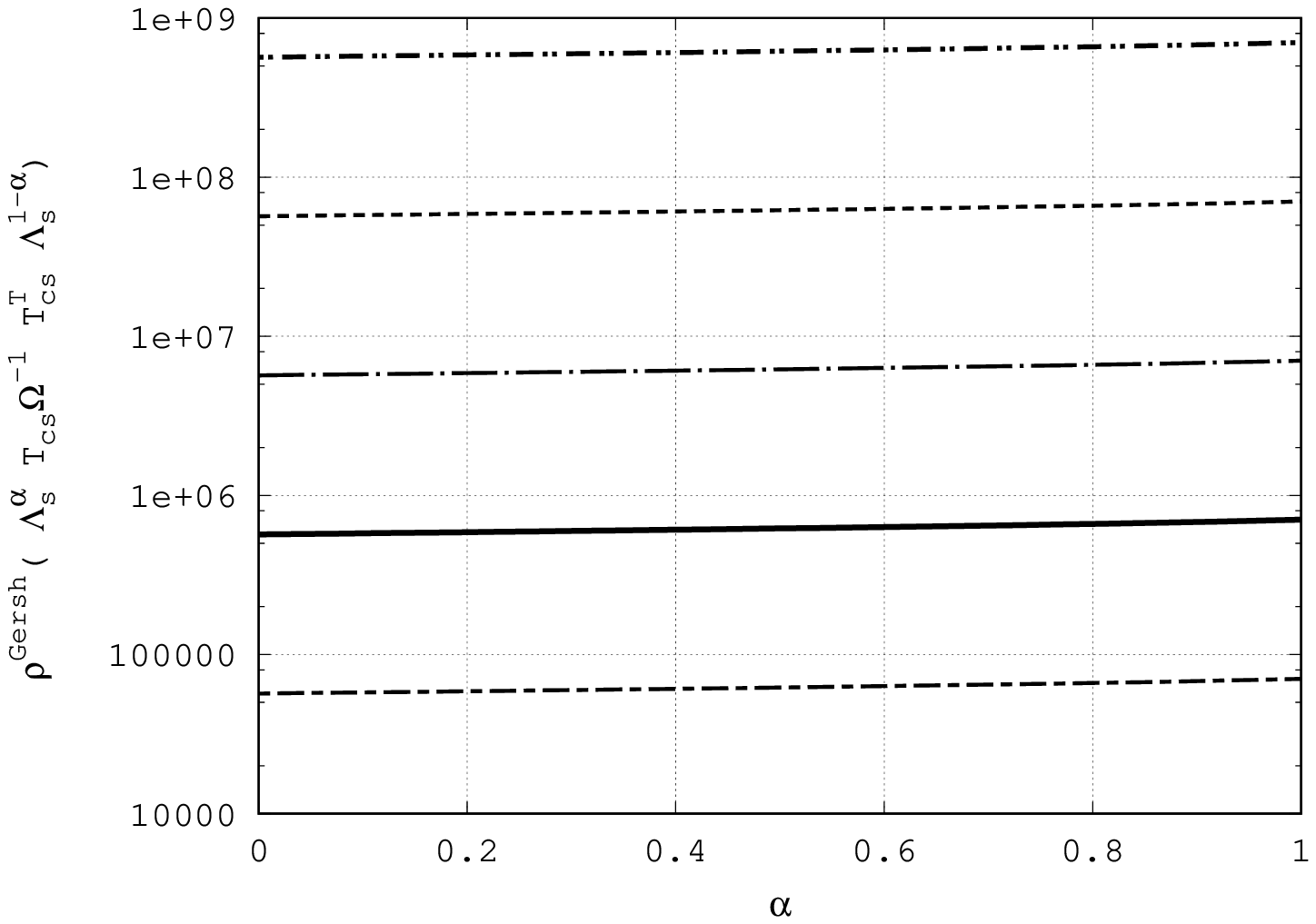}
  }
\caption{Same as in Figures~\ref{results_RBC_1e8}
  and~\ref{results_RBC_1e10} but for the high-lift airfoil 30P30N
  displayed in Figure~\ref{schema_30P30N} at different Reynolds
  numbers using the same mesh.}
\label{results_30P30N}
\end{figure}

\begin{figure}[!t]
  \centering{
    \includegraphics[height=0.69\textwidth,angle=-90]{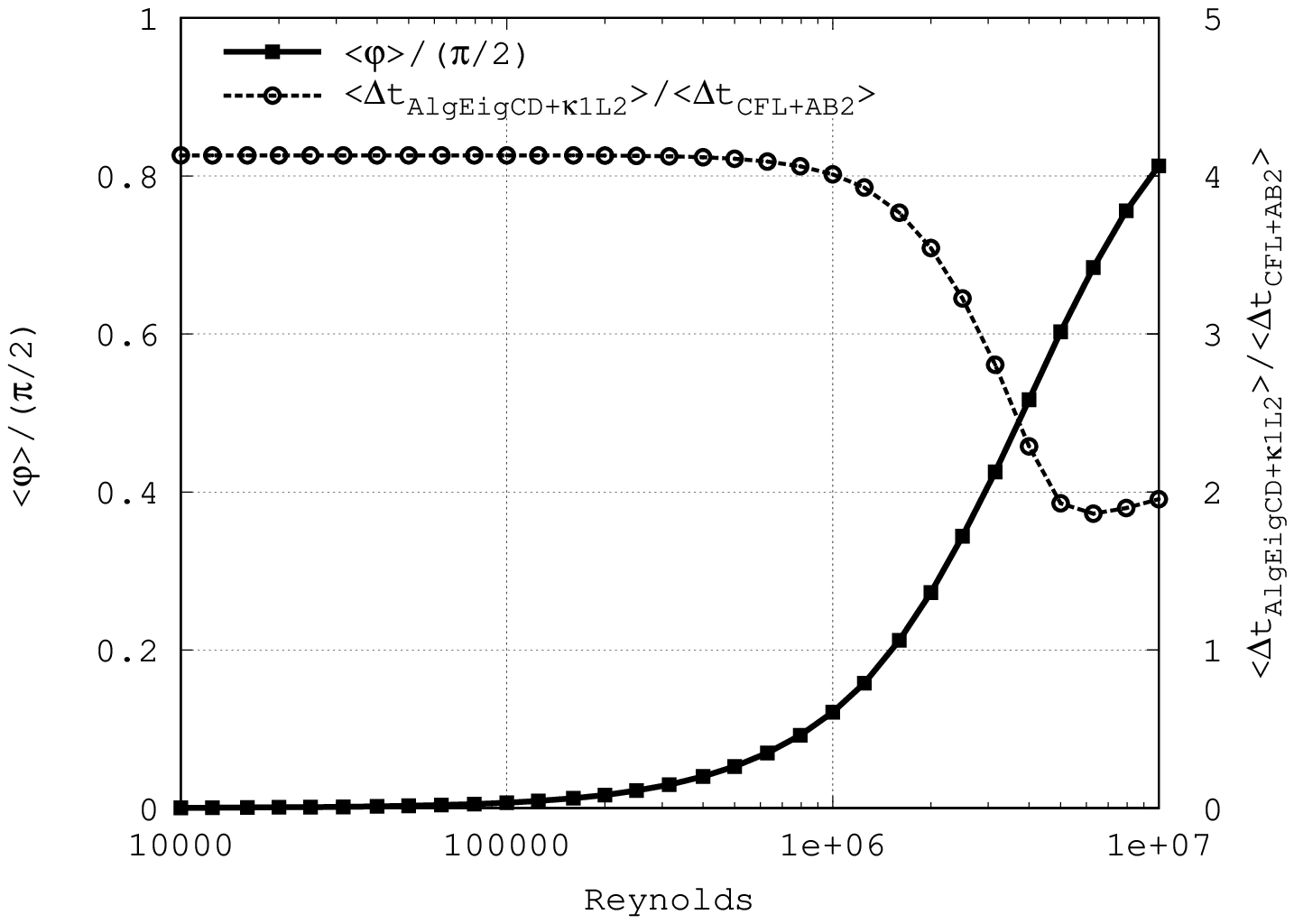}
    \includegraphics[height=0.69\textwidth,angle=-90]{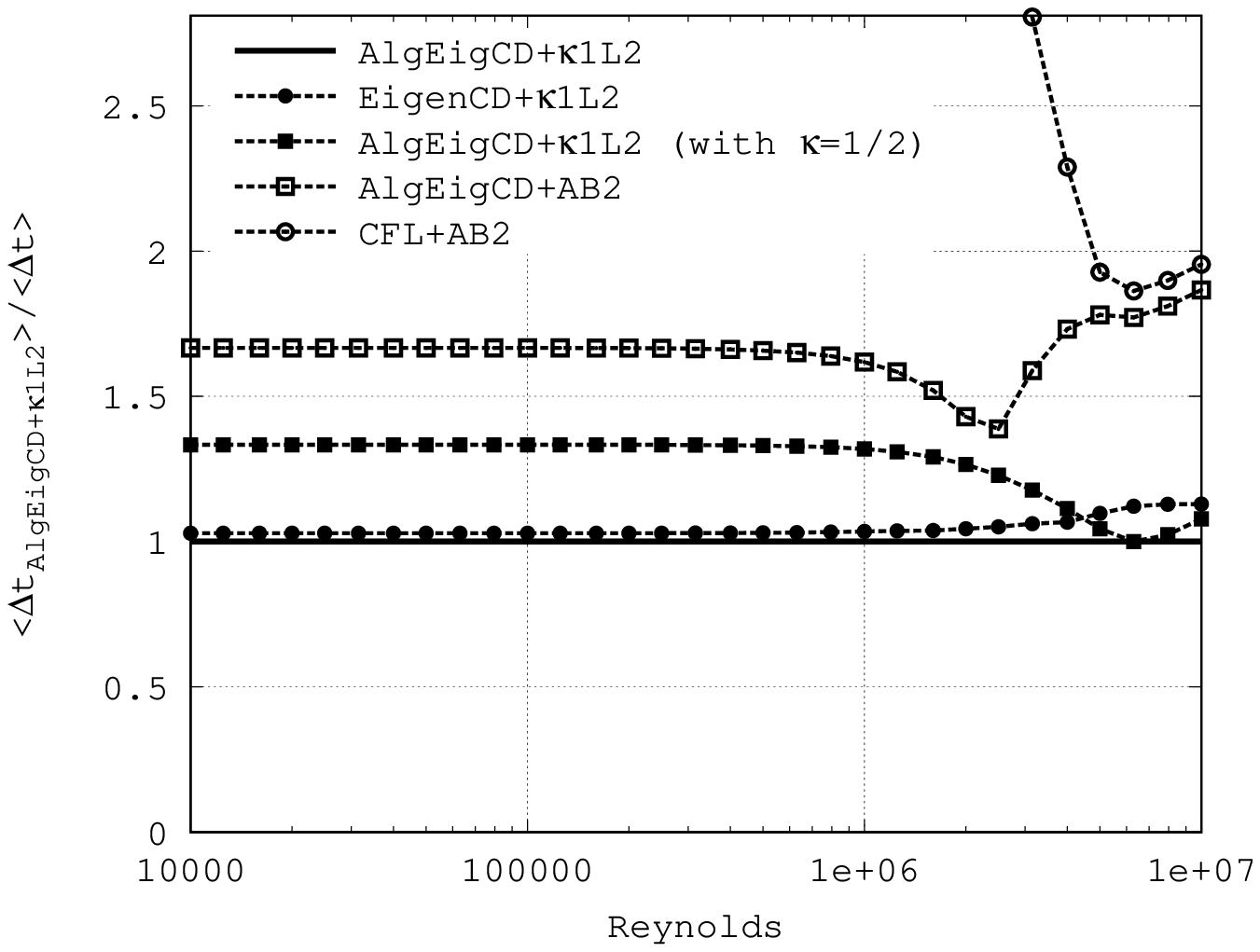}
  }
  \caption{{Top: same} as in
    Figure~\ref{results_RBC_1e8_1e10} but for the high-lift airfoil
    30P30N displayed in Figure~\ref{schema_30P30N} at different
    Reynolds numbers using the same mesh. {Bottom: comparison of
      $\avgtime{\Dt_{\NewMeth}}$ with the $\avgtime{\Dt}$ obtained
      with other three approaches apart from the \CFLAB, which is
      already displayed in the top figure.}}
\label{results_phi_30P30N}
\end{figure}

Finally, the last test-case is a {3D} flow around a 30P30N
multi-element high-lift airfoil at an angle of attack of $5.5^o$ (see
Figure~\ref{schema_30P30N} {and Ref.\cite{PAS16-AIAA} for details
  of the flow}). The mesh is unstructured with $\approx 12.5M$ control
volumes combining hexahedral elements ($\approx 8.2M$) and triangular
prisms ($\approx 4.3M$). {Flow fields have been obtained
  with the in-house NOISEtte code~\cite{GOR22}. }Results displayed in
Figure~\ref{results_30P30N} (top) show the same tendency as in the
previous case confirming that optimal eigenbounds are obtained by
setting $\alpha = 0$ in Eqs.~(\ref{AlgEigConv2})
and~(\ref{AlgEigDiff2}). In this case, the adaptability of the method
has been studied on the same mesh but changing the Reynolds number,
$Re$. Results for a wide range of $Re$ are shown in
Figure~\ref{results_phi_30P30N} (top). As expected for low-$Re$, the
diffusive term is dominant, \ie~$\angle \approx 0$ whereas for (very)
high-$Re$ the convective term becomes the dominant one. Regarding the
ratio $\avgtime{\Dt_{\NewMeth}}/\avgtime{\Dt_{\CFLAB}}$ in this case
it takes values around $4$ for $Re \lesssim 10^{6}$ and goes down to
approximately~$2$ at $Re=10^{7}$. Notice that for the range of
$Re$-numbers, this mesh is designed for (see
Figure~\ref{schema_30P30N}, right), the overall gain in terms of $\Dt$
is approximately $4$. Similar ratios were already observed for the
\EigenCD~method on unstructured meshes~\cite{TRI08-JCP2}. {This
  overall gain results from a combination of factors, which are
  analyzed in detail in Figure~\ref{results_phi_30P30N} (bottom) where
  the ratio $\avgtime{\Dt_{\NewMeth}}$ is compared with the
  $\avgtime{\Dt}$ obtained with other three approaches apart from the
  \CFLAB. Namely, (i) \EigenCD+$\timeintparam 1L2$ is the same as
  \AlgEigCD+$\timeintparam 1L2$ but directly using the Gershgorin
  circle theorem to matrices $\vcvectc^{-1} \convc$ and $\vcvectc^{-1}
  \diffc$. Interestingly enough, the new method provides slightly
  better estimations. Nevertheless, the main advantage respect to the
  $\EigenCD$ method proposed in~\cite{TRI08-JCP2} is that the new
  method does not require to compute the coefficients of the matrix
  and it relies on very simple algebraic kernels, which simplifies its
  implementation and guarantees cross-platform portability. Then, (ii)
  $\NewMeth$ with $\timeintparam=1/2$ consists on using the new
  $\AlgEigCD$ method to compute the eigenbounds of the convective and
  diffusive operators but forcing $\timeintparam=1/2$. Notice that in
  this case, the $\timeintparam 1L2$ and the AB2 schemes have exactly
  the same stability region; however, the method is still using the
  information regarding the location of the eigenvalues in the complex
  plane to find out with is the maximum $\Dt$ that lies inside the
  stability region. In this way, the differences respect to
  $\avgtime{\Dt_{\NewMeth}}$ are only due to the fact that we are not
  allowing the $\timeintparam 1L2$ scheme to self-adapt. Finally, the
  approach (iii) \AlgEigCD+AB2 is basically the same as the
  \CFLAB~method given in Eq.(\ref{CFL_code}) but replacing
  $\lambda_{CFL}^{\convmat}$ and $\lambda_{CFL}^{\diff}$ by the values
  obtained with the new \AlgEigCD~method. Therefore, in this case, the
  differences respect to $\avgtime{\Dt_{\NewMeth}}$ are due to the
  self-adaptivity of the $\timeintparam 1L2$ scheme. Therefore, the
  difference between this last method and the \CFLAB~method can only
  be attributed to the inaccuracy in the computation of
  $\lambda_{CFL}^{\convmat}$ and $\lambda_{CFL}^{\diff}$ in
  Eq.(\ref{CFL_code}). From the results shown in
  Figure~\ref{results_phi_30P30N}, it becomes clear that the
  expression used to compute $\lambda_{CFL}^{\diff}$ is quite
  inaccurate for unstructured grids.}

\section{Concluding remarks}

\label{conclusions}

In summary, the newly proposed \AlgEigCD~method simply relies on the
construction of the matrix $|\Tcs \vcvectc^{-1} \Tcs\traspose |$ which
can be done {at the} pre-processing stage. Then, this
matrix is used to {compute} eigenbounds of matrices
$\vcvectc^{-1} \diff$ and $\vcvectc^{-1} \conv$ as follows
\begin{eqnarray}
\label{AlgEigConv3}
\rho ( \vcvectc^{-1} \convc ) \le& 1/4 \rho ( |\Tcs \vcvectc^{-1} \Tcs\traspose | | \Fs | ) &\le 1/4 \max \{ |\Tcs \vcvectc^{-1} \Tcs\traspose | \diag ( | \Fs | )  \} , \\
\label{AlgEigDiff3}
\rho ( \vcvectc^{-1} \diffc ( \diffv ) ) =&  \phantom{1/4} \rho( \hphantom{|} \Tcs \vcvectc^{-1} \Tcs\traspose \hphantom{|} \hphantom{|} \pdiffm \hphantom{|}) &\le \hphantom{1/4} \max\{ | \Tcs \vcvectc^{-1} \Tcs\traspose | \diag(\hphantom{|}\pdiffm\hphantom{|}) \} ,
\end{eqnarray}
\noindent where the former inequality follows from
Eq.(\ref{AlgEigConv1}) and the application of the Gershgorin circle
theorem to matrix $| \Tcs \Tcs\traspose | | \Fs | $. Similarly, for
the latter and Eq.(\ref{eigenboundD1}). Notice that in these cases,
the diagonal matrix $\vcvectc^{-1}$ has been introduced (see
Remark~\ref{remark_vcvectc}). {Numerical results show an
  improvement with respect to other less general approaches,
  especially for unstructured meshes.} This is an observation that was
already done in Ref.~\cite{TRI08-JCP2} where the \EigenCD~method was
proposed. Nevertheless, the new method is slightly improving the
former one. {All these performance improvements are attributed
  to the fact that the spectral radius is computed in a more accurate
  manner leading to larger time-steps. Indeed, the cost itself of the
  different methods is very similar and just represents a very small
  fraction of the overall simulation.} However, the key elements of
the newly proposed \AlgEigCD~are the fact that no new matrix have to
be re-computed every time-step and that, in practice, only relies on
a~\SpMV~where only the vectors $\diag ( | \Fs | )$ and $\diag( \pdiffm
)$ change on time. Hence, implementation and cross-platform
portability are straightforward. {Moreover, since the same
  matrix (see Eqs.~\ref{AlgEigConv3} and~\ref{AlgEigDiff3}) in used
  for both the convective and all the diffusive terms (notice that,
  apart from the momentum equation, you may have other transport
  equations), all the required~\SpMV's can be computed at once,
  replacing a set of \SpMV’s by a sparse matrix-matrix product (\SpMM)
  which leads to a higher arithmetic intensity and, therefore, a
  better performance~\cite{ALSTRI22-Sym4CFD}.}

\mbigskip

Although the proposed methodology has been deduced in the context of
the symmetry-preserving spatial discretization outlined in
Section~\ref{SymPres}, it can be applied to other schemes resulting
from some sort of blending, \eg~hybrid schemes, flux limiters, etc,
between the symmetry-preserving and the first-order upwind scheme
(see~\ref{appendix_upwind}, for details). {This virtually
  includes all practical low-order (first- or second-order)
  discretizations.} On the other hand, the authors are aware that many
other (higher-order) schemes exists in the CFD literature that indeed
lead to slightly higher spectral radius than lower-order
schemes~\cite{RUATRI19-DIS}. In these cases, an appropriate correction
factor should be introduced to guarantee the stability of the
method. {This correcting factor, which is scheme dependent, can
  be analytically computed using a Fourier analysis as
  in~\cite{RUATRI19-DIS}.} 

\mbigskip

Finally, it worth mentioning that we have plans to extend this method
to other time-integration schemes with larger stability domains and
subsequently larger time-steps. {This raises the
  question regarding the accuracy of the solution. For DNS and LES
  simulations, the time-step computed {\it \`{a} la} CFL is usually
  smaller than the smallest temporal scale of the
  flow~\cite{TRI08-JCP2}. This is the case of all problems analysed in
  this paper. Nevertheless, users of this method (or similar ones)
  must be aware that the stability of the time-integration scheme does
  not guarantee the absence of numerical errors or artifacts that may
  affect the quality of the solution~\cite{PLATRI24DLES14-RKschemes}.}
Another interesting line of research would be combining this approach
with the existing family of symplectic~\cite{SAN13} and
pseudo-symplectic~\cite{CAP17} RK~time-integration methods to get rid
of the artificial dissipation introduced by the temporal schemes. This
is also part of our future research plans.

\section*{Acknowledgments} {F.X.T.}, {X.A-F.}, {A.A-B.} and {A.O.}
are supported by SIMEX project (PID2022-142174OB-I00) of
\emph{Ministerio de Ciencia e Innovaci\'{o}n} and the RETOtwin project
(PDC2021-120970-I00) of \emph{Ministerio de Econom\'{i}a y
  Competitividad}, Spain. {A.A-B.} was supported by the predoctoral
  grants DIN2018-010061 and 2019-DI-90, by
  MCIN/AEI/10.13039/501100011033 and the Catalan Agency for Management
  of University and Research Grants (AGAUR). Calculations were carried
  out on MareNostrum~4 supercomputer at BSC. The authors thankfully
  acknowledge these institutions. Authors also wish to thank
  Dr.~J.~Ruano and Dr.~I.~J\'{o}n\'{a}s for their helpful comments and
  advices.

\appendix

\section{Playing with incidence matrices}

\label{appendix_play}

Let us consider the cell-to-face incidence matrix $\Tcs \in
\real^{\vvlength \times \pvlength}$ which has two non-zero elements
per row (a $+1$ and a $-1$ corresponding to the cells adjacent to a
face) {and the face-to-cell incidence matrix, $\Tsc =
  \Tcs\traspose \in \real^{\pvlength \times \vvlength}$. For instance,
  for the mesh with $4$~control volumes and $8$~faces shown in
  Figure~\ref{mesh} (right), the latter reads}
\begin{equation}
\label{Tsc_def}
\Tsc = \Tcs\traspose =
\left(
\begin{array}{rrrrrrrr}
 0 &  0 & -1 & +1 &  0 &  0 & +1 &  0 \\
+1 &  0 &  0 & -1 &  0 & -1 &  0 &  0 \\
-1 & +1 &  0 &  0 &  0 &  0 &  0 & +1 \\
 0 & -1 & +1 &  0 & +1 &  0 &  0 &  0
\end{array}
\right).
\end{equation}
We want to show the following two properties.
\begin{theorem}
\label{symmetry}
Given a diagonal matrix $\pdiffm \in \real^{\vvlength \times
  \vvlength}$ with strictly positive diagonal values, the matrix
$\Tcs\traspose \pdiffm \Tcs \in \real^{\pvlength \times \pvlength}$ is
symmetric positive semi-definite.
\end{theorem}
\begin{proof}
Symmetry of matrix $\Tcs\traspose \pdiffm \Tcs$ follows
straightforwardly, \ie~$(\Tcs\traspose \pdiffm \Tcs)\traspose =
\Tcs\traspose \pdiffm \Tcs$. {Positive (semi-)definetiness
  follows from the fact that the diagonal coefficients of the matrix
  $\pdiffm$ are strictly positive
\begin{equation}
  \velv_c\traspose \Tcs\traspose \pdiffm \Tcs \velv_c = \velw_s\traspose \pdiffm \velw_s \ge 0 \hspace{6.69mm} \forall \velv_c \in \real^{\pvlength} ,
\end{equation}
\noindent where $\velw_s = \Tcs \velv_c \in
\real^{\vvlength}$. Equality only holds for the unity vector,
$\veconec \in \real^{\pvlength}$, which is the only vector that
belongs to the kernel of $\Tcs$, \ie~$\veconec \in \kernel{\Tcs}$.}
\end{proof}

\begin{theorem}
\label{skew-symmetry}
Given a diagonal matrix $\Fs \in \real^{\vvlength \times \vvlength}$
such as $\diag(\Fs) \in \kernel{\Tcs\traspose}$, the matrix
$\Tcs\traspose \Fs |\Tcs| \in \real^{\pvlength \times \pvlength}$ is
skew-symmetric.
\end{theorem}
\begin{proof}
To prove that the matrix $\genmat \equiv \Tcs\traspose \Fs |\Tcs|$ is
skew-symmetric, we need to show that $\canobas_j\traspose \genmat
\canobas_i = - \canobas_i\traspose \genmat \canobas_j \hphantom{k}
\forall i,j \in \{1,\dots,\pvlength\}$, where $\canobas_{k}$ are
elements of the canonical basis. Firstly, we can compute the
coefficients of $\genmat$ as follows
\begin{equation}
\genmatcoeff_{ik} = t_{ji} f_j | t_{jk} | ,
\end{equation}
\noindent where $\genmatcoeff_{ik} = [ \genmat ]_{ik}$, $t_{jk} =
          [\Tcs]_{jk}$ and $f_j = [ \diag ( \Fs ) ]_j$. Then,
          recalling that the {cell-to-face incidence matrix
            $\Tcs$ has only two non-zero elements per row (and
            $\Tsc=\Tcs\traspose$ per column; see Eq.~\ref{Tsc_def}), a
            $+1$ and a $-1$, we can easily show that the off-diagonal
            elements satisfy}
\begin{equation}
\label{off-diagonal_rel}
\genmatcoeff_{ki} = t_{jk} f_j | t_{ji} | , \hspace{3mm} \genmatcoeff_{ik} = t_{ji}
f_j | t_{jk} | \hspace{3.69mm} \Longrightarrow \hspace{3.69mm}
\genmatcoeff_{ik} = - \genmatcoeff_{ki} \hphantom{kk} \forall i\neq k . 
\end{equation}
\noindent Finally, the diagonal elements of $\genmat$ are given by
\begin{equation}
[\diag(\genmat)]_i = t_{ji} f_j | t_{ji} |  ,
\end{equation}
\noindent which can be re-arranged noticing that $t_{ji} | t_{ji} | =
t_{ji}$,
{
\begin{equation}
[\diag(\genmat)]_i = t_{ji} f_j = 0 \hspace{6.69mm} \forall i \in \{1,\dots,\pvlength\} ,
\end{equation}
\noindent since $\diag(\Fs) \in \kernel{\Tcs\traspose}$. Together with
Eq.(\ref{off-diagonal_rel}), this shows that matrix $\Tcs\traspose \Fs
|\Tcs|$ is skew-symmetric.}
\end{proof}

{
\begin{theorem}
\label{prop_incidence_matrix}
Given a diagonal matrix $\Fs \in \real^{\vvlength \times \vvlength}$,
matrices $\left| \Tcs\traspose | \Fs | \Tcs \right| \in
\real^{\pvlength \times \pvlength}$ and $| \Tcs\traspose | | \Fs | |
\Tcs | \in \real^{\pvlength \times \pvlength}$ are equal
\begin{equation}
\label{InAd_prop2_theorem}
\left| \Tcs\traspose | \Fs | \Tcs \right| = | \Tcs\traspose | | \Fs | | \Tcs | .
\end{equation}
\end{theorem}
\begin{proof}
To prove that $\genmat \equiv \left| \Tcs\traspose | \Fs | \Tcs
\right|$ and $\genmatB \equiv | \Tcs\traspose | | \Fs | | \Tcs |$ are
identical matrices, firstly we compute their coefficients as follows
\begin{align}
\genmatcoeff_{ik} &= | t_{ji} f_j  t_{jk} | , \\
\genmatBcoeff_{ik} &= | t_{ji} | | f_j | | t_{jk} | ,
\end{align}
\noindent where $\genmatcoeff_{ik} = [ \genmat ]_{ik}$,
$\genmatBcoeff_{ik} = [ \genmatB ]_{ik}$, $t_{jk} = [\Tcs]_{jk}$ and
$f_j = [ \diag ( \Fs ) ]_j$. Then, recalling that the cell-to-face
incidence matrix $\Tcs$ has only two non-zero elements per row $j$
(and $\Tsc=\Tcs\traspose$ per column; see Eq.~\ref{Tsc_def}), a $+1$
and a $-1$, it is easy to see that non-zero off-diagonal elements of
matrices $\genmat$ and $\genmatB$ (summation index $j$ which
corresponds to the faces) have only one non-zero contribution that
corresponds to the two cells, $i$ and $k$, adjacent to the face
$j$. Therefore, it follows straightforwardly
\begin{equation}
\label{off-diagonal_rel2}
\genmatcoeff_{ik} = | t_{ji} f_j  t_{jk} | = | t_{ji} | | f_j | | t_{jk} | = \genmatBcoeff_{ik} \hphantom{kkk} \forall i\neq k . 
\end{equation}
\noindent Finally, the diagonal elements are given by
\begin{equation}
[\diag(\genmat)]_i = | t_{ji} f_j  t_{ji} | = | f_j t_{ji}^2 | = | t_{ji} | | f_j | | t_{ji} | = [\diag(\genmatB)]_i.
\end{equation}
Together with Eq.(\ref{off-diagonal_rel2}), this shows that matrices
$\left| \Tcs\traspose | \Fs | \Tcs \right|$ and $| \Tcs\traspose | |
\Fs | | \Tcs |$ are equal.
\end{proof}
}


\section{Dealing with upwinding schemes}

\label{appendix_upwind}

Let us firstly consider a first-order upwind scheme~\cite{PAT80}. In
this case, the interpolation operator, $\Sscal$, needed to construct
the convective matrix (see Eq.~\ref{conv_oper}) is given by a row-wise
linear combination of $1/2|\Tcs|$ (cell-to-face unweighted
interpolation) and $1/2 \Tcs$ (difference between adjacent cell
values). Then, the sign matrix in front of $1/2 \Tcs$ depends of the
flow direction, \ie~on the sign of mass fluxes, $\Fs$. Hence, the
first-order upwind scheme reads as follows
\begin{equation}
\convUP \equiv \dive \Us \Sscal^{\UP}  \hspace{6.69mm} \text{where} \hspace{6.69mm} \Sscal^{\UP} \equiv \frac{1}{2} | \Tcs | + \frac{1}{2} \sign(\Us) \Tcs .
\end{equation}
\noindent This can be re-written in terms of mass fluxes across faces,
$\Fs$, and the cell-to-face incidence matrix, $\Tcs$, 
\begin{equation}
\convUP \equiv \frac{1}{2} \Tcs\traspose \Fs | \Tcs | + \frac{1}{2} \Tcs\traspose \Fs \sign(\Fs) \Tcs ,
\end{equation}
\noindent where $\sign(\Fs) \in \real^{\vvlength \times \vvlength}$
results into a diagonal matrix containing the signs of $\Fs$. This
expression can be further simplified 
\begin{equation}
\convUP = \convc - \frac{1}{2} \diffc(\diag(|\Fs|)) ,
\end{equation}
\noindent where $\convc$ and $\diffc$ are respectively the
symmetry-preserving discretization of the convective term given in
Eq.(\ref{conv_oper3}) and the discrete diffusive operator given in
Eq.(\ref{diff_oper3}).

\mbigskip

Nevertheless, in many occasions upwind scheme is blended with
symmetry-preserving scheme, \eg~hybrid schemes, flux
limiters~\cite{VALALVTRI21-FL},... At the end, this blending between
symmetry-preserving {($\scaUPblend=1$)} and upwind
{($\scaUPblend=0$)} can be defined in terms of a {vector
  $\UPblend \in \real^{\vvlength}$ defined at the faces. It can be
  seen as a second input parameter for the convective operator}
\begin{equation}
\convcarg{\velh,\UPblend} \equiv \convc - \frac{1}{2} \diffc(|\Fs| ( \vecones - \UPblend ) ) .
\end{equation}
\noindent Then, the eigenvalues of $\convcarg{\velh,\UPblend}$ can be
bounded using Bendixson theorem (see Theorem~\ref{Bendixson_theorem})
as follows: imaginary contributions come from $\convc$ whereas
negative real-valued contributions can simply be added to the
diffusive term by replacing
\begin{equation}
\pdiffm \longrightarrow \pdiffm + \frac{1}{2} \diag ( | \Fs | ( \vecones - \UPblend ) ) .
\end{equation}


\section{Self-adaptive time-integration scheme}

\label{SAT_summary}

For the sake of completeness, this appendix shortly revises the
self-adaptive second-order time-integration scheme $\timeintparam$1L2,
which was originally proposed in Ref.~\cite{TRI08-JCP2}. Given the
first-order ordinary differential equation
\begin{equation}
\frac{\ud \sca}{\ud t} = f ( \sca ) ,
\end{equation}
\noindent it is discretized in time as follows
\begin{equation}
\label{oneleg2}
\left( \timeintparam + \frac{1}{2} \right) \sca^{n+1} - 2 \timeintparam \sca^{n} + \left( \timeintparam - \frac{1}{2} \right) \sca^{n-1} = \dt f \left( (1+\timeintparam)\sca^n - \timeintparam \sca^{n-1} \right) .
\end{equation}
\noindent Then, assuming that $\dt$ is small enough, the non-linear
function $f ( \cdot )$ can be linearized, \ie~$f ( x ) \approx \vapf
x$ where $\vapf \in \complex$. In this way, the problem becomes
\begin{equation}
\label{oneleg4}
\left( \timeintparam + \frac{1}{2} \right) \sca^{n+1} - 2 \timeintparam \sca^{n} + \left( \timeintparam - \frac{1}{2} \right) \sca^{n-1} = 
\vapfn \dtn \left( (1+\timeintparam)\sca^n - \timeintparam \sca^{n-1} \right) ,
\end{equation}
\noindent where $\vapfn$ is the unitary vector $\vapfn = \vapf / \|
\vapf \| \in \complex$ with $\angle \in [ -\pi/2 , \pi/2 ]$ and $\dtn
= \dt \| \vapf \| \in \real^{+}$. It can be viewed as a generalization
of the classical second-order Adams--Bashforth scheme
($\timeintparam=1/2$) where the parameter $\timeintparam$ is used to
adapt the region of stability to the instantaneous flow conditions in
order to maximize $\Dt$. The idea of the method is depicted in
Figure~\ref{stab_region} (bottom). To keep the time-integration method
stable, we need that the eigenvalues, $\vapg$, of the amplification
matrix $T$
\begin{equation}
\label{oneleg5}
\left( 
\begin{array}{c}
\sca^{n+1} \\ \sca^{n}
\end{array}
\right) = T
\left( 
\begin{array}{c}
\sca^{n} \\
\sca^{n-1}
\end{array}
\right)
\hspace{3mm}
\text{with} \hspace{2mm} 
T= \left(
\begin{array}{cc}
A (\timeintparam , \vapfn \dtn ) & B  (\timeintparam , \vapfn  \dtn ) \\
1 & 0
\end{array}
\right) ,
\end{equation}
\noindent to be smaller than unity, $\| \vapg \| \le 1$. The complex 
functions $A$ and $B$ are given by $A ( x , y ) = ( 2x + xy + y ) / (
x + 1/2)$ and $B (x,y) = - (x + xy - 1/2) / ( x + 1/2)$,
respectively. Therefore, the two eigenvalues of the linear system are
given by
\begin{equation}
\vapg = \frac{1}{2} \left( A \pm \sqrt{A^2 + 4 B} \right) .
\end{equation}
\noindent Hence, the idea of the method reads: given a $\angle$, to
determine which is the $\timeintparam$ that leads to the maximum
$\dtn$ possible (see Algorithm~\ref{k1L2_algorithm}).
\begin{algorithm}[!t]
\caption{Self-adaptive time-integration scheme $\timeintparam$1L2 }
\begin{algorithmic}[1]
  \Require $\vapf \in \complex$ where $\realpart(\vapf) \le 0$. Note: in our case, $|\realpart(\vapf)| \ge \rho( \vcvectc^{-1} \diff ( \diffv ))$ and $\imagpart ( \vapf ) \ge \rho ( \vcvectc^{-1} \conv )$
  \Ensure $\timeintparam$, $\dt$
\State \label{k1L2_algorithm_step1} Determine $\angle$, where $\vapfn = \vapf / \| \vapf \|$, {\ie~$\angle = \tan^{-1} ( \imagpart(\vapf) / |\realpart(\vapf)| )$.}
\State \label{k1L2_algorithm_step2} Find $\timeintparam = \Kopt ( \angle )$ and $\dtn = \Topt ( \angle )$ {using Eqs.(\ref{Topt_eq}) and~(\ref{Kopt_eq}).}
\State \label{k1L2_algorithm_step3} Advance on time the problem (\ref{oneleg2}) with $\dt = \dtn / \| \vapf \|$ and $\timeintparam$.
\end{algorithmic}
\label{k1L2_algorithm}
\end{algorithm}
Note that since the stability domain is always symmetric with respect
to the real axis we can restrict the analysis to the range $\angle \in
[ 0 , \pi/2 ]$. Thus, the only thing that remains is to determine the
exact form of the functions $\Kopt(\cdot)$ and $\Topt(\cdot)$
{(see step~\ref{k1L2_algorithm_step2} in
  Algorithm~\ref{k1L2_algorithm})}. This was numerically found in
Ref.\cite{TRI08-JCP2} (Figure~\ref{KTopt} displays the form of these
functions). In order to provide an easy-to-implement method, they can
be approximated by means of piece-wise polynomial functions. For
$\Topt ( \angle )$, the following approximation was proposed
\begin{figure}[!t]
\centerline{
\includegraphics[width=63mm,angle=-90]{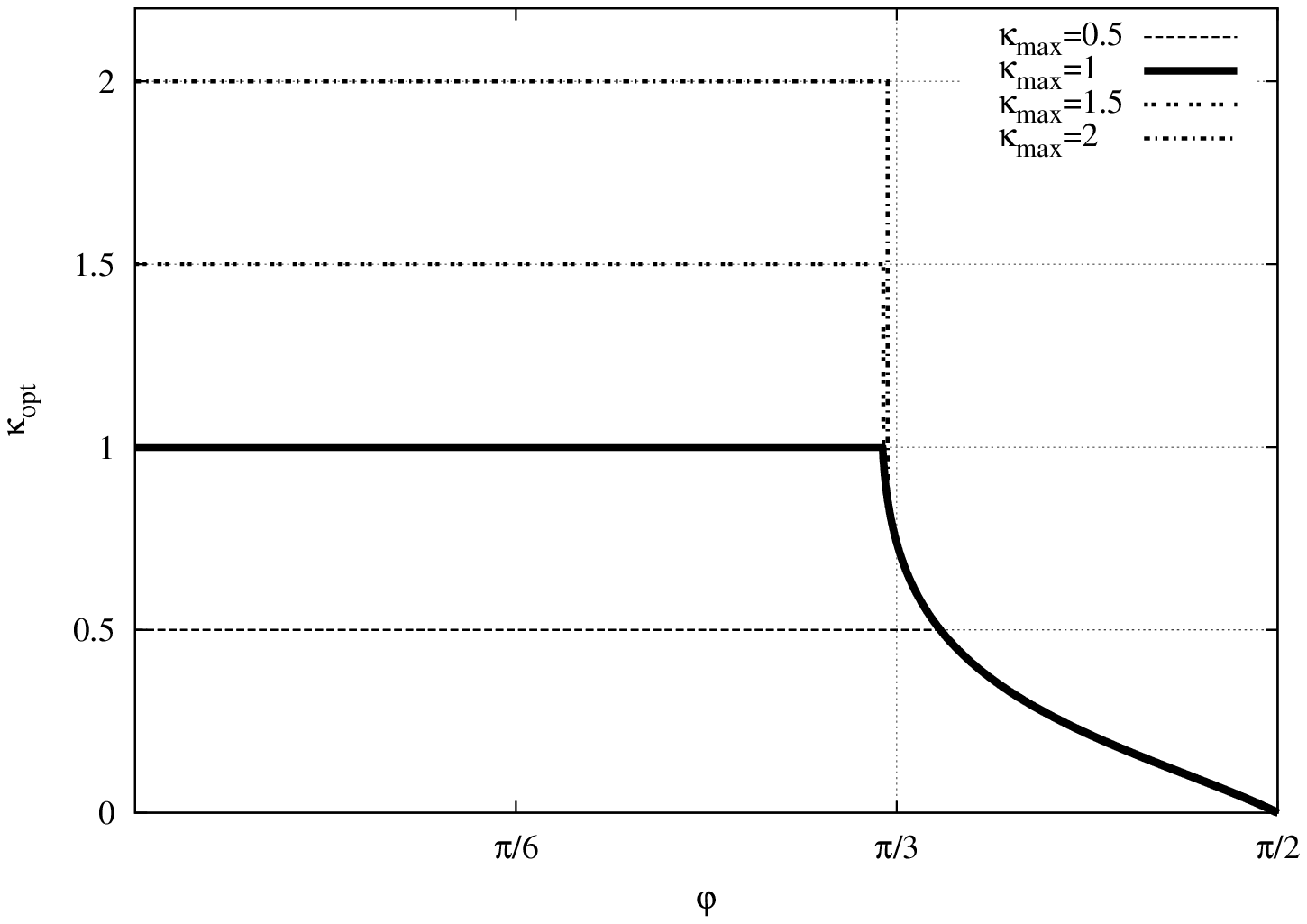}
}
\centerline{
\includegraphics[width=63mm,angle=-90]{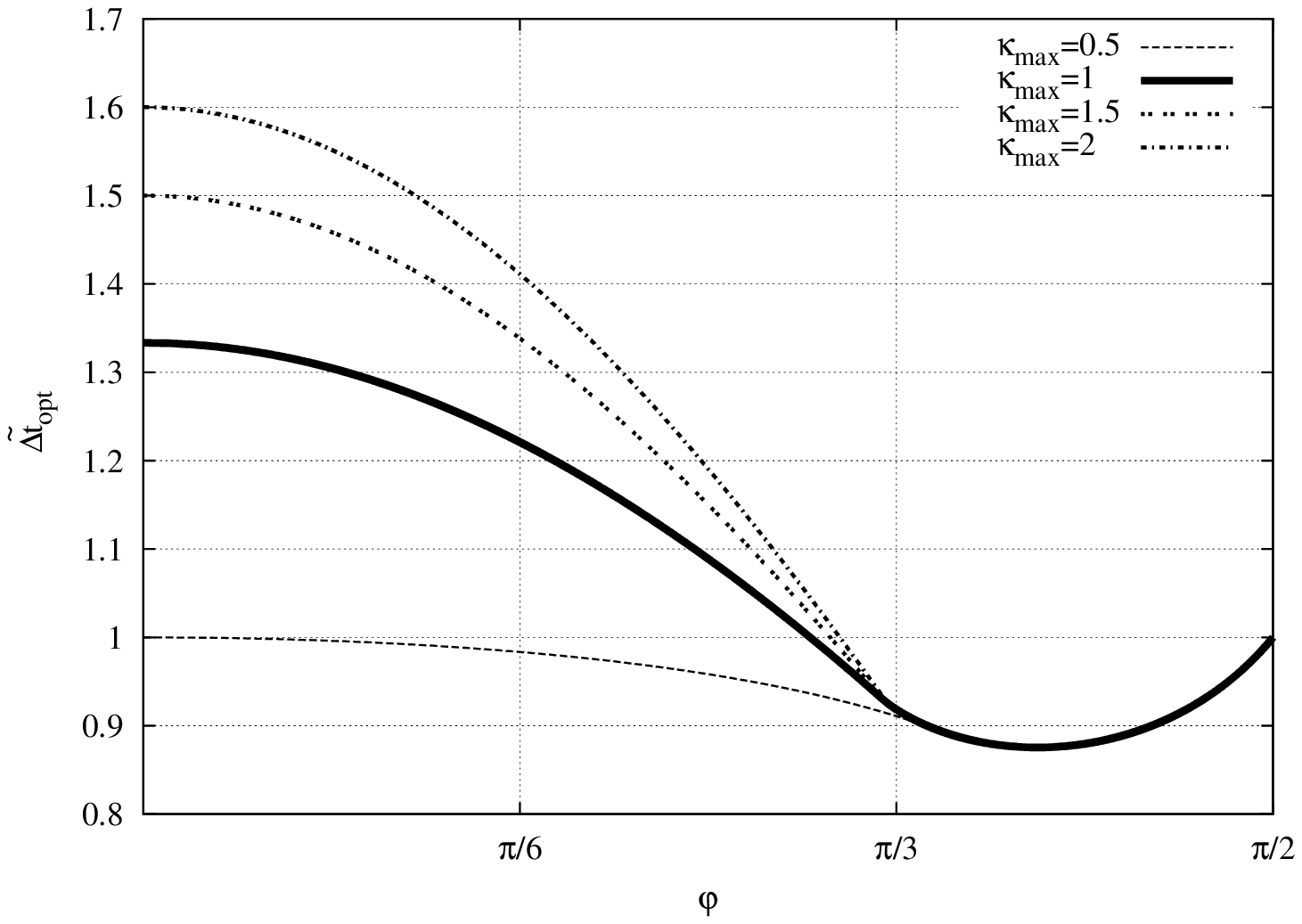}
}
\caption{Functions $\Kopt$ (top) and $\Topt$ (bottom) as a function of 
$\angle \in [ 0 , \pi/2]$ for different values of
$\timeintparam_{max}$ from $0.5$ to $2$. The chosen option is
$\timeintparam_{max} = 1$ (solid line). Note that for values
$\timeintparam_{max} \gtrsim 1$, $\Kopt (\angle)$ is not a continuous
function.}
\label{KTopt}
\end{figure}
\begin{equation} 
\label{Topt_eq}
\Topt ( \angle ) \approx
\left\{
\begin{array}{ll}
G(\angle,0\hphantom{_1},c_1,c_2,0,\angle_1,4/3,t_1)     & \hspace{3mm} \text{if} \hspace{5mm} 0 \le \angle < \angle_1 \\
G(\angle,c_3,c_4,c_5,\angle_1,\pi/2,t_1,1) & \hspace{3mm} \text{if} \hspace{5mm} \angle_1\le \angle \le \pi/2 
\end{array}
\right.
\end{equation}
\noindent where $\angle_1 = \tan^{-1} (164/99)$, $t_1=0.9302468$, and
the function $G$ is a piece-wise quartic interpolation of the form
\begin{equation}
G ( x,a,b,c,x_0,x_1,f_0,f_1) = ( a x^2 + b x + c ) Q \left( x , x_0 , x_1 \right) + L \left( x,x_0,x_1,f_0,f_1 \right) ,
\end{equation}
\noindent where $L(x,x_0,x_1,f_0,f_1) = f_0 +
(x-x_0)(f_1-f_0)/(x_1-x_0)$ is a piece-wise linear interpolation and
$Q(x,x_0,x_1) = (x-x_0) (x-x_1)$, respectively. In this way, we can
guarantee the continuity of the resulting expression of $\Topt
(\angle)$. Then, using least squares criterion, the set of constants
follows: $c_1=0.0647998$, $c_2=-0.386022$, $c_3=3.72945$,
$c_4=-9.38143$ and $c_5=7.06574$. Similarly, we propose to approximate
$\Kopt(\angle)$ as follows
\begin{equation}
\label{Kopt_eq}
\Kopt ( \angle ) \approx  
\left\{
\begin{array}{ll}
1  & \hspace{7mm} 0\hphantom{_1}  \le  \angle \le \angle_1 \\
G(\angle,c_{6\hphantom{0}},c_{7\hphantom{0}},c_{8\hphantom{0}},\angle_1, \angle_2 ,1,k_1) & \hspace{7mm} \angle_1 <\angle \le \angle_2 \\
G(\angle,c_{9\hphantom{0}},c_{10},c_{11},\angle_2, \angle_3,k_1,k_2) & \hspace{7mm} \angle_2 <\angle \le \angle_3 \\
G(\angle,c_{12},c_{13},c_{14},\angle_3, \pi/2,k_2,0) & \hspace{7mm} \angle_3 <\angle \le \pi/2 
\end{array}
\right. 
\end{equation}
\noindent where $\angle_2=\pi/3$, $\angle_3=(3/5)^2 \pi$, $k_1=0.73782212$ 
and $k_2=0.44660387$. Then, least square minimization leads to the
following values: $c_6=2403400$, $c_7=-5018490$, $c_8=2620140$,
$c_9=2945$, $c_{10}=-6665.76$, $c_{11}=3790.54$, $c_{12}=4.80513$,
$c_{13}=-16.9473$ and $c_{14}=15.0155$, respectively. The maximum
errors for $\Topt ( \angle )$ and $\Kopt ( \angle )$ are around $0.08
\%$ and $0.25 \%$, respectively.


\section{Construction of the matrix $|\Tcs \vcvectc^{-1} \Tcs\traspose |$}

\label{matrix_construction}

The proposed method (see Eqs.~\ref{AlgEigConv3} and~\ref{AlgEigDiff3})
relies on the construction of the following face-to-face matrix
\begin{equation}
|\Tcs \vcvectc^{-1} \Tcs\traspose | \in \real^{\vvlength \times \vvlength} ,
\end{equation}
\noindent where $\Tcs \in \real^{\vvlength \times \pvlength}$ is the
cell-to-face incidence matrix, which has two non-zero elements per
row: a $+1$ and a $-1$ corresponding to the cells adjacent to a face
{(see Eq.~\ref{Tsc_def})}, and $\vcvectc \in \real^{\pvlength
  \times \pvlength}$ is the diagonal matrix containing the
cell-centered volumes. Hence, taking the mesh with $4$~control volumes
and $8$~faces shown in Figure~\ref{mesh} (right), the $4^{th}$ row of
matrix $|\Tcs \vcvectc^{-1} \Tcs\traspose |$, \ie~corresponding to the
velocity $U_{4} = [ \velh ]_{4}$, reads
\begin{equation}
\label{abs_Tcs_iV_Tsc_def}
\left[ |\Tcs \vcvectc^{-1} \Tcs\traspose | \right]_{4} = \left( \frac{1}{V_{c2}} , 0, \frac{1}{V_{c1}} , \frac{1}{V_{c1}} + \frac{1}{V_{c2}} , 0 , \frac{1}{V_{c2}} , \frac{1}{V_{c1}} , 0 \right) ,
\end{equation}
\noindent where $V_{c1} = [ \vcvectc ]_{c1,c1}$ and $V_{c2} = [
  \vcvectc ]_{c2,c2}$ are the volumes of the two cells adjacent to the
face number~$4$. {For the sake of completeness, other relevant
  face-to-face matrices used in this paper, such as $-\Tcs
  \Tcs\traspose$, $| \Tcs \Tcs\traspose|$ and $\Tcs |\Tcs\traspose|$
  read
\begin{align}
\label{Tcs_Tsc_def}
\left[ - \hphantom{|}\Tcs \hphantom{|}\Tcs\traspose\hphantom{|} \right]_{4} &= \left( +1 , 0, +1 , -2 , 0 , -1 , -1 , 0 \right) , \\
\label{abs_Tcs_Tsc_def}
\left[ \hphantom{+} | \Tcs \hphantom{|}\Tcs\traspose | \right]_{4} &= \left( +1 , 0, +1 , +2 , 0 , +1 , +1 , 0 \right) , \\
\label{Tcs_absTsc_def}
\left[ \hphantom{+} \hphantom{|}\Tcs |\Tcs\traspose| \right]_{4} &= \left( -1 , 0, +1 , \hphantom{+}0 , 0 , -1 , +1 , 0 \right) .
\end{align}}
\noindent The rest of rows are constructed in the same manner.




\end{document}